\documentclass[11pt,twoside,a4paper]{report}
\usepackage[T1,T2A]{fontenc}
\usepackage[utf8]{inputenc}
\usepackage[german,russian,english,UKenglish]{babel}
\usepackage[hyperref=true,style=authoryear-comp, sorting=none, sortlocale=auto, language=auto, backend=bibtex]{biblatex}
\usepackage{amsmath,amssymb, amsthm}
\usepackage{mathtools}
\usepackage{units}

\usepackage{graphicx}
\graphicspath{{figures/}}

\usepackage[xindy]{imakeidx}
\makeindex[title=Index, program=xindy, intoc=true]

\usepackage{color}
\usepackage{enumitem}

\usepackage[font=small, labelfont=bf]{caption}

\usepackage[printonlyused, withpage]{acronym}

\usepackage{tikz}
\usetikzlibrary{positioning, calc}
\usepgflibrary{shapes.arrows}
\usetikzlibrary{arrows.meta}
\usetikzlibrary{decorations.markings}

\tikzset{
  vertex/.style={
    circle, inner sep =1.25pt, minimum size=3pt, fill=black, outer sep=0
  }
}

\usepackage{ctable}

\newcolumntype{C}{>{$}c<{$}}
\newcolumntype{L}{>{$}l<{$}}
\newcolumntype{R}{>{$}r<{$}}
\newcolumntype{M}{>{\centering\arraybackslash$}X<{$}}
\usepackage{dcolumn}
\newcolumntype{d}[1]{D{.}{.}{#1}}

\newcommand*{\mat}[1]{\boldsymbol{#1}}
\newcommand*{\coord}[1]{\mathbf{#1}}
\newcommand*{\vece}{\coord{e}}
\newcommand*{\veck}{\coord{k}}
\newcommand*{\vecr}{\coord{r}}
\newcommand*{\vecR}{\coord{R}}
\newcommand*{\vecs}{\coord{s}}

\newcommand*{\vecx}{\coord{x}}
\newcommand*{\vecX}{\coord{X}}
\newcommand*{\vecy}{\coord{y}}

\newcommand*{\binteg}[3]{\int^{\mathrlap{#3}}_{\mathrlap{#2}}\ud{#1}\,}
\newcommand*{\integ}[1]{\!\int\!\ud{#1}\:}
\newcommand*{\iinteg}[2]{\integ{#1}\!\!\integ{#2}}

\DeclareMathOperator{\diag}{diag}
\DeclareMathOperator{\Ei}{Ei}

\DeclareMathOperator{\KummerM}{M}
\DeclareMathOperator{\Laguerre}{L}

\DeclareMathOperator{\LegendreP}{P}
\DeclareMathOperator{\LegendreQ}{Q}
\DeclareMathOperator{\Order}{\mathcal{O}}

\DeclareMathOperator{\SpherHarm}{Y}
\DeclareMathOperator{\Trace}{Tr}

\DeclarePairedDelimiter{\abs}{\lvert}{\rvert}
\DeclarePairedDelimiter{\av}{\langle}{\rangle}
\DeclarePairedDelimiter{\norm}{\lVert}{\rVert}
\DeclarePairedDelimiterX\cbraket[2]{(}{)}{#1\delimsize\vert#2}
\DeclarePairedDelimiterX\cbraketx[2]{(}{)}{#1\delimsize\Vert#2}
\DeclarePairedDelimiterX\braket[2]{\langle}{\rangle}{#1\delimsize\vert#2}
\DeclarePairedDelimiterX\braketx[2]{\langle}{\rangle}{#1\delimsize\Vert#2}
\DeclarePairedDelimiterX\brakket[3]{\langle}{\rangle}{#1\delimsize\vert#2\delimsize\vert#3}
\DeclarePairedDelimiterX\brakketx[3]{\langle}{\rangle}{#1\delimsize\Vert#2\delimsize\Vert#3}
\DeclarePairedDelimiterX\ket[1]{\lvert}{\rangle}{#1}
\DeclarePairedDelimiterX\bra[1]{\langle}{\rvert}{#1}
\DeclarePairedDelimiterX\set[2]{\lbrace}{\rbrace}{#1 : #2}

\newcommand*{\opOneHF}[1]{#1{\gamma}^{\text{HF}}}
\newcommand*{\oneHF}{\opOneHF{}}

\newcommand*{\dens}{\rho}	

\newcommand*{\Complex}{\mathbb{C}}

\newcommand*{\du}{\partial}
\newcommand*{\eqspace}{\phantom{{} = {}}}
\newcommand*{\e}{\mathrm{e}}
\newcommand*{\half}{\frac{1}{2}}
\newcommand*{\Hone}{\mathcal{H}}
\newcommand*{\im}{\textrm{i}}
\newcommand*{\Integers}{\mathbb{Z}}
\newcommand*{\isDefinedAs}{\coloneqq}
\newcommand*{\nfrac}[2]{\nicefrac{#1}{#2}}
\newcommand*{\nhalf}{\nicefrac{1}{2}}
\newcommand*{\perm}{\wp}
\newcommand*{\Reals}{\mathbb{R}}
\newcommand*{\rref}{\vecr_{\text{ref}}}
\newcommand*{\sAdenifeDsi}{\eqqcolon}

\newcommand*{\ud}{\mathrm{d}}
\newcommand{\xmapsback}[1]{\xleftarrow{#1}\raisebox{.4pt}{\hspace{-5pt}$\shortmid$\hspace{3pt}}}

\theoremstyle{definition}
\newtheoremstyle{exercise}{\topsep}{\topsep}{}{}{\bfseries}{.}{.5em}{{\color{red}\thmname{#1}\thmnumber{ #2}\thmnote{ [#3]}}}
\theoremstyle{exercise}
\newtheorem{exercise}{Exercise}[chapter]
\newlist{subexercise}{enumerate}{2}
\setlist[subexercise]{label=\alph*), ref=\theexercise.\alph*}
\newtheorem{example}{Example}[chapter]

\theoremstyle{theorem}
\newtheorem*{theorem}{Theorem}

\allowdisplaybreaks[4]

\newbool{normSlatDet}
\booltrue{normSlatDet}
\newbool{MullikenIntegApprox}
\boolfalse{MullikenIntegApprox}

\definecolor{DarkBlue}{rgb}{0.2,0.2,0.5}
\definecolor{DarkGreen}{rgb}{0,0.4,0}
\definecolor{Indigo}{rgb}{0.51,0,0.51}

\title{Hartree--Fock \& \\ Density Functional Theory}
\author{Klaas Giesbertz}

\ifx\PipelineBuildVersion\undefined
  \newcommand{\FileVersion}{2022}
\else
  \newcommand{\FileVersion}{\PipelineBuildVersion}
\fi

\date{version \FileVersion}

\usepackage[pdftitle={\@title}, pdfauthor={\@author}, colorlinks=true, unicode, bookmarksnumbered=true, linkcolor=DarkBlue, citecolor=DarkGreen, urlcolor=Indigo]{hyperref}

\begin{document}

\maketitle

\tableofcontents


\chapter{Hartree--Fock}

\section{Introduction}

In quantum chemistry the quantum mechanical properties of molecules are studied. The full non-relativistic molecular Hamiltonian treats both the electrons and nuclei quantum mechanically, so reads
\begin{align}
\hat{H}^{\text{full}}
&= \sum_{\mu=1}^{\crampedclap{N_{\text{nuc}}}}\frac{-\hbar^2}{2M_{\mu}}\nabla^2_{\vecR_{\mu}} +
\frac{e^2}{4\pi\epsilon_0}\sum_{\mu=1}^{\crampedclap{N_{\text{nuc}}}}
\sum_{\nu = \mu + 1}^{\crampedclap{N_{\text{nuc}}}}
\frac{Z_{\mu}Z_{\nu}}{\abs{\vecR_{\mu} - \vecR_{\nu}}} - {} \\*
&\eqspace
\frac{e^2}{4\pi\epsilon_0}\sum_{\mu=1}^{\crampedclap{N_{\text{nuc}}}}
\sum_{i=1}^N\frac{Z_{\mu}}{\abs{\vecR_{\mu} - \vecr_i}} +
\frac{-\hbar^2}{2m_e}\sum_{i=1}^N\nabla^2_{\vecr_i} +
\frac{e^2}{4\pi\epsilon_0}\sum_{i=1}^N\sum_{j=1}^{i-1}\frac{1}{\abs{\vecr_i - \vecr_j}} . \notag 
\end{align}
As a convention, we will use latin letters to refer to electrons and greek letters to indicate nuclei. Further, we used the following quantities
\begin{itemize}
\item$N_{\text{nuc}}$ total number of nuclei,
\item$N$ total number of electrons,
\item$\hbar$ Planck's constant, $h$, divided by $2\pi$,
\item $e$ elementary charge,
\item $\epsilon_0$ vacuum permitivitty,
\item $M_{\mu}$ mass of nucleus $\mu$,
\item $Z_{\mu}$ atom number of atom $\mu$, i.e.\ the number of protons,
\item $\vecR_{\mu}$ position of nucleus $\mu$,
\item $m_e$ mass of an electron,
\item $\vecr_i$ position of electron $i$.
\end{itemize}
The solution to the Schrödinger equation of the full molecular Hamiltonian is the full molecular wavefunction, depending on all positions and spin variables of all nuclei and electrons
\begin{align}
\Psi_{\text{full}}\bigl(\vecX_1,\dotsc,\vecX_{N_{\text{nuc}}},\vecx_1,\dotsc,\vecx_N\bigr) ,
\end{align}
where we used combined space-spin coordinates
\begin{align}
\vecX_{\mu} &\isDefinedAs \vecR_{\mu}\Sigma_{\mu} &
&\text{and} &
\vecx_i &\isDefinedAs \vecr_i\sigma_i .
\end{align}
In these combined space-spin coordinates, $\Sigma_{\mu}$ and $\sigma_i$ denote the spin coordinates of nucleus $\mu$ and electron $i$ respectively.

The nuclei are much heavier than the electrons, so we expect the nuclear quantum effects to be small compared to the electronic ones. We will therefore focus on the electronic part of the Hamiltonian by taking the limit $M_{\mu} \to \infty$, which turns the nuclei into classical point charges. This is called the \emph{Born--Oppenheimer approximation}. In fact, Born and Oppenheimer showed more precisely under which assumptions this approximation is valid and how to include the nuclear effects in a second step \autocite{BornOppenheimer1927}. It has become clear that the nuclear part plays a crucial role in chemical reactions and alternative approaches have been put forward \autocite{AbediMaitraGross2010, KylanpaaRantala2011}. Nevertheless, we will focus on the electronic part in this course, as it already provides a significant challenge.

The (electronic) Hamiltonian central in this course will be
\begin{multline}\label{eq:elecHam}
\hat{H} = -\half\sum_{i=1}^N\nabla^2_{\vecr_i} +
\sum_{i=1}^N\underbrace{\sum_{\mu=1}^{\crampedclap{N_{\text{nuc}}}}
\frac{-Z_{\mu}}{\abs{\vecR_{\mu} - \vecr_i}}}_{= v(\vecr_i)} + {} \\
\sum_{1 \leq i < j \leq N}
\underbrace{\frac{1}{\abs{\vecr_i - \vecr_j}}}_{ = w(\vecr_i,\vecr_j)} +
\underbrace{\sum_{1\leq\mu<\nu\leq N_{\text{nuc}}}\frac{Z_{\mu}Z_{\nu}}{\abs{\vecR_{\mu} - \vecR_{\nu}}}}_{ = E_{\text{nuc}}} ,
\end{multline}
which is expressed now in \acf{a.u.}: $\hbar = e = m_e = 4\pi\epsilon_0 = 1$,
so
\begin{itemize}
\item length is measured in Bohr, $a_0 = \frac{4\pi\epsilon_0\hbar^2}{m_ee^2}$.
\item energy is measured in \unit{Hartree} $= \frac{m_e}{\hbar^2}\bigl(\frac{e^2}{4\pi\epsilon_0}\bigr)^2$, so the ground state energy of hydrogen is -\unit[\nhalf]{Hartree} = -\unit[13.6]{eV}.
\end{itemize}
Note that the interaction of the electrons with the nuclei has now become a simple local (=multiplicative) potential $v_{\text{nuc}}(\vecr_i)$ and only the interaction, $w(\vecr_i,\vecr_j)$, among the electrons is the remaining complicating term. The interaction between the nuclei has reduced to a constant $E_{\text{nuc}}$, which only yields a shift in the energy (gauge) and is typically only added at the end of the calculation as it does not affect the eigenfunctions of the Hamiltonian. The electronic wavefunction now only contains the space-spin coordinates of the electrons
\begin{align}
\Psi(\vecx_1,\dotsc,\vecx_N) ,
\end{align}
so the Born--Oppenheimer approximation provides a significant simplification.

Though we have simplified our problem significantly, it is still impossible to solve exactly in general. We therefore need to resort to approximate solutions. To gain a better understanding of these approximate solutions, we should first state more precisely in which function space our solution should reside.

\subsection{1-particle space}
We will first characterise the 1-particle space, since the $N$-particle space can be constructed from it in a relatively straightforward manner. As we want the wavefunction, $\psi(\vecr)$, to be a probability amplitude, its square needs to be integrable.
\mkbibfootnote{We actually also want the kinetic energy to be finite, so we can actually work in a subspace of $L^2$ in which we demand that the gradient does not become too large. This is the following Sobolev space
\begin{align*}
H^1(\Reals^3) \isDefinedAs
\set*{f \colon \Reals^3 \to \Complex}{\integ{\vecr}\bigl(\abs{f(\vecr)}^2 + \abs{\nabla f(\vecr)}^2\bigr) < \infty}
\end{align*}
To be more precise, this is the space for the so-called weak solutions of the Schrödinger equations. For strong solutions also the Laplacian should be bounded, so these solutions should be sought in the subspace $H^2(\Reals^3)$ \autocite{RuggenthalerPenzLeeuwen2015}.
}
The space of square integrable functions $\Reals^3 \to \Complex$ is named $L^2(\Reals^3)$, which is in more mathematical notation
\begin{align}
L^2(\Reals^3) \isDefinedAs
\set*{f \colon \Reals^3 \to \Complex}{\integ{\vecr}\abs{f(\vecr)}^2 < \infty} .
\end{align}
Note that this definition also works for different spaces by replacing $\Reals^3$ by the appropriate space, e.g.\ $\Reals$ for 1D and $[a,b]$ for a 1D box. An important feature of $L^2$ is that it is also a Hilbert space, i.e.\ one can define a proper inner product as
\begin{align}
\braket{f}{g} \isDefinedAs \integ{\vecr}f^*(\vecr)g(\vecr) .
\end{align}
This allows us to use the concept of orthogonality and to expand functions in a basis. More on that in just a moment.

As an electron is a spin-half particle, the value of its wavefunction also depends on the spin variable, which is a point in a 2-dimensional vector space. This 2-dimensional vector space is typically denoted by $\Integers_2$. The two basis elements of $\Integers_2$ are represented in many ways, e.g. $\{-1,1\}$, $\{-\nhalf,\nhalf\}$, $\{0,1\}$, $\{\alpha,\beta\}$ and $\{\uparrow,\downarrow\}$. So could also see $\psi(\vecx)$ as a vector containing two functions 
\( \bigl(\begin{smallmatrix}\psi_{\alpha}(\vecr) \\ \psi_{\beta}(\vecr) \end{smallmatrix}\bigr) \).
The only essential aspect is that it has the structure of a 2-dimensional vector space. The wavefunction of an electron, $\psi \colon \Reals^3 \times \Integers_2 \to \Complex$, needs to reside in the following Hilbert space
\begin{align}
\Hone \isDefinedAs L^2(\Reals^3) \otimes \Integers_2
= \set*{f \colon \Reals^3 \times \Integers_2 \to \Complex}{\integ{\vecx}\abs{f(\vecx)}^2 < \infty} .
\end{align}
The integral over the space-spin coordinate is a short-hand notation for integration of the space variable $\vecr$ and summation over the spin variable $\sigma$
\begin{align}
\integ{\vecx} \isDefinedAs \integ{\vecr}\sum_{\crampedclap{\sigma \in \Integers_2}}
= \integ{\vecr}\sum_{\crampedclap{\sigma \in \{-1,1\}}}
= \integ{\vecr}\sum_{\crampedclap{\sigma \in \{\uparrow,\downarrow\}}} ,
\end{align}
where we have explicitly expressed the basis elements of $\Integers_2$ for some possible choices in the last two equalities.

As we took the tensor product of two Hilbert spaces, $\Hone$ is also a Hilbert space with the inner product
\begin{align}
\braket{f}{g} \isDefinedAs \integ{\vecx}f^*(\vecx)g(\vecx)
= \integ{\vecr}\sum_{\crampedclap{\sigma \in \Integers_2}} f^*(\vecr\sigma)g(\vecr\sigma) .
\end{align}
Since $\Hone$ is a Hilbert space, any element of the Hilbert space (function) can be expanded in a basis
\begin{align}
\psi(\vecx) = \sum_{i=1}^{\infty}\psi_i\phi_i(\vecx) .
\end{align}
So a Hilbert space has the same structure as a vector space though (possibly) infinite dimensional.

Note that you can easily generalise these consideration to general spin-$s$ particles by replacing $\Integers_2$ by $\Integers_{2s+1}$.

\begin{exercise}\label{ex:FourierCoefs}
Show that $\psi_i = \braket{\phi_i}{\psi}$ if we assume that $\{\phi_i\}$ forms an orthonormal basis, i.e.\ $\braket{\phi_i}{\phi_j} = \delta_{ij}$.
\end{exercise}

\subsection{$N$-particle spaces}
To construct an $N$-particle Hilbert space, we simply glue $N$ 1-particle Hilbert spaces together by taking the tensor product
\begin{align}
\Hone^N \isDefinedAs \bigotimes_{i=1}^N\Hone = \underbrace{\Hone \otimes \dotsb \otimes \Hone}_{\text{$N$ times}} ,
\end{align}
which is similar to the construction of $\Reals^3 = \Reals \times \Reals \times \Reals$ out $\Reals$. We can readily construct a basis for $\Hone^N$ by considering all possible products of 1-body basis functions, so a general function $f \in \Hone^N$ can be expanded as
\begin{align}
f(\vecx_1,\dotsc,\vecx_N) = \sum_{\crampedclap{i_1,\dotsc,i_N}} f_{i_1\dotsc i_N}
\underbrace{\phi_{i_1}(\vecx_1)\phi_{i_2}(\vecx_2)\dotsb\phi_{i_N}(\vecx_N)}_%
{\sAdenifeDsi\Phi_{i_1i_2\dotsc i_N}(\vecx_1,\vecx_2,\dotsc,\vecx_N)} .
\end{align}
Compare this construction with the formation of higher dimensional monomials out of monomials in 1D. For example, for 3D monomials we have
\begin{align*}
&1\cdot1\cdot1, \\
&x_1\cdot1\cdot1, &
&1\cdot x_2\cdot1, &
&1\cdot1\cdot x_3 , \\
&x_1^2\cdot1\cdot1, &
&x_1\cdot x_2\cdot1, &
&x_1\cdot 1\cdot x_3, \\
&1\cdot x_2^2\cdot 1, &
&1\cdot x_2\cdot x_3, &
&1\cdot 1\cdot x_3^2 , \text{etc.}
\end{align*}
That is all we need for non-identical particles. For identical quantum particles, however, we need to do some more work, since the expectation value of any operator is not allowed to change upon permutation. The spin-statistics theorem \autocite{Fierz1939, Pauli1940} states that in 3D there are only two options. The wavefunction is either
\begin{description}[labelindent=1em]
\item[symmetric] which corresponds to integer spin particles, i.e.\ bosons, or
\item[anti-symmetric] which are half-integer particles, i.e.\ fermions.
\end{description}
As we should either have a fully symmetric or anti-symmetric wavefunction to describe bosons or fermion respectively, the only thing we need to do is to adapt the product basis to the required symmetry
\begin{align}\label{eq:symBasis}
\Phi^{\pm}_{i_1i_2\dotsc i_N}(\vecx_1,\dotsc,\vecx_N)
\isDefinedAs \ifbool{normSlatDet}{\frac{1}{\sqrt{N!}}}{}\begin{vmatrix}
\phi_{i_1}(\vecx_1)	&\phi_{i_2}(\vecx_1)	&\ldots	&\phi_{i_N}(\vecx_1)	\\
\phi_{i_1}(\vecx_2)	&\phi_{i_2}(\vecx_2)	&\ldots	&\phi_{i_N}(\vecx_2)	\\
\vdots			&\vdots			&\ddots	&\vdots			\\
\phi_{i_1}(\vecx_N)	&\phi_{i_2}(\vecx_N)	&\ldots	&\phi_{i_N}(\vecx_N)
\end{vmatrix}_{\mathrlap{\pm}} .
\end{align}
The minus sign means that we use the determinant (fermions) and the plus sign means that we use the permanent (bosons). The permanent is simply a determinant without the alternating signs. The basis states for the fermions are often referred to as Slater determinants \autocite{Slater1929}, though they were used earlier by Heisenberg \autocite{Heisenberg1926} and Dirac \autocite{Dirac1926b}.

\ifbool{normSlatDet}{The factor $(N!)^{-\nhalf}$ is the standard normalization factor for Slater determinants, though in the proper construction, this factor should be absent \autocite{StefanucciLeeuwen2013}.}{
As we now are working with (anti-)symmetric functions, we need to redefine the inner product as $\Psi(\vecx_1,\vecx_2)$ refers to the same points as $\Psi(\vecx_2,\vecx_1)$. As we should only take the unique part into account, the inner product now becomes
\begin{align}\label{eq:antiSymInnerProd}
\braket{f}{g} \isDefinedAs&
\integ{\vecx_1}\!\!\binteg{\vecx_2}{\vecx_2 > \vecx_1}{}\dotsi\binteg{\vecx_N}{\vecx_N > \vecx_{N-1}}{}f^*(\vecx_1,\vecx_2,\dotsc,\vecx_N)g(\vecx_1,\vecx_2,\dotsc,\vecx_N) \\*
{}=& \frac{1}{N!}\integ{\vecx_1}\!\!\integ{\vecx_2}\dotsi\integ{\vecx_N}
f^*(\vecx_1,\vecx_2,\dotsc,\vecx_N)g(\vecx_1,\vecx_2,\dotsc,\vecx_N) . \notag
\end{align}
In virtually all textbooks and papers, the inner product is not taken over the unique sector, so does not have the $(N!)^{-1}$ factor. To compensate for this error, the Slater determinants gets a $1/\sqrt{N!}$ normalisation factor. A more extensive explanation why the inner product defined in~\eqref{eq:antiSymInnerProd} is the correct one for the probabilistic interpretation of the square of the wavefunction can be found in~\autocite{StefanucciLeeuwen2013}.
}

\begin{exercise}
Show that the symmetrised basis functions in~\eqref{eq:symBasis} are orthonormal, if the one-particle functions $\phi_i(\vecx)$ are also orthonormal.
\end{exercise}

\section{Full~\acf*{CI}}
Combined with the variational principle, the basis expansion quickly leads to the following idea to build approximate wavefunctions,
\mkbibfootnote{The expansion of a solution in a finite number of basis functions, including the convergence considerations of the expansion are known in mathematics as the Galerkin approximation \autocite{Galerkin1915} and was introduced by Walther Ritz in 1908 to whom Galerkin refers.}
which is known in the quantum chemistry community as full \ac{CI}.
\begin{enumerate}[align=left, label=\textbf{Step \arabic*})]
\item Select a finite number of one-electron basis functions which you deem important.

\item Construct the corresponding anti-symmetric $N$-electron basis (Slater determinants), $\Phi_I$.

\item Use the variational principle to optimise the expansion parameters $c_I$ in
\begin{align}
\Psi(\vecx_1,\dotsc,\vecx_N) = \sum_Ic_I\Phi_I(\vecx_1,\dotsc,\vecx_N),
\end{align}
where $I \isDefinedAs i_1,i_2,\dotsc,i_N$ with $i_1 < i_2 < \dotsb < i_N$.
\mkbibfootnote{We could equally well have taken $i_1 > i_2 > \dotsb > i_N$, or a more complicated scheme, as long as the set $\{\Phi_I\}$ is linearly independent.} That is, we demand that the gradient of the energy with respect to the expansion coefficients, $c_I$, vanishes which leads to the secular equations
\begin{align}\label{eq:fullCI}
\sum_J\bigl(H_{IJ} - E\,S_{IJ}\bigr)c_J = 0 .
\end{align} 
\end{enumerate}

\begin{exercise}
Derive the secular full \ac{CI} equations.
\end{exercise}

\begin{exercise}
How many $N$-body determinants can be constructed from $M$ one-electron basis functions?
\textbf{Challenge:} How many permanents can be constructed?
\end{exercise}

\noindent
The number of determinants grows very quickly in full \ac{CI} with the number of one-electron basis functions and number of electrons. There are many selection schemes to use only a subset of the determinants, which lead to a large variety of \ac{CI} methods. Now, let us to try to get away with the a single determinant: the `best' one. With the variational principle at our disposal, we will define the `best' Slater determinant as the one that yields the lowest energy. Only the orbitals can be varied in the Slater determinant, so we will need to optimise the energy with respect to the orbitals. So we would like to have an explicit expression of the energy in terms of the orbitals, $E_{\text{HF}}[\{\phi_i\}]$. As $E_{\text{HF}}$ is now a function of functions, we call it a functional.

This approximation to retain only one Slater determinant is known as \acf{HF}. It was originally proposed by Hartree \autocite{Hartree1928}, but he only took the Pauli exclusion principle into account and not the full anti-symmetry of the wavefunction. This was pointed out independently by Slater \autocite{Slater1929} and by Fock \autocite{Fock1930}. The anti-symmetry leads to an important additional term in the energy expression, called the exchange energy.

\section{Slater--Condon rules}
The Slater--Condon rules \autocite{Slater1929, Condon1930} which deal with the expectation values of (possibly different) Slater determinants (both $N$-body)
\begin{align}
\brakket{\Phi_I}{\hat{O}}{\Phi_J} .
\end{align}
The Slater--Condon rules only deal with the case where both determinants are constructed from the same orthonormal basis. They can readily be extended to a non-orthonormal bases \autocite{Lowdin1955}, but we will not need them and are therefore out of the scope of this course.

First we need to introduce the concept of an $n$-body operator. An $n$-body operator is an operator which acts only on $n$ different particles simultaneously
\begin{subequations}
\begin{align}
\hat{O}_0	&= o_0 = \text{constant} , \\*
\hat{O}_1	&= \sum_i\hat{o}_i = \sum_i\hat{o}(\vecx_i) , \\*
\hat{O}_2	&= \sum_{i < j}\hat{o}_{ij} = \half\sum_{i\neq j}\hat{o}(\vecx_i,\vecx_j) , \\*
		\shortvdotswithin{=}
\hat{O}_n	&= \sum_{\crampedclap{i_1 < \dotsb < i_n}}\hat{o}_{i_1 \dotsb i_n}
= \frac{1}{n!}\sum_{\crampedclap{i_1 \neq \dotsb \neq i_n}}\hat{o}(\vecx_{i_1},\dotsc,\vecx_{i_n}) .
\end{align}
\end{subequations}
The interaction between the nuclei, $E_{\text{nuc}}$, in the electronic Hamiltonian~\eqref{eq:elecHam} is an example of a 0-body operator. The kinetic energy of the electrons, $\hat{T} = -\half\sum_i\nabla^2_i$, and the interaction of the electrons with the nuclei, $\hat{V}_{\text{nuc}} =\sum_iv_{\text{nuc}}(\vecr_i)$, are examples of 1-body operators. The interaction between the electrons, $\hat{W} = \half\sum_{i \neq j}w(\vecr_i,\vecr_j)$ is clearly a 2-body operator. An example of a mixed operator is the total spin operator
\begin{align}
\hat{S}^2 \isDefinedAs \mat{\hat{S}} \cdot \mat{\hat{S}}
= \sum_{i \neq j} \mat{\hat{S}}_i \cdot \mat{\hat{S}}_j + \sum_i\hat{S}^2_i.
\end{align}
The first part on the right-hand side is two-body and the last is one-body.
The advantage of the notion of $n$-body operators is that due to the indistinguishability of quantum particles, that we only need to calculate the expectation value for one set of them and multiply by the number of $n$-body sets to get the full expectation value. In words this sounds rather cryptic, but in formulae it simply means that
\begin{subequations}\label{eq:nBodyReduced}
\begin{align}
\brakket{\Phi}{\hat{O}_0}{\Psi} &= \brakket{\Phi}{o_0}{\Psi} = o_0\braket{\Phi}{\Psi} , \\*
\label{eq:one-body}
\brakket{\Phi}{\hat{O}_1}{\Psi} &= N\brakket{\Phi}{\hat{o}_1}{\Psi} , \\*
\label{eq:two-body}
\brakket{\Phi}{\hat{O}_2}{\Psi} &= \frac{N(N-1)}{2}\brakket{\Phi}{\hat{o}_{12}}{\Psi} , \\*
	&\vdotswithin{=} \notag \\
\brakket{\Phi}{\hat{O}_n}{\Psi} &= \binom{N}{n}\brakket{\Phi}{\hat{o}_{1\dotsc n}}{\Psi} .
\end{align}
\end{subequations}
We assume for the many-body operators that they are symmetric, as the particles are indistinguishable, e.g.\ for the two-body operator we need that $\hat{o}_{ij} = \hat{o}_{ji}$.

\begin{exercise}
Why are the elements which have some of the indices equal excluded from a general $n$-body operator. More concrete, why is the term $i = j$ not included in a two-body operator?
\end{exercise}

\begin{exercise}
Derive the relations in~\eqref{eq:nBodyReduced}. Start with the 0-body operator, then consider the 1-body operator and the 2-body operator. Finally argue that the formula for a general $n$-body operator is correct.
\end{exercise}

\noindent
The Slater--Condon rules can be derived by simply working out the expectation values. For 0-body (constant) and 1-body operators this is relatively straightforward, but more more-body operators this becomes quite a dirty business, in particular to keep track of all the phase factors when dealing with fermions, i.e.\ determinants. Therefore, we introduce a graphical representation of a determinant (which also works for permanents of course)~\autocite{StefanucciLeeuwen2013}. A permanent\slash{}determinant of an $n\times n$ matrix, $\mat{A}$, is defined as
\begin{align}
\abs{\mat{A}}_{\pm} \isDefinedAs \sum_{\perm} (\pm)^{\perm}\prod_{i=1}^nA_{i\,\perm(i)} ,
\end{align}
where the ``$+$'' sign refers to the permanent and the ``$-$'' sign to the determinant. The symbol $\perm$ denotes a permutation of the indices, so $\perm\bigl(1,2,\dotsc,n\bigr) = \bigl(\perm(1),\perm(2),\dotsc,\perm(n) = \bigl(1',2',\dotsc,n'\bigr)$. For the permanent, we have $(+)^{\perm} = 1$ for any permutation. The notation $(-)^{\perm}$ denotes the sign of the permutation, so $(-)^{\perm} = 1$ if an even number of pairs needed to interchanged to achieve the permutation and $(-)^{\perm} = -1$ if an odd number of swaps were needed for the permutation. For example, consider the permutation $\perm(1,2,3,4,5) = (2,1,5,3,4)$. This permutation can be built up by swapping first positions 4 and 5, which we will denote as $(4,5)$. Next we swap positions 3 and 4, so $(3,4)$ and finally we perform $(1,2)$. The complete permutation operation can therefore be expressed as $\perm = (1,2)(3,4)(4,5)$. So 3 swaps are needed which tells us that the sign of the permutation is $-1$. It is helpful to make a graphical representation of the permutation. The right column denotes the indices $i$ and the left column their positions.
\begin{multline}
(1,2)(3,4)(4,5)\vcenter{\hbox{%
\begin{tikzpicture}
  \node[vertex,label=left:$1$] (i1) {};
  \foreach \n/\nmin in {2/1,3/2,4/3,5/4}
    \node[vertex, below=0.5 of i\nmin, label=left:$\n$] (i\n) {};
  \foreach \n in {1,...,5} \node[vertex, right=2 of i\n, label=right:$\n$] (j\n) {};
  \foreach \n in {1,...,5} \draw (i\n) -- (j\n);
\end{tikzpicture}
}}
= (1,2)(3,4)\vcenter{\hbox{%
\begin{tikzpicture}
  \node[vertex,label=left:$1$] (i1) {};
  \foreach \n/\nmin in {2/1,3/2,4/3,5/4}
    \node[vertex, below=0.5 of i\nmin, label=left:$\n$] (i\n) {};
  \foreach \n in {1,...,5} \node[vertex, right=2 of i\n, label=right:$\n$] (j\n) {};
  \foreach \n in {1,2,3} \draw (i\n) -- (j\n);
  \draw (i4) to[out=0,in=180] (j5);
  \draw (i5) to[out=0,in=180] (j4);
\end{tikzpicture}
}} \\
= (1,2)\vcenter{\hbox{%
\begin{tikzpicture}
  \node[vertex,label=left:$1$] (i1) {};
  \foreach \n/\nmin in {2/1,3/2,4/3,5/4}
    \node[vertex, below=0.5 of i\nmin, label=left:$\n$] (i\n) {};
  \foreach \n in {1,...,5} \node[vertex, right=2 of i\n, label=right:$\n$] (j\n) {};
  \foreach \n in {1,2} \draw (i\n) -- (j\n);
  \draw (i3) to[out=0,in=180] (j5);
  \draw (i4) to[out=0,in=180] (j3);
  \draw (i5) to[out=0,in=180] (j4);
\end{tikzpicture}
}}
= \vcenter{\hbox{%
\begin{tikzpicture}
  \node[vertex,label=left:$1$] (i1) {};
  \foreach \n/\nmin in {2/1,3/2,4/3,5/4}
    \node[vertex, below=0.5 of i\nmin, label=left:$\n$] (i\n) {};
  \foreach \n in {1,...,5} \node[vertex, right=2 of i\n, label=right:$\n$] (j\n) {};
  \draw (i1) to[out=0,in=180] (j2);
  \draw (i2) to[out=0,in=180] (j1);
  \draw (i3) to[out=0,in=180] (j5);
  \draw (i4) to[out=0,in=180] (j3);
  \draw (i5) to[out=0,in=180] (j4);
\end{tikzpicture}
}} .
\end{multline}
You might notice that the number of crossings of the lines, $n_c$, is exactly equal to the number of permutations we made. This is not always the case, but it is easy to convince oneself that the parity is always the same. Consider the right-hand side of a permutation graph and interchange two vertices $i$ and $j$
\begin{align}
\vcenter{\hbox{%
\begin{tikzpicture}
  \coordinate (i2) {};
  \foreach \n/\nmin in {3/2,4/3,5/4}
    \coordinate[below=0.6 of i\nmin] (i\n) {};
  \coordinate[above=1.2 of i2] (i1) {};
  \coordinate[below=1 of i5] (i6) {};
  \foreach \n in {2,...,5} \coordinate[right=1 of i\n] (j\n) {};
  \coordinate[above=0.8 of j2] (j1) {};
  \coordinate[below=0.8 of j5] (j6) {};
  \foreach \n in {3,4} \node[vertex, right=2 of j\n] (k\n) {};
  \node [vertex, right=2 of j2, label=right:$i$] (k2) {};
  \node [vertex, right=2 of j5, label=right:$j$] (k5) {};
  \node[vertex, above=0.6 of k2.center, anchor=center] (k1) {};
  \node[vertex, below=0.5 of k5.center, anchor=center] (k6) {};
  \foreach \n in {1,...,5} \draw[dashed, dash phase=10pt] (i\n) -- (j\n);
  \draw[dashed] (i6) to[out=-30,in=210] (j6);
  \draw (j1) to[out=-14.036, in=180] (k1);
  \foreach \n in {2,...,5} \draw (j\n) -- (k\n);
  \draw (j6) to[out=30, in=180] (k6);
\end{tikzpicture}
}}
\xrightarrow{(ij)}\quad
\vcenter{\hbox{%
\begin{tikzpicture}
  \coordinate (i2) {};
  \foreach \n/\nmin in {3/2,4/3,5/4}
    \coordinate[below=0.6 of i\nmin] (i\n) {};
  \coordinate[above=1.2 of i2] (i1) {};
  \coordinate[below=1 of i5] (i6) {};
  \foreach \n in {2,...,5} \coordinate[right=1 of i\n] (j\n) {};
  \coordinate[above=0.8 of j2] (j1) {};
  \coordinate[below=0.8 of j5] (j6) {};
  \foreach \n in {3,4} \node[vertex, right=2 of j\n] (k\n) {};
  \node [vertex, right=2 of j2, label=right:$j$] (k2) {};
  \node [vertex, right=2 of j5, label=right:$i$] (k5) {};
  \node[vertex, above=0.6 of k2.center, anchor=center] (k1) {};
  \node[vertex, below=0.5 of k5.center, anchor=center] (k6) {};
  \foreach \n in {1,...,5} \draw[dashed, dash phase=10pt] (i\n) -- (j\n);
  \draw[dashed] (i6) to[out=-30,in=210] (j6);
  \draw (j1) to[out=-14.036, in=180] (k1);
  \foreach \n in {3,4} \draw (j\n) -- (k\n);
  \draw (j2) to[out=0, in=180] (k5);
  \draw (j5) to[out=0, in=180] (k2);
  \draw (j6) to[out=30, in=180] (k6);
\end{tikzpicture}
}} .
\end{align}
You see that the interchanges of nodes $i$ and $j$ leads to one additional crossing plus an even number of crossings, since both lines attached to $i$ and $j$ need to cross any other line. Therefore, we find the important rule
\begin{align}
(\pm)^{n_c} = (\pm)^{\perm} .
\end{align}
As an example, the permanent\slash{}determinant of a $2\times2$ matrix can be worked out with the help of these graphs as
\begin{align}
\abs{A}_{\pm}
&= \vcenter{\hbox{%
\begin{tikzpicture}
  \node[vertex,label=left:$1$] (i1) {};
  \node[vertex, below=1 of i1, label=left:$2$] (i2) {};
  \foreach \n in {1,2} \node[vertex, right=2 of i\n, label=right:$\n$] (j\n) {};
  \draw (i1) to node[above] {$A_{11}$} (j1);
  \draw (i2) to node[below] {$A_{22}$} (j2);
\end{tikzpicture}
}} + \vcenter{\hbox{%
\begin{tikzpicture}
  \node[vertex,label=left:$1$] (i1) {};
  \node[vertex, below=1 of i1, label=left:$2$] (i2) {};
  \foreach \n in {1,2} \node[vertex, right=2 of i\n, label=right:$\n$] (j\n) {};
  \draw (i1) to[out=0, in=180] node[near start, above] {$A_{11}$} (j2);
  \draw (i2) to[out=0, in=180] node[near start, below] {$A_{22}$} (j1);
\end{tikzpicture}
}} \notag \\
&= A_{11}A_{22} \pm A_{12}A_{21} .
\end{align}

\begin{exercise}
Use the graphical representation to work out the permanent\slash{}determinant for a general $3\times3$ matrix.
\end{exercise}

\noindent
Now we have the necessary ingredients to work out the Slater--Condon rules for \emph{determinants}. For permanents the rules are more complicated, since there is no anti-symmetry which forbids orbitals to be occupied multiple times. We will therefore skip the derivation of the bosonic Slater--Condon rules, as we aim to deal with electrons (fermions) in this course.

\subsection{0-body operators}
The 0-body operators simply reduce to the calculation of the overlap $\braket{\Phi_I}{\Phi_J}$, as the constant can be pulled out of the integration. Let us first consider the term of the determinant without any permutations
\begin{multline}
\integ{\vecx_1}\!\!\integ{\vecx_2}\dotsi\integ{\vecx_N}\phi_{i_1}^*(\vecx_1)\phi_{i_2}^*(\vecx_2)\dotsb\phi_{i_N}^*(\vecx_N)\phi_{j_1}(\vecx_1)\phi_{j_2}(\vecx_2)\dotsb\phi_{j_N}(\vecx_N) \\
= \vcenter{\hbox{%
\begin{tikzpicture}
  \node[vertex,label=left:$i_1$] (i1) {};
  \foreach \n/\nmin in {2/1,3/2,4/3}
    \node[vertex, below=0.7 of i\nmin, label=left:$i_{\n}$] (i\n) {};
  \foreach \n in {1,...,4} \node[vertex, right=1.5 of i\n, label=below:$\vecx_{\n}$] (x\n) {};
  \foreach \n in {1,...,4} \node[vertex, right=1.5 of x\n, label=right:$j_{\n}$] (j\n) {};
  \node[vertex, below=of i4, label=left:$i_N$] (i5) {};
  \node[vertex, below=of x4, label=below:$\vecx_N$] (x5) {};
  \node[vertex, below=of j4, label=right:$j_N$] (j5) {};
  \foreach \n in {1,...,5} \draw (i\n) -- (j\n);
  \foreach \r in {i,j} \node[anchor=center] at ($ (\r4.center) ! 0.4 ! (\r5.center) $){$\vdots$};
  \node[above=0.15 of x1] (int) {$\displaystyle\integ{\vecx}\phi^*_{i_k}(\vecx_k)\phi_{j_k}(\vecx_k)$};
  \draw[<-] (x1) -- +(0,0.45);
\end{tikzpicture}
}}
= \vcenter{\hbox{%
\begin{tikzpicture}
  \node[vertex,label=left:$i_1$] (i1) {};
  \foreach \n/\nmin in {2/1,3/2,4/3}
    \node[vertex, below=0.7 of i\nmin, label=left:$i_{\n}$] (i\n) {};
  \foreach \n in {1,...,4}
    \node[vertex, right=1.5 of i\n, label=above:$\braket{\phi_{i_\n}}{\phi_{j_\n}}$] (x\n) {};
  \foreach \n in {1,...,4} \node[vertex, right=1.5 of x\n, label=right:$j_{\n}$] (j\n) {};
  \node[vertex, below=of i4, label=left:$i_N$] (i5) {};
  \node[vertex, below=of x4, label=above:$\braket{\phi_{i_N}}{\phi_{j_N}}$] (x5) {};
  \node[vertex, below=of j4, label=right:$j_N$] (j5) {};
  \foreach \n in {1,...,5} \draw (i\n) -- (j\n);
  \foreach \r in {i,j} \node[anchor=center] at ($ (\r4.center) ! 0.4 ! (\r5.center) $){$\vdots$};
\end{tikzpicture}
}} \\
= \vcenter{\hbox{%
\begin{tikzpicture}
  \node[vertex,label=left:$i_1$] (i1) {};
  \foreach \n/\nmin in {2/1,3/2,4/3}
    \node[vertex, below=0.7 of i\nmin, label=left:$i_{\n}$] (i\n) {};
  \foreach \n in {1,...,4} \node[vertex, right=1.5 of i\n, label=above:$\delta_{i_{\n}j_{\n}}$] (x\n) {};
  \foreach \n in {1,...,4} \node[vertex, right=1.5 of x\n, label=right:$j_{\n}$] (j\n) {};
  \node[vertex, below=of i4, label=left:$i_N$] (i5) {};
  \node[vertex, below=of x4, label=above:$\delta_{i_Nj_N}$] (x5) {};
  \node[vertex, below=of j4, label=right:$j_N$] (j5) {};
  \foreach \n in {1,...,5} \draw (i\n) -- (j\n);
  \foreach \r in {i,j} \node[anchor=center] at ($ (\r4.center) ! 0.4 ! (\r5.center) $){$\vdots$};
\end{tikzpicture}
}} = \delta_{i_1j_1}\delta_{i_2j_2}\dotsb\delta_{i_Nj_N} \sAdenifeDsi \delta_{IJ} ,
\end{multline}
where we used the orthonormality of the orbitals in the last step, $\braket{\phi_i}{\phi_j} = \delta_{ij}$.
Now consider two different permutations of the indices in both determinants
\begin{align}
\vcenter{\hbox{%
\begin{tikzpicture}
  \node[vertex,label=left:$i_1$] (i1) {};
  \foreach \n/\nmin in {2/1,3/2,4/3}
    \node[vertex, below=0.7 of i\nmin, label=left:$i_{\n}$] (i\n) {};
  \node[vertex, right=1.5 of i1, label=above:$\delta_{i_2j_4}$] (x1) {};
  \node[vertex, right=1.5 of i2, label=above:$\delta_{i_1j_1}$] (x2) {};
  \node[vertex, right=1.5 of i3, label=above:$\delta_{i_3j_2}$] (x3) {};
  \node[vertex, right=1.5 of i4, label=above:$\delta_{i_4j_3}$] (x4) {};
  \foreach \n in {1,...,4} \node[vertex, right=1.5 of x\n, label=right:$j_{\n}$] (j\n) {};
  \node[vertex, below=of i4, label=left:$i_N$] (i5) {};
  \node[vertex, below=of x4, label=above:$\delta_{i_Nj_N}$] (x5) {};
  \node[vertex, below=of j4, label=right:$j_N$] (j5) {};
  \draw (i1) to[out=0, in=180] (x2) to[out=0, in=180] (j1);
  \draw (i2) to[out=0, in=180] (x1) to[out=0, in=180] (j4);
  \draw (i3) -- (x3) to[out=0, in=180] (j2);
  \draw (i4) -- (x4) to[out=0, in=180] (j3);
  \draw (i5) -- (j5);
  \foreach \r in {i,j} \node[anchor=center] at ($ (\r4.center) ! 0.4 ! (\r5.center) $){$\vdots$};
\end{tikzpicture}
}} = 0 ,
\end{align}
as we work with indices in increasing order. For example, $i_2$ can never be equal to $j_4$, since combined with ordering this implies that $i_3 > i_2 = j_4 > j_2$. We also need $i_3 = j_2$ for the integral to be non-zero, which cannot be.

As the indices are ordered, we need the $I = J$ and only when we have the same permutation in both determinants, $\Phi_I$ and $\Phi_J$, we get a contribution to the integral. For example
\begin{align}
\vcenter{\hbox{%
\begin{tikzpicture}
  \node[vertex,label=left:$i_1$] (i1) {};
  \foreach \n/\nmin in {2/1,3/2,4/3}
    \node[vertex, below=0.7 of i\nmin, label=left:$i_{\n}$] (i\n) {};
  \node[vertex, below=1 of i4, label=left:$i_N$] (iN) {}; 
  \foreach \n/\k in {1/2,2/4,3/3,4/1,N/N}
    \node[vertex, right=1.5 of i\k, label=above:$\delta_{i_{\n}j_{\n}}$] (x\k) {};
  \foreach \n in {1,2,3,4,N} \node[vertex, right=1.5 of x\n, label=right:$j_{\n}$] (j\n) {};
  \foreach \n/\k in {1/2,2/4,3/3,4/1,N/N} \draw (i\n) to[out=0, in=180] (x\k) to[out=0, in=180] (j\n);
  \foreach \r in {i,j} \node[anchor=center] at ($ (\r4.center) ! 0.4 ! (\r N.center) $){$\vdots$};
\end{tikzpicture}
}} = \vcenter{\hbox{%
\begin{tikzpicture}
  \node[vertex,label=left:$i_1$] (i1) {};
  \foreach \n/\nmin in {2/1,3/2,4/3}
    \node[vertex, below=0.7 of i\nmin, label=left:$i_{\n}$] (i\n) {};
  \node[vertex, below=1 of i4, label=left:$i_N$] (iN) {}; 
  \foreach \n/\k in {1/3,2/1,3/4,4/N,N/2}
    \node[vertex, right=1.5 of i\k, label=above:$\delta_{i_{\n}j_{\n}}$] (x\k) {};
  \foreach \n in {1,2,3,4,N} \node[vertex, right=1.5 of x\n, label=right:$j_{\n}$] (j\n) {};
  \foreach \n/\k in {1/3,2/1,3/4,4/N,N/2} \draw (i\n) to[out=0, in=180] (x\k) to[out=0, in=180] (j\n);
  \foreach \r in {i,j} \node[anchor=center] at ($ (\r4.center) ! 0.4 ! (\r N.center) $){$\vdots$};
\end{tikzpicture}
}} \!\!\! = \delta_{IJ} . 
\end{align}
The only remaining question is how many of these terms we have. This is simply the number of terms generated by a single Slater determinant, $N!$, as the term in the other determinant needs to correspond to exactly the same permutation. Hence, we find that the Slater determinants are orthonormal, if we use an orthonormal orbital basis
\begin{align}
\braket{\Phi_I}{\Phi_J}
= \delta_{IJ} .
\end{align}

\subsection{1-body operators}
Now we proceed with the one-body operators. The difference compared to the 0-body operators is that the integration over the first coordinate, $\vecx_1$, now contains an operator, $\brakket{\phi_{i_k}}{\hat{o}_1}{\phi_{i_l}}$, so the orbitals it connect do not need to be equal anymore for a finite expectation value. All the other orbitals still need to be equal, so there are two options to consider: 1) one orbital is different in the determinants, or 2) all orbitals are the same.

Let us first consider the case that the orbitals in both determinants are the same, i.e.\ $I = J$. If we only consider the terms where the operator works on $i_1$, we find
\begin{align}
\vcenter{\hbox{%
\begin{tikzpicture}
  \node[vertex,label=left:$i_1$] (i1) {};
  \foreach \n/\nmin in {2/1,3/2,4/3}
    \node[vertex, below=0.7 of i\nmin, label=left:$i_{\n}$] (i\n) {};
  \node[vertex, below=1 of i4, label=left:$i_N$] (iN) {};
  \node[vertex, right=1.5 of i1, label=above:$\brakket{\phi_{i_1}}{\hat{o}_1}{\phi_{i_1}}$] (x1) {};
  \foreach \n/\k in {2,3,4,N}
    \node[vertex, right=1.5 of i\k, label=above:$\delta_{i_{\n}i_{\n}}$] (x\k) {};
  \foreach \n in {1,2,3,4,N} \node[vertex, right=1.5 of x\n, label=right:$i_{\n}$] (j\n) {};
  \foreach \n/\k in {1,2,3,4,N} \draw (i\n) to[out=0, in=180] (x\k) to[out=0, in=180] (j\n);
  \foreach \r in {i,j} \node[anchor=center] at ($ (\r4.center) ! 0.4 ! (\r N.center) $){$\vdots$};
\end{tikzpicture}
}}
&= \vcenter{\hbox{%
\begin{tikzpicture}
  \node[vertex,label=left:$i_1$] (i1) {};
  \foreach \n/\nmin in {2/1,3/2,4/3}
    \node[vertex, below=0.7 of i\nmin, label=left:$i_{\n}$] (i\n) {};
  \node[vertex, below=1 of i4, label=left:$i_N$] (iN) {};
  \node[vertex, right=1.5 of i1, label=above:$\brakket{\phi_{i_1}}{\hat{o}_1}{\phi_{i_1}}$] (x1) {};
  \foreach \n/\k in {2/3,3/4,4/2,N}
    \node[vertex, right=1.5 of i\k, label=above:$\delta_{i_{\n}i_{\n}}$] (x\k) {};
  \foreach \n in {1,2,3,4,N} \node[vertex, right=1.5 of x\n, label=right:$i_{\n}$] (j\n) {};
  \foreach \n/\k in {1,2/3,3/4,4/2,N} \draw (i\n) to[out=0, in=180] (x\k) to[out=0, in=180] (j\n);
  \foreach \r in {i,j} \node[anchor=center] at ($ (\r4.center) ! 0.4 ! (\r N.center) $){$\vdots$};
\end{tikzpicture}
}} \notag \\
&= \brakket{\phi_{i_1}}{\hat{o}_1}{\phi_{i_1}} .
\end{align}
As we keep the orbital $\phi_{i_1}$ fixed, we can only make permutations of the remaining $i_2,\dotsc,i_N$ elements, so we have $(N-1)!$ of those elements. These are the only graphs which lead to a contribution of the form $\brakket{\phi_{i_1}}{\hat{o}_1}{\phi_{i_1}}$. If we would choose different permutations for both determinants, the contribution of the graph is zero.

Here we have singled out the orbital $\phi_{i_1}$, but the same argument works for any orbital $\phi_{i_k}$ as
\begin{align}
\vcenter{\hbox{%
\begin{tikzpicture}
  \node[vertex,label=left:$i_1$] (i1) {};
  \node[vertex, below=0.7 of i1, label=left:$i_2$] (i2) {};
  \node[vertex, below=0.7 of i2, label=left:$i_3$] (i3) {};
  \node[vertex, below=1 of i3, label=left:$i_k$] (ik) {};
  \node[vertex, below=1 of ik, label=left:$i_N$] (iN) {};
  \node[vertex, right=1.5 of i1, label=above:$\brakket{\phi_{i_k}}{\hat{o}_1}{\phi_{i_k}}$] (x1) {};
  \node[vertex, right=1.5 of i2, label=above:$\delta_{i_2i_2}$] (x2) {};
  \node[vertex, right=1.5 of i3, label=above:$\delta_{i_1i_1}$] (x3) {};
  \node[vertex, right=1.5 of ik] (xk) {};
  \node[vertex, right=1.5 of iN] (xN) {};
  \foreach \n in {1,2,3,k,N} \node[vertex, right=1.5 of x\n, label=right:$i_{\n}$] (j\n) {};
  \draw (i1) to[out=0, in=180] (x3) to[out=0, in=180] (j1);
  \draw (i2) -- (j2);
  \draw (i3) to[out=0, in=135] +(0.5,-0.25); \draw (j3)  to[in=45, out=180] +(-0.5,-0.25);
  \draw (ik) to[out=0, in=180] (x1) to[out=0, in=180] (jk);
  \draw (iN) to[out=0, in=225] +(0.5,0.25); \draw (jN) to[in=-45, out=180] +(-0.5,0.25);
  \draw (xk) to[out=0, in=135] +(0.5,-0.25); \draw (xk) to[in=45, out=180] +(-0.5,-0.25);
  \draw (xN) to[out=0, in=225] +(0.5,0.25); \draw (xN) to[in=-45, out=180] +(-0.5,0.25);
  \foreach \r in {i,j} \node[anchor=center] at ($ (\r3.center) ! 0.4 ! (\r k.center) $){$\vdots$};
  \foreach \r in {i,j} \node[anchor=center] at ($ (\r k.center) ! 0.4 ! (\r N.center) $){$\vdots$};
\end{tikzpicture}
}}
&= \vcenter{\hbox{%
\begin{tikzpicture}
  \node[vertex,label=left:$i_k$] (ik) {};
  \node[vertex, below=0.7 of ik, label=left:$i_1$] (i1) {};
  \node[vertex, below=0.7 of i1, label=left:$i_2$] (i2) {};
  \node[vertex, below=1 of i3, label=left:$i_{k-1}$] (ikm) {};
  \node[vertex, below=1 of ikm, label=left:$i_N$] (iN) {};
  \node[vertex, right=1.5 of ik, label=above:$\brakket{\phi_{i_k}}{\hat{o}_1}{\phi_{i_k}}$] (xk) {};
  \node[vertex, right=1.5 of i1, label=above:$\delta_{i_2i_2}$] (x1) {};
  \node[vertex, right=1.5 of i2, label=above:$\delta_{i_1i_1}$] (x2) {};
  \node[vertex, right=1.5 of ikm] (xkm) {};
  \node[vertex, right=1.5 of iN] (xN) {};
  \foreach \n/\nt in {k,1,2,km/k-1,N} \node[vertex, right=1.5 of x\n, label=right:$i_{\nt}$] (j\n) {};
  \draw (ik) -- (jk);
  \draw (i1) to[out=0, in=180] (x2) to[out=0, in=180] (j1);
  \draw (i2) to[out=0, in=180] (x1) to[out=0, in=180] (j2);
  \draw (ikm) to[out=0, in=135] +(0.5,-0.25); \draw (jkm) to[in=45, out=180] +(-0.5,-0.25);
  \draw (iN) to[out=0, in=225] +(0.5,0.25); \draw (jN) to[in=-45, out=180] +(-0.5,0.25);
  \draw (xkm) to[out=0, in=135] +(0.5,-0.25); \draw (xkm) to[in=45, out=180] +(-0.5,-0.25);
  \draw (xN) to[out=0, in=225] +(0.5,0.25); \draw (xN) to[in=-45, out=180] +(-0.5,0.25);
  \foreach \r in {i,j} \node[anchor=center] at ($ (\r2.center) ! 0.4 ! (\r km.center) $){$\vdots$};
  \foreach \r in {i,j} \node[anchor=center] at ($ (\r km.center) ! 0.4 ! (\r N.center) $){$\vdots$};
\end{tikzpicture}
}} \notag \\
&= \brakket{\phi_{i_k}}{\hat{o}_1}{\phi_{i_k}} .
\end{align}
Again, we can permute the other indices $i_m \neq i_k$ in $(N-1)!$ ways without affecting the value. The expectation value for $J=I$ therefore becomes
\begin{align}
\brakket{\Phi_I}{\hat{O}_1}{\Phi_I}
= \frac{N\,(N-1)!}{N!}\sum_{k=1}^N\brakket{\phi_{i_k}}{\hat{o}_1}{\phi_{i_k}}
= \sum_{k=1}^N\brakket{\phi_{i_k}}{\hat{o}_1}{\phi_{i_k}} ,
\end{align}
where we included the pre-factor for the 1-body operators~\eqref{eq:one-body}.

Now consider the case that the orbitals $i_k$ and $j_l$ are different. The only option to get a non-zero contribution to the integral is to connect them to the operator $\hat{o}_1$, so we have
\begin{align}
\vcenter{\hbox{%
\begin{tikzpicture}
  \node[vertex,label=left:$i_1$] (i1) {};
  \node[vertex, below=0.7 of i1, label=left:$i_2$] (i2) {};
  \node[vertex, below=1 of i2, label=left:$i_k$] (ik) {};
  \node[vertex, below=0.7 of ik, label=left:$i_{k+1}$] (ikp) {};
  \node[vertex, below=1 of ikp, label=left:$i_l$] (il) {};
  \node[vertex, below=1 of il, label=left:$i_N$] (iN) {};
  \node[vertex, right=1.5 of i1, label=above:$\brakket{\phi_{i_k}}{\hat{o}_1}{\phi_{j_l}}$] (x1) {};
  \node[vertex, right=1.5 of i2, label=above:$\delta_{i_1j_1}$] (x2) {};
  \node[vertex, right=1.5 of ik, label=above:$\delta_{i_{k+1}j_k}$] (xk) {};
  \node[vertex, right=1.5 of il, label=above:$\delta_{i_Nj_N}$] (xl) {};
  \foreach \n in {kp,N} \node[vertex, right=1.5 of i\n] (x\n) {};
  \foreach \n/\nt in {1,2,k,kp/k+1,l,N} \node[vertex, right=1.5 of x\n, label=right:$j_{\nt}$] (j\n) {};
  \draw (ik) to[out=0, in=180] (x1) to[out=0, in=180] (jl);
  \draw (i1) to[out=0, in=180] (x2) to[out=0, in=180] (j1);
  \draw (ikp) to[out=0, in=180] (xk) to[out=0, in=180] (jk);
  \draw (jkp) -- (xkp) to[out=180, in=45] +(-0.5,-0.25);
  \draw (il) to[out=0, in=225] +(0.5,0.25);
  \draw (i2) to[out=0, in=135] +(0.5,-0.25); \draw (j2) to[in=45, out=180] +(-0.5,-0.25);
  \draw (iN) to[out=0, in=180] (xl) to[out=0, in=180] (jN);
  \draw (xN) to[out=0, in=225] +(0.5,0.25); \draw (xN)  to[in=-45, out=180] +(-0.5,0.25);
  \foreach \r in {i,j} \node[anchor=center] at ($ (\r2.center) ! 0.4 ! (\r k.center) $){$\vdots$};
  \foreach \r in {i,j} \node[anchor=center] at ($ (\r kp.center) ! 0.4 ! (\r l.center) $){$\vdots$};
  \foreach \r in {i,j} \node[anchor=center] at ($ (\r l.center) ! 0.4 ! (\r N.center) $){$\vdots$};
\end{tikzpicture}
}}
&= (-1)^{k-l}\vcenter{\hbox{%
\begin{tikzpicture}
  \node[vertex,label=left:$i_k$] (ik) {};
  \node[vertex, below=0.7 of ik, label=left:$i_1$] (i1) {};
  \node[vertex, below=1 of i1, label=left:$i_{k-1}$] (ikm) {};
  \node[vertex, below=0.7 of ikm, label=left:$i_{k+1}$] (ikp) {};
  \node[vertex, below=1 of ikp, label=left:$i_l$] (il) {};
  \node[vertex, below=1 of il, label=left:$i_N$] (iN) {};
  \node[vertex, right=1.5 of ik, label=above:$\brakket{\phi_{i_k}}{\hat{o}_1}{\phi_{j_l}}$] (xk) {};
  \node[vertex, right=1.5 of i1, label=above:$\delta_{i_1j_1}$] (x1) {};
  \node[vertex, right=1.5 of ikm, label=above:$\delta_{i_{k+1}j_k}$] (xkm) {};
  \node[vertex, right=1.5 of il, label=above:$\delta_{i_Nj_N}$] (xl) {};
  \foreach \n in {kp,N} \node[vertex, right=1.5 of i\n] (x\n) {};
  \foreach \n/\nt in {k,1,km/k-1,kp/k,l/l-1,N} \node[vertex, right=1.5 of x\n, label=right:$j_{\nt}$] (j\n) {};
  \draw (ik) -- (jk);
  \draw (i1) -- (j1);
  \draw (ikp) to[out=0, in=180] (xkm) to[out=0, in=180] (jkp);
  \draw (ikm) to[out=0, in=135] +(0.5,-0.25); \draw (jkm) to[in=45, out=180] +(-0.5,-0.25);
  \draw (il) to[out=0, in=225] +(0.5,0.25); \draw (jl) to[out=180, in=-45] +(-0.5,0.25);
  \draw (xkp) to[out=180, in=45] +(-0.5,-0.25); \draw (xkp) to[out=0, in=135] +(0.5,-0.25);
  \draw (iN) to[out=0, in=180] (xl) to[out=0, in=180] (jN);
  \draw (xN) to[out=0, in=225] +(0.5,0.25); \draw (xN)  to[in=-45, out=180] +(-0.5,0.25);
  \foreach \r in {i,j} \node[anchor=center] at ($ (\r1.center) ! 0.4 ! (\r km.center) $){$\vdots$};
  \foreach \r in {i,j} \node[anchor=center] at ($ (\r kp.center) ! 0.4 ! (\r l.center) $){$\vdots$};
  \foreach \r in {i,j} \node[anchor=center] at ($ (\r l.center) ! 0.4 ! (\r N.center) $){$\vdots$};
\end{tikzpicture}
}} \notag \\
&= (-1)^{k-l}\brakket{\phi_{i_k}}{\hat{o}_1}{\phi_{j_l}} ,
\end{align}
where the $(-1)^{k-l}$ phase factor is the result of the permutation $(12)(23)\dotsb(k-1\,k) \to (-1)^{k-1}$ on the left side to get $i_k$ on the first position without affecting the ordering of the other indices and the permutation $(12)(23)\dotsb(l-1\,l) \to (-1)^{l-1}$ on the left side to get $j_l$ on the first position.

As in the previous diagrams, we can permute the other indices in $(N-1)!$ without affecting (zeroing) the integral and we find
\begin{align}
\brakket{\Phi_I}{\hat{O}_1}{\Phi_J} = (-1)^{k-l}\brakket{\phi_{i_k}}{\hat{o}_1}{\phi_{j_l}} ,
\end{align}
if only the orbitals $i_k$ and $j_l$ differ. Collecting the results for the 1-body operators, we find as Slater--Condon rule
\begin{align}
\brakket{\Phi_I}{\hat{O}_1}{\Phi_J} = \begin{dcases*}
\sum_{k=1}^N\brakket{\phi_{i_k}}{\hat{o}_1}{\phi_{i_k}}	&if $I = J$, \\
(-1)^{k-l}\brakket{\phi_{i_k}}{\hat{o}_1}{\phi_{j_l}}		&if only $i_k \neq j_l$, \\
0											&otherwise .
\end{dcases*}
\end{align}

\subsection{2-body operators}
Now we consider the 2-body operators, so we have a non-trivial integral with two pairs of orbitals
\begin{align}
\brakket{i_ki_l}{\hat{o}_{12}}{j_mj_n}
&\isDefinedAs \brakket{\phi_{i_k}\phi_{i_l}}{\hat{o}_{12}}{\phi_{j_m}\phi_{j_n}} \\
&\isDefinedAs \integ{\vecx_1}\!\!\integ{\vecx_2}
\phi_{i_k}^*(\vecx_1)\phi_{i_l}^*(\vecx_1)\hat{o}(\vecx_1,\vecx_2)\phi_{j_m}(\vecx_1)\phi_{j_n}(\vecx_2) , \notag
\end{align}
where we have omitted the explicit notation of $\phi$, as it is clear that we are always referring to the orbitals $\phi$ and it makes the notation less bulky.
For the two-body operators we will find that the integrals always come in pairs and we will use the following notation for such a pair
\begin{align}
\brakketx{kl}{\hat{o}_{12}}{mn}
&\isDefinedAs \brakket{kl}{\hat{o}_{12}}{mn} - \brakket{kl}{\hat{o}_{12}}{nm} .
\end{align}
As we have now a two body operator even three orbitals in the determinants need to be different to make the integral identically zero. So now we need to consider three cases: 1) no orbitals different, 2) one orbital different, 3) two orbitals different.

Let us start again when both determinants are comprised of the same orbital set, so $I = J$. For simplicity we will first consider only terms where the operator works on $i_1$ and $i_2$
\begin{align}\label{eq:twoBodyNoOrbDifH}
\vcenter{\hbox{%
\begin{tikzpicture}
  \node[vertex,label=left:$i_1$] (i1) {};
  \foreach \n/\nmin in {2/1,3/2,4/3}
    \node[vertex, below=0.7 of i\nmin, label=left:$i_{\n}$] (i\n) {};
  \node[vertex, below=1 of i4, label=left:$i_N$] (iN) {};
  \node[vertex, right=1.5 of i1, label=above:$\brakket{i_1i_2}{\hat{o}_{12}}{i_1i_2}$] (x1) {};
  \node[vertex, right=1.5 of i2] (x2) {};
  \foreach \n/\k in {3,4,N}
    \node[vertex, right=1.5 of i\k, label=above:$\delta_{i_{\n}i_{\n}}$] (x\k) {};
  \foreach \n in {1,2,3,4,N} \node[vertex, right=1.5 of x\n, label=right:$i_{\n}$] (j\n) {};
  \foreach \n/\k in {1,2,3,4,N} \draw (i\n) to[out=0, in=180] (x\k) to[out=0, in=180] (j\n);
  \draw[dashed] (x1) -- (x2);
  \foreach \r in {i,j} \node[anchor=center] at ($ (\r4.center) ! 0.4 ! (\r N.center) $){$\vdots$};
\end{tikzpicture}
}}
&= \vcenter{\hbox{%
\begin{tikzpicture}
  \node[vertex,label=left:$i_1$] (i1) {};
  \foreach \n/\nmin in {2/1,3/2,4/3}
    \node[vertex, below=0.7 of i\nmin, label=left:$i_{\n}$] (i\n) {};
  \node[vertex, below=1 of i4, label=left:$i_N$] (iN) {};
  \node[vertex, right=1.5 of i1, label=above:$\brakket{i_2i_1}{\hat{o}_{12}}{i_2i_1}$] (x1) {};
  \node[vertex, right=1.5 of i2] (x2) {};
  \foreach \n/\k in {3,4,N}
    \node[vertex, right=1.5 of i\k, label=above:$\delta_{i_{\n}i_{\n}}$] (x\k) {};
  \foreach \n in {1,2,3,4,N} \node[vertex, right=1.5 of x\n, label=right:$i_{\n}$] (j\n) {};
  \foreach \n/\k in {1/2,2/1,3,4,N} \draw (i\n) to[out=0, in=180] (x\k) to[out=0, in=180] (j\n);
  \draw[dashed] (x1) -- (x2);
  \foreach \r in {i,j} \node[anchor=center] at ($ (\r4.center) ! 0.4 ! (\r N.center) $){$\vdots$};
\end{tikzpicture}
}} \notag \\
&= \brakket{i_1i_2}{\hat{o}_{12}}{i_1i_2} ,
\end{align}
where we used the symmetry of the 2-body operator. As the remaining indices can be permuted in $(N-2)!$ ways without affecting the integral, we have $2\,(N-2)!$ of the terms. Due to the operator in the integral, it is not necessary to directly connect $i_1$ to $i_1$ and $i_2$ to $i_2$ to have a possible non-zero contribution to the expectation value. Also connecting $i_1$ to $i_2$ can yield a contribution. These are graphs of the form
\begin{align}\label{eq:twoBodyNoOrbDifX}
\vcenter{\hbox{%
\begin{tikzpicture}
  \node[vertex,label=left:$i_1$] (i1) {};
  \foreach \n/\nmin in {2/1,3/2,4/3}
    \node[vertex, below=0.7 of i\nmin, label=left:$i_{\n}$] (i\n) {};
  \node[vertex, below=1 of i4, label=left:$i_N$] (iN) {};
  \node[vertex, right=1.5 of i1, label=above:$\brakket{i_2i_1}{\hat{o}_{12}}{i_1i_2}$] (x1) {};
  \node[vertex, right=1.5 of i2] (x2) {};
  \foreach \n/\k in {3,4,N}
    \node[vertex, right=1.5 of i\k, label=above:$\delta_{i_{\n}i_{\n}}$] (x\k) {};
  \foreach \n in {1,2,3,4,N} \node[vertex, right=1.5 of x\n, label=right:$i_{\n}$] (j\n) {};
  \foreach \n/\k in {1/2,2/1} \draw (i\n) to[out=0, in=180] (x\k) to[out=0, in=180] (j\k);
  \foreach \n/\k in {3,4,N} \draw (i\n) to[out=0, in=180] (x\k) to[out=0, in=180] (j\n);
  \draw[dashed] (x1) -- (x2);
  \foreach \r in {i,j} \node[anchor=center] at ($ (\r4.center) ! 0.4 ! (\r N.center) $){$\vdots$};
\end{tikzpicture}
}}
&= \vcenter{\hbox{%
\begin{tikzpicture}
  \node[vertex,label=left:$i_1$] (i1) {};
  \foreach \n/\nmin in {2/1,3/2,4/3}
    \node[vertex, below=0.7 of i\nmin, label=left:$i_{\n}$] (i\n) {};
  \node[vertex, below=1 of i4, label=left:$i_N$] (iN) {};
  \node[vertex, right=1.5 of i1, label=above:$\brakket{i_2i_1}{\hat{o}_{12}}{i_1i_2}$] (x1) {};
  \node[vertex, right=1.5 of i2] (x2) {};
  \foreach \n/\k in {3,4,N}
    \node[vertex, right=1.5 of i\k, label=above:$\delta_{i_{\n}i_{\n}}$] (x\k) {};
  \foreach \n in {1,2,3,4,N} \node[vertex, right=1.5 of x\n, label=right:$i_{\n}$] (j\n) {};
  \foreach \n/\k in {1/2,2/1} \draw (i\k) to[out=0, in=180] (x\k) to[out=0, in=180] (j\n);
  \foreach \n/\k in {3,4,N} \draw (i\n) to[out=0, in=180] (x\k) to[out=0, in=180] (j\n);
  \draw[dashed] (x1) -- (x2);
  \foreach \r in {i,j} \node[anchor=center] at ($ (\r4.center) ! 0.4 ! (\r N.center) $){$\vdots$};
\end{tikzpicture}
}} \notag \\
&= \brakket{i_1i_2}{\hat{o}_{12}}{i_2i_1} .
\end{align}
Again, by interchanging the remaining nodes we see the there are $2\,(N-2)!$ of these terms. Taking into account that in the last diagrams we got 1 additional crossing, connecting the nodes $i_1$ and $i_2$ to the operator yields the following contribution to the expectation value
\begin{align}
2(N-2)!\bigl(\brakket{i_1i_2}{\hat{o}_{12}}{i_1i_2} - \brakket{i_1i_2}{\hat{o}_{12}}{i_2i_1}\bigr) .
\end{align}
To get any other pair on the first positions, we always make an even number of permutations and the contribution will have exactly the same form as for the $i_1,i_2$ pair. So to get the full expectation value, we simply need to sum over all possible combinations
\begin{align}
\brakket{\Phi_I}{\hat{O}_2}{\Phi_I}
&= \sum_{r < s}\bigl(\brakket{i_ri_s}{\hat{o}_{12}}{i_ri_s} - \brakket{i_ri_s}{\hat{o}_{12}}{i_si_r}\bigr)
\notag \\
&= \sum_{r < s}\brakketx{i_ri_s}{\hat{o}_{12}}{i_ri_s} ,
\end{align}
where we took the pre-factor for the 2-body operator~\eqref{eq:two-body} into account.

Now consider the case that the determinants only differ by one orbital, i.e.\ orbitals $i_k$ and $j_m$. This implies that these orbitals always need to be connected to the operator to have a non-zero contribution to the expectation value. A typical diagram now looks
\begin{multline}\label{eq:twoBodyOneOrbDifH}
\vcenter{\hbox{%
\begin{tikzpicture}
  \node[vertex,label=left:$i_1$] (i1) {};
  \node[vertex, below=0.7 of i1, label=left:$i_2$] (i2) {};
  \node[vertex, below=1 of i2, label=left:$i_k$] (ik) {};
  \node[vertex, below=1 of ik, label=left:$i_{r-1}$] (irm) {};
  \node[vertex, below=0.7 of irm, label=left:$i_r$] (ir) {};
  \node[vertex, below=1 of ir, label=left:$i_m$] (im) {};
  \node[vertex, right=1.5 of i1, label=above:$\brakket{i_ki_r}{\hat{o}_{12}}{j_mi_r}$] (x1) {};
  \foreach \n in {2,k,r,rm,m} \node[vertex, right=1.5 of i\n] (x\n) {};
  \draw[dashed] (x1) -- (x2);
  \foreach \n/\nm in {1,2,r,rm/r-1,k,m} \node[vertex, right=1.5 of x\n, label=right:{$j_{\nm}$}] (j\n) {};
  \draw (ik) to[out=0, in=180] (x1) to[out=0, in=180] (jm);
  \draw (ir) to[out=0, in=180] (x2) to[out=0, in=180] (jrm);
  \draw (i1) to[out=0, in=180] (xk) to[out=0, in=180] (j1);
  \draw (i2) to[out=0, in=180] (xrm) to[out=0, in=180] (j2);
  \draw (jk) to[out=180, in=45] +(-0.5,-0.25);
  \draw (irm) to[out=0, in=225] +(0.5,0.25);
  \draw (jr) to[out=180, in=0] (xm) to[out=180, in=-45] +(-0.5,0.25);
  \draw (im) to[out=0, in=180] (xr) to[out=0, in=225] +(0.5,0.25);
  \foreach \r in {i,j} \node[anchor=center] at ($ (\r2.center) ! 0.4 ! (\r k.center) $){$\vdots$};
  \foreach \r in {i,j} \node[anchor=center] at ($ (\r k.center) ! 0.4 ! (\r rm.center) $){$\vdots$};
  \foreach \r in {i,j} \node[anchor=center] at ($ (\r r.center) ! 0.4 ! (\r m.center) $){$\vdots$};
  \foreach \r in {i,j} \node[anchor=center, below=0 of \r m] {$\vdots$};
\end{tikzpicture}
}}
\!\!\!\!= (-1)^{\mathrlap{k-m}}\;\;\vcenter{\hbox{%
\begin{tikzpicture}
  \node[vertex,label=left:$i_k$] (ik) {};
  \node[vertex, below=0.7 of ik, label=left:$i_1$] (i1) {};
  \node[vertex, below=0.7 of i1, label=left:$i_2$] (i2) {};
  \node[vertex, below=1 of i2, label=left:$i_{r-1}$] (irm) {};
  \node[vertex, below=0.7 of irm, label=left:$i_r$] (ir) {};
  \node[vertex, right=1.5 of ik, label=above:$\brakket{i_ki_r}{\hat{o}_{12}}{j_mi_r}$] (xk) {};
  \foreach \n in {1,2,r,rm} \node[vertex, right=1.5 of i\n] (x\n) {};
  \draw[dashed] (xk) -- (x1);
  \foreach \n/\nt in {1,2,rm/r-2,r/r-1} \node[vertex, right=1.5 of x\n, label=right:{$j_{\nt}$}] (j\n) {};
  \node[vertex, right=1.5 of xk, label=right:{$j_m$}] (jm) {};
  \draw (ik) -- (jm);
  \draw (ir) to[out=0, in=180] (x1) to[out=0, in=180] (jr);
  \draw (i1) to[out=0, in=180] (x2) to[out=0, in=180] (j1);
  \draw (i2) to[out=0, in=180] (xrm) to[out=0, in=180] (j2);
  \draw (irm) to[out=0, in=180] (xr) to[out=0, in=180] (jrm);
  \foreach \r in {i,j} \node[anchor=center] at ($ (\r2.center) ! 0.4 ! (\r rm.center) $){$\vdots$};
  \foreach \r in {i,j} \node[anchor=center, below=0 of \r r] {$\vdots$};
\end{tikzpicture}
}} \\
= (-1)^{k-m}\vcenter{\hbox{%
\begin{tikzpicture}
  \node[vertex,label=left:$i_k$] (ik) {};
  \node[vertex, below=0.7 of ik, label=left:$i_r$] (ir) {};
  \node[vertex, below=0.7 of ir, label=left:$i_1$] (i1) {};
  \node[vertex, below=0.7 of i1, label=left:$i_2$] (i2) {};
  \node[vertex, right=1.5 of ik, label=above:$\brakket{i_ki_r}{\hat{o}_{12}}{j_mi_r}$] (xk) {};
  \foreach \n in {r,1,2} \node[vertex, right=1.5 of i\n] (x\n) {};
  \draw[dashed] (xk) -- (xr);
  \foreach \n/\nt in {r/r-1,1,2} \node[vertex, right=1.5 of x\n, label=right:{$j_{\nt}$}] (j\n) {};
  \node[vertex, right=1.5 of xk, label=right:{$j_m$}] (jm) {};
  \draw (ik) -- (jm);
  \draw (ir) -- (jr);
  \draw (i1) -- (j1);
  \draw (i2) -- (j2);
  \foreach \r in {i,j} \node[anchor=center, below=0 of \r2] {$\vdots$};
\end{tikzpicture}
}}
\!\!\!\!= (-1)^{k-m}\brakket{i_ki_r}{\hat{o}_{12}}{j_mi_r} .
\end{multline}
The nodes not attached to the operator can be permuted in $(N-2)!$ ways without affecting the integral. Additionally we permute the first two nodes on both sides of the last diagram of~\eqref{eq:twoBodyOneOrbDifH} just as in~\eqref{eq:twoBodyNoOrbDifH}, so we have $2(N-2)!$ of these terms. Also we have the exchange variant as in~\eqref{eq:twoBodyNoOrbDifX}, so the full contribution becomes involving the orbital $\phi_{i_r}$ becomes
\begin{align}
(-1)^{k-m}2(N-2)!\bigl(\brakket{i_ki_r}{\hat{o}_{12}}{j_mi_r} - \brakket{i_ki_r}{\hat{o}_{12}}{i_rj_m}\bigr)
\end{align}
Now we need only to sum over all $i_r \neq i_k$ and to take the 2-body pre-factor into account~\eqref{eq:two-body} to get the complete expectation value
\begin{align}
\brakket{\Phi_I}{\hat{O}_2}{\Phi_J}
&= (-1)^{k-m}\sum_{r \neq k}
\bigl(\brakket{i_ki_r}{\hat{o}_{12}}{j_mi_r} - \brakket{i_ki_r}{\hat{o}_{12}}{i_rj_m}\bigr) \\
&= (-1)^{k-m}\sum_{r \neq k}\brakketx{i_ki_r}{\hat{o}_{12}}{j_mi_r}
\end{align}
As the final case we need to consider the possibility that the determinants differ in precisely two orbitals with indices $i_k < i_l$ and $i_m < i_n$. With this ordering of the different orbitals, the `direct' contribution becomes
\begin{align}\label{eq:twoBodyTwoOrbDifH}
\vcenter{\hbox{%
\begin{tikzpicture}
  \node[vertex,label=left:$i_1$] (i1) {};
  \node[vertex, below=0.7 of i1, label=left:$i_2$] (i2) {};
  \node[vertex, below=1 of i2, label=left:$i_k$] (ik) {};
  \node[vertex, below=1 of ik, label=left:$i_l$] (il) {};
  \node[vertex, right=1.5 of i1, label=above:$\brakket{i_ki_l}{\hat{o}_{12}}{j_mj_n}$] (x1) {};
  \node[vertex, right=1.5 of i2] (x2) {};
  \draw[dashed] (x1) -- (x2);
  \foreach \n/\nm in {2/1,3/2,4/3} \node[vertex, below=0.7 of x\nm] (x\n) {};
  \foreach \n/\nm in {1,2} \node[vertex, right=1.5 of x\n, label=right:{$j_{\n}$}] (j\n) {};
  \node[vertex, below=1.5 of j2, label=right:$j_m$] (jm) {};
  \node[vertex, below=1 of jm, label=right:$j_n$] (jn) {};
  \draw (ik) to[out=0, in=180] (x1) to[out=0, in=180] (jm);
  \draw (il) to[out=0, in=180] (x2) to[out=0, in=180] (jn);
  \foreach \n/\x in {1/3,2/4} \draw (i\n) to[out=0, in=180] (x\x) to[out=0, in=180] (j\n);
  \foreach \b/\e in {i2/ik,ik/il,j2/jm,jm/jn}
    \node[anchor=center] at ($ (\b.center) ! 0.4 ! (\e.center) $){$\vdots$};
  \foreach \r in {il,jn} \node[anchor=center, below=0 of \r] {$\vdots$};
\end{tikzpicture}
}}
&= (-1)^{k-m}\vcenter{\hbox{%
\begin{tikzpicture}
  \node[vertex,label=left:$i_k$] (ik) {};
  \node[vertex, below=0.7 of ik, label=left:$i_1$] (i1) {};
  \node[vertex, below=0.7 of i1, label=left:$i_2$] (i2) {};
  \node[vertex, below=1 of i2, label=left:$i_l$] (il) {};
  \node[vertex, right=1.5 of ik, label=above:$\brakket{i_ki_l}{\hat{o}_{12}}{j_mj_n}$] (xk) {};
  \node[vertex, right=1.5 of i1] (x1) {};
  \node[vertex, right=1.5 of i2] (x2) {};
  \node[vertex, below=0.7 of x2] (x3) {};
  \draw[dashed] (xk) -- (x1);
  \foreach \n/\nm in {k/m,1,2} \node[vertex, right=1.5 of x\n, label=right:{$j_{\nm}$}] (j\nm) {};
  \node[vertex, below=1.5 of j2, label=right:$j_n$] (jn) {};
  \draw (ik) -- (jm);
  \draw (il) to[out=0, in=180] (x1) to[out=0, in=180] (jn);
  \foreach \n/\x in {1/2,2/3} \draw (i\n) to[out=0, in=180] (x\x) to[out=0, in=180] (j\n);
  \foreach \b/\e in {i2/il,j2/jn}
    \node[anchor=center] at ($ (\b.center) ! 0.4 ! (\e.center) $){$\vdots$};
  \foreach \r in {il,jn} \node[anchor=center, below=0 of \r] {$\vdots$};
\end{tikzpicture}
}} \notag \\
&= (-1)^{k+l-m-n}\vcenter{\hbox{%
\begin{tikzpicture}
  \node[vertex,label=left:$i_k$] (ik) {};
  \foreach \n/\nm in {l/k,1/l,2/1} \node[vertex, below=0.7 of i\nm, label=left:$i_{\n}$] (i\n) {};
  \node[vertex, right=1.5 of ik, label=above:$\brakket{i_ki_l}{\hat{o}_{12}}{j_mj_n}$] (xk) {};
  \foreach \n in {l,1,2} \node[vertex, right=1.5 of i\n] (x\n) {};
  \draw[dashed] (xk) -- (xl);
  \foreach \n/\nm in {k/m,l/n,1,2} \node[vertex, right=1.5 of x\n, label=right:{$j_{\nm}$}] (j\nm) {};
  \foreach \n/\nm in {k/m,l/n,1,2} \draw (i\n) -- (j\nm);
  \foreach \r in {i2,j2} \node[anchor=center, below=0 of \r] {$\vdots$};
\end{tikzpicture}
}} \notag \\
&= (-1)^{k+l-m-n}\brakket{i_ki_l}{\hat{o}_{12}}{j_mj_n} .
\end{align}
The same arguments as in the other cases tell us that there are $2(N-2)!$ of these terms and that there is also an exchange variant as in~\eqref{eq:twoBodyNoOrbDifX}. The including the two-body pre-factor~\eqref{eq:two-body}, the full expectation value becomes
\begin{align}
\brakket{\Phi_I}{\hat{O}_2}{\Phi_J}
&= (-1)^{k+l-m-n}\bigl(\brakket{i_ki_l}{\hat{o}_{12}}{j_mj_n} -\brakket{i_ki_l}{\hat{o}_{12}}{j_nj_m}\bigr) \notag\\
&= (-1)^{k+l-m-n}\brakketx{i_ki_l}{\hat{o}_{12}}{j_mj_n} .
\end{align}
Collecting all the results, we find the following Slater--Condon rule for 2-body operators
\begin{equation}
\brakket{\Phi_I}{\hat{O}_2}{\Phi_J}
= \begin{dcases*}
\sum_{r < s}\brakketx{i_ri_s}{\hat{o}_{12}}{i_ri_s}
	&if $I = J$, \\
(-1)^{k-m}\sum_{r \neq k}\brakketx{i_ki_r}{\hat{o}_{12}}{j_mi_r}
	&if only $i_k \neq j_m$, \\
(-1)^{k+l-m-n}\brakketx{i_ki_l}{\hat{o}_{12}}{j_mj_n}
	&if $i_k \neq j_m$ and $i_l \neq j_n$, \\
0	&otherwise.
\end{dcases*}\!\!\!\!\!\!\!\!\!\!\!\!
\end{equation}

\subsection{More integral notation}
The interaction between the electrons will be typically the two-body operator under consideration, so we will introduce even shorter notation for this particular two-body operator
\begin{subequations}
\begin{align}
\braket{ij}{kl}
\isDefinedAs{}& \brakket{ij}{\hat{w}_{12}}{kl}
= \brakket{\phi_i\phi_j}{\hat{w}_{12}}{\phi_k\phi_l} \\
{}={}& \integ{\vecx_1}\!\!\integ{\vecx_2}
\phi^*_i(\vecx_1)\phi_j^*(\vecx_2)\,w(\vecr_1,\vecr_2)\,\phi_k(\vecx_1)\phi_l(\vecx_2) . 
\notag
\end{align}
This way of writing the two-electron (interaction) integrals is called the physicist's notation. As there is a physicist's notation, there is also a chemist's notation which collects the charge densities instead of the bra-ket combinations
\begin{align}
\cbraket{ij}{kl} \isDefinedAs
\integ{\vecx_1}\!\!\integ{\vecx_2}
\phi^*_i(\vecx_1)\phi_j(\vecx_1)\,w(\vecr_1,\vecr_2)\,\phi_k^*(\vecx_2)\phi_l(\vecx_2) .
\end{align}
\end{subequations}
Both notations also have a variant which includes the exchange term
\begin{subequations}
\begin{align}
\braketx{ij}{kl} &\isDefinedAs \braket{ij}{kl} - \braket{ij}{lk} , \\
\cbraketx{ij}{kl} &\isDefinedAs \cbraket{ij}{kl} - \cbraket{il}{kj} .
\end{align}
\end{subequations}

\begin{exercise}
Show that the \ac{HF} energy functional can be written as
\begin{align}
E^{\text{HF}}[\{\phi_i\}] = \sum_{i=1}^N\brakket{\phi_i}{\hat{h}}{\phi_i} +
\half\sum_{i,j=1}^N\braketx{\phi_i\phi_j}{\phi_i\phi_j} + E_{\text{nuc}} ,
\end{align}
where $\hat{h} \isDefinedAs -\half\nabla^2 + v(\vecr)$.
\end{exercise}

\begin{exercise}
Show that $\braketx{ij}{kl} = \cbraketx{ik}{jl}$.
\end{exercise}

\begin{exercise}[\textbf{\ac*{RHF} for H$_2$ in a minimal basis}]
\label{ex:RHF}
In this exercise we will calculate the (singlet) \acf*{RHF} energy for the H$_2$ molecule in a minimal basis\index{minimal basis}, i.e.\ two $1s$ orbitals on each hydrogen atom
\begin{align}
\chi_A(\vecr) &= N_{\zeta}\e^{-\zeta\abs{\vecr - \vecR_A}} &
\chi_B(\vecr) &= N_{\zeta}\e^{-\zeta\abs{\vecr - \vecR_B}}
\end{align}

\begin{subexercise}
\item Show that the normalisation constant of the $1s$ orbitals is
\begin{align}
N_{\zeta} = \sqrt{\frac{\zeta^3}{\pi}} .
\end{align}

\item Show that the overlap of the $1s$ functions is
\begin{align}
s &\isDefinedAs \braket{\chi_A}{\chi_B}
= \biggl(1 + \rho + \frac{\rho^2}{3}\biggr)\e^{-\rho} ,
\end{align}
where $\rho \isDefinedAs \zeta R$ and $R \isDefinedAs \abs{\vecR_A - \vecR_B}$ is the internuclear distance.

\textbf{Hint:} You can use different coordinate systems to solve the integral. You might be inclined to use the more familiar spherical coordinate system, but the integrals are not very straightforward. The most natural one which leads to the easiest integrals is the prolate spheroidal coordinate system (Ap.~\ref{ap:prolSperCoords}), as there are two nuclei which can be placed at the foci of ellipses.
First show that $\chi_{A/B} = N_{\zeta}\e^{-\frac{\rho}{2}(\xi \mp \eta)}$ and then perform the integral over $\xi$, $\eta$ and $\phi$.
\end{subexercise}

\noindent
The $1s$ functions on the different hydrogen atoms are not orthogonal ($s \neq 0$). One way to orthogonalise them is to make symmetry adapted combination ($\sigma_g$\slash{}$\sigma_u$). Due to the symmetry of the system, we expect one of them to be the \ac{HF} solution.

\begin{subexercise}[resume]
\item Show that the normalised symmetry adapted basis functions are
\begin{align}
\sigma_g &= \frac{\chi_A + \chi_B}{\sqrt{2(1+s)}} &
&\text{and} &
\sigma_u &= \frac{\chi_A - \chi_B}{\sqrt{2(1-s)}} .
\end{align}
Why are these functions orthogonal?

\item Argue that a doubly occupied $\sigma_g$ orbital, i.e.\ a determinant with spin up and a spin down electron in the $\sigma_g$ orbital would yield the lowest energy.
\end{subexercise}

\noindent
To calculate the \acs*{RHF} energy, we need to evaluate the expectation values of the different parts of the Hamiltonian. We will first consider the one-body part.

\begin{subexercise}[resume]
\item\label{ex:RHFoneBody} Show that
\begin{subequations}
\begin{align}
\brakket{\chi_A}{\hat{h}}{\chi_A} &= \brakket{\chi_B}{\hat{h}}{\chi_B}
= \frac{\zeta^2}{2} - \zeta + \frac{\zeta}{\rho}\bigl(\e^{-2\rho}(1 + \rho) - 1\bigr) , \\
\brakket{\chi_A}{\hat{h}}{\chi_B}
&= \biggl[\frac{\zeta^2}{2}\biggl(1 + \rho - \frac{\rho^2}{3}\biggr) - 2\zeta(1 + \rho)\biggr]\e^{-\rho} .
\end{align}
\end{subequations}
\textbf{Hint:} If you do not want to evaluate the Laplacian explicitly, you can use that the orbitals $\chi_{A/B}$ solve a hydrogenic Schrödinger equation with charge $\zeta$, i.e.
\begin{align}
\biggl(-\half\nabla^2 - \frac{\zeta}{r_{A/B}}\biggr)\chi_{A/B} = -\frac{\zeta^2}{2}\chi_{A/B} .
\end{align}

\item Evaluate the long bond distance limit, $R \to \infty$, of the one-body terms. Did you expect this result? Explain.

\item Show that
\begin{align}
\brakket{\sigma_g}{\hat{h}}{\sigma_g}
&= \frac{\brakket{\chi_A}{\hat{h}}{\chi_A} + \brakket{\chi_A}{\hat{h}}{\chi_B}}{1 + s} .
\end{align}
\end{subexercise}

\noindent
\ifbool{MullikenIntegApprox}{%
Evaluation of the two-body part is more challenging in general, since the two-electron integrals are 6D. To ease the computation, we use the Mulliken approximation~\autocite{Mulliken1949}.
The idea of the Mulliken approximation is to regard a two-electron integral as an interactions between the two charge densities $\phi_i^*(\vecr)\phi_l(\vecr)$ and $\phi_j^*(\vecr)\phi_k(\vecr)$. These charge densities will be proportional to the overlap between the orbitals, so Mulliken approximated them as
\begin{align}
\phi_i^*(\vecr)\phi_l(\vecr) \approx S_{il}\frac{\abs{\phi_i(\vecr)}^2 + \abs{\phi_l(\vecr)}^2}{2}
\end{align}
and similar for $\phi_j^*(\vecr)\phi_k(\vecr)$ of course. Using this approximation to the charge densities, one finds for the integrals
\begin{align}
\cbraket{il}{jk} \approx \cbraket{il}{jk}_M
\isDefinedAs \frac{S_{il}S_{jk}}{4}\bigl[\cbraket{ii}{jj} + \cbraket{ii}{kk} +
\cbraket{ll}{jj} + \cbraket{ll}{kk}\bigr] .
\end{align}}

\begin{subexercise}[resume]
\item Show that
\ifbool{MullikenIntegApprox}{%
\begin{subequations}
\begin{align}
\cbraket{AA}{AB}_M
&= \frac{s}{2}\bigl[\cbraket{AA}{AA} + \cbraket{AA}{BB}\bigr] , \\
\cbraket{AB}{AB}_M
&= \frac{s^2}{2}\bigl[\cbraket{AA}{AA} + \cbraket{AA}{BB}\bigr] , \\
\cbraket{\sigma_g\sigma_g}{\sigma_g\sigma_g}
&= \frac{1}{2(1+s)^2}\bigl[\cbraket{AA}{AA} + 4\cbraket{AA}{AB} + {} \notag \\*
&\eqspace\hphantom{\frac{1}{2(1+s)^2}\bigl[}
\cbraket{AA}{BB} + 2\cbraket{AB}{AB}\bigr] , \\
\cbraket{\sigma_g\sigma_g}{\sigma_g\sigma_g}_M
&= \half\bigl[\cbraket{AA}{AA} + \cbraket{AA}{BB}\bigr] ,
\end{align}
\end{subequations}}{%
\begin{align}
\cbraket{\sigma_g\sigma_g}{\sigma_g\sigma_g}
&= \frac{1}{2(1+s)^2}\bigl[\cbraket{AA}{AA} + 4\cbraket{AA}{AB} + {} \notag \\*
&\eqspace\hphantom{\frac{1}{2(1+s)^2}\bigl[}
\cbraket{AA}{BB} + 2\cbraket{AB}{AB}\bigr] ,
\end{align}}
where we used the abbreviations $A/B = \chi_{A/B}$.
\end{subexercise}

\noindent
\ifbool{MullikenIntegApprox}{%
So we find that in the Mulliken approximation, only $\cbraket{AA}{BB}$ needs to be calculated, as $\cbraket{AA}{AA}$ follows by taking the limit $\vecR_B \to \vecR_A$.}

\begin{subexercise}[resume]
\item Show that
\begin{align}
\cbraket{AA}{BB} &= \zeta\biggl[\frac{1}{\rho} - \e^{-2\rho}
\biggl(\frac{1}{\rho} + \frac{11}{8} + \frac{3\rho}{4} + \frac{\rho^2}{6}\biggr)\biggr] \ifbool{MullikenIntegApprox}{.}{ , \\
\cbraket{AA}{AB} &= \frac{\zeta\e^{-\rho}}{16\rho}\Bigl(5 + 2\rho + 16\rho^2 - (5+2\rho)\e^{-2\rho}\Bigr)
}
\end{align}
\textbf{Hint:} First evaluate the Coulomb potential due to one of the charge densities, e.g.\ the density $\abs{\chi_A(\vecr)}^2$ yields the Coulomb potential
\begin{align}
v_A(\vecr') = \integ{\vecr}\frac{\abs{\chi_A(\vecr)}^2}{\abs{\vecr' - \vecr}} .
\end{align}
Note that you already calculated this integral for $\vecr' = \vecR_B$ in part~\ref{ex:RHFoneBody}. Next you only need to work out $\integ{\vecr}\abs{\chi_B(\vecr)}^2v_A(\vecr)$, whose parts you also already did before.

\item Show that
\begin{align}
\cbraket{AA}{AA} = \lim_{R \to 0}\cbraket{AA}{BB} = \frac{5\zeta}{8} .
\end{align}
\end{subexercise}

\ifbool{MullikenIntegApprox}{}{
\noindent
The remaining integral turns out to be quite nasty (see Sec.~\ref{sec:extraIntegralInfo} for more details) and can only be reduced to
\begin{align}\label{eq:cABAB}
\cbraket{AB}{AB}
&= \frac{\zeta\e^{-2\rho}}{120\rho}\Bigl[75\rho - 138\rho^2 - 72\rho^3 - 8\rho^4+ {} \notag \\*
&\eqspace\hphantom{\frac{\zeta\e^{-2\rho}}{120\rho}\Bigl[}
16\bigl(\gamma+\ln(\rho)\bigr)(9 + 18\rho + 15\rho^2 + 6\rho^3 + \rho^4) - {} \notag \\*
&\eqspace\hphantom{\frac{\zeta\e^{-2\rho}}{120\rho}\Bigl[}
32(9 - 3\rho^2 + \rho^4)\e^{2\rho}\Ei(-2\rho) + {} \notag \\*
&\eqspace\hphantom{\frac{\zeta\e^{-2\rho}}{120\rho}\Bigl[}
16(3 - 3\rho + \rho^2)^2\e^{4\rho}\Ei(-4\rho)\Bigr],
\end{align}
where the \href{http://mathworld.wolfram.com/Euler-MascheroniConstant.html}{Euler--Mascheroni constant} is defined as
\begin{align}
\gamma \isDefinedAs -\binteg{z}{0}{\infty}\e^{-z}\ln(z) = 0.577\,215\,664\,9\ldots
\end{align}
and the \href{http://mathworld.wolfram.com/ExponentialIntegral.html}{exponential integral} is defined as
\begin{align}
\Ei(-x) \isDefinedAs -\binteg{z}{x}{\infty}\frac{\e^{-z}}{z} .
\end{align}}

\noindent
Now we have all the ingredients to calculate the restricted \ac{HF} energy for the doubly occupied $\sigma_g$ orbital.

\begin{subexercise}[resume]
\item Show that
\begin{align}\label{eq:ERHF-H2minBas}
E_{\text{RHF}}[(\sigma_g)^2] &= 2\brakket{\sigma_g}{\hat{h}}{\sigma_g} +
\cbraket{\sigma_g\sigma_g}{\sigma_g\sigma_g} + \frac{1}{R} .
\end{align}
\end{subexercise}

\noindent
The following exercises are easiest to do with some math program, e.g.\ \textsc{Mathematica}.

\begin{subexercise}[resume]
\item\label{ex:RHFexpo} Optimise the orbital exponent, $\zeta$, at each distance $R$ to minimise the energy. Plot the exponent as a function of the bonding distance. What do you notice in the dissociation limit?

\item\label{ex:RHFenergy} Plot also the optimised \acs*{RHF} energy (with optimised $\zeta$) as a function of $R$. Is the dissociation limit as you would expect?
\end{subexercise}

\end{exercise}

\subsection{Extra information on exercise~\ref{ex:RHF}}
\label{sec:extraIntegralInfo}
\ifbool{MullikenIntegApprox}{
The Mulliken approximation for the two-electron integrals can be avoided if one is willing to evaluate two additional integrals. The $\cbraket{AA}{AB}$ integral can be handled in the same manner as the $\cbraket{AA}{BB}$ integral and yields
\begin{align}
\cbraket{AA}{AB} = \frac{\zeta\e^{-\rho}}{16\rho}\Bigl(5 + 2\rho + 16\rho^2 - (5+2\rho)\e^{-2\rho}\Bigr) .
\end{align}}
The other integral turns out to be a nasty one, which probably can only be handled in the prolate spheroidal coordinate system. The Coulomb interaction needs to expressed in the Von Neumann expansion
\begin{multline}
\frac{1}{\abs{\vecr_1 - \vecr_2}}
= \frac{2}{R}\sum_{l=0}^{\infty}\;\sum_{\mathclap{m=-l}}^l\alpha_{ml}
\LegendreP^{\abs{m}}_l(\xi_<)\LegendreQ^{\abs{m}}_l(\xi_>) \times{} \\
\LegendreP^{\abs{m}}_l(\eta_1)\LegendreP^{\abs{m}}_l(\eta_2)\,\e^{\im m(\phi_1 - \phi_2)} ,
\end{multline}
where $\LegendreP^{\abs{m}}_l(x)$ and $\LegendreQ^{\abs{m}}_l(x)$ are the associated Legendre polynomials of the first and second kind respectively. The expansion coefficients are given as
\begin{align}
\alpha_{lm} = (-1)^m(2l+1)\left(\frac{(l - \abs{m})!}{(l + \abs{m})!}\right)^2
\end{align}
and $\xi_> \isDefinedAs \max(\xi_1,\xi_2)$ and $\xi_< \isDefinedAs \min(\xi_1,\xi_2)$. After a long massage, the integral yields
\ifbool{MullikenIntegApprox}{%
\begin{align}
\cbraket{AB}{AB}
&= \frac{\zeta\e^{-2\rho}}{120\rho}\Bigl[75\rho - 138\rho^2 - 72\rho^3 - 8\rho^4+ {} \notag \\*
&\eqspace\hphantom{\frac{\zeta\e^{-2\rho}}{120\rho}\Bigl[}
16\bigl(\gamma+\ln(\rho)\bigr)(9 + 18\rho + 15\rho^2 + 6\rho^3 + \rho^4) - {} \notag \\*
&\eqspace\hphantom{\frac{\zeta\e^{-2\rho}}{120\rho}\Bigl[}
32(9 - 3\rho^2 + \rho^4)\e^{2\rho}\Ei(-2\rho) + {} \notag \\*
&\eqspace\hphantom{\frac{\zeta\e^{-2\rho}}{120\rho}\Bigl[}
16(3 - 3\rho + \rho^2)^2\e^{4\rho}\Ei(-4\rho)\Bigr],
\end{align}
where the \href{http://mathworld.wolfram.com/Euler-MascheroniConstant.html}{Euler--Mascheroni constant} is defined as
\begin{align}
\gamma \isDefinedAs -\binteg{z}{0}{\infty}\e^{-z}\ln(z) = 0.577\,215\,664\,9\ldots
\end{align}
and the \href{http://mathworld.wolfram.com/ExponentialIntegral.html}{exponential integral} is defined as
\begin{align}
\Ei(-x) \isDefinedAs -\binteg{z}{x}{\infty}\frac{\e^{-z}}{z} .
\end{align}}{the expression in~\eqref{eq:cABAB}.}

\section{The \acs*{HF} \acf*{SCF} equations}
\label{sec:HF-SCF}

In the derivation of the \ac{HF} energy functional we assumed that the orbitals are orthonormal, so we need to respect this constraint. There are different ways to handle this constraint. We will use the method of Lagrange multipliers to do this. If you do not know about Lagrange multipliers or need a refresh of your memory, read Appendix~\ref{ap:Lagrangian}. Here are two exercises to practice your skill

\begin{exercise}
Find the point on the parabola $y(x) = \frac{1}{5}(x - 1)^2$ closest to the point $(x,y) = (1,2)$ in the Euclidean norm~\autocite[problem 12.18]{NocedalWright2006}. So consider the following minimisation problem
\begin{align}
&\min_{\mathclap{x,y \in \Reals}}&  &f(x,y) = (x-1)^2 + (y-2)^2 \notag \\*
&\text{subjet to }& &(x-1)^2 = 5y .
\end{align}

\begin{subexercise}
\item Construct the Lagrangian and find the points which satisfy the first order optimality conditions.

\item Which of these points are solutions?

\item It is tempting to eliminate the $(x - 1)^2$ term in $f$ to transform the problem into an unconstraint minimisation. Show that this procedure does \emph{not} yield the correct minimum.

\item Can you pinpoint the error we made by eliminating $(x - 1)^2$?
\end{subexercise}
\end{exercise}

\begin{exercise}
Find the maximum and minimum of $f(x,y,z) = 4y - 2z$ subject to the constraints $2x - y - z = 2$ and $x^2 + y^2 = 1$~\autocite[example 5]{PaulsNotesCalc3}.
\end{exercise}

\noindent
To enforce the orthonormality of the orbitals we introduce the following Lagrangian
\begin{align}\label{eq:HFlagrangian}
L[\{\phi_i,\psi_i^*\},\mat{\epsilon}]
&= E_{\text{HF}}[\{\phi_i,\psi_i^*\}] - E_{\text{nuc}} - \sum_{r,s=1}^N\epsilon_{rs}\bigl(\braket{\psi_i}{\phi_r} - \delta_{sr}\bigr) \\
& = \sum_{r=1}^N\brakket{\psi_r}{\hat{h}}{\phi_r} +
\half\sum_{\mathclap{r,s=1}}^N\braketx{\psi_r\psi_s}{\phi_r\phi_s} - 
\sum_{\mathclap{r,s=1}}^N\epsilon_{rs}\bigl(\braket{\psi_s}{\phi_r} - \delta_{sr}\bigr) . \notag 
\end{align}
Thanks to the Lagrange multipliers, we can now vary over $\phi_i$ and $\psi_i$ independently, but the constraint will enforce that we find $\psi_i = \phi_i$ and $\braket{\phi_i}{\phi_j} = \delta_{ij}$. Now consider small perturbations in the orbitals $\delta\phi_i$ and $\delta\psi^*_i$. Further note that the Lagrange multiplier matrix $\mat{\epsilon}$ needs to be hermitian, to ensure that the Lagrangian always yields real values. Since if $L$ would be complex valued, we cannot minimize it.\footnote{Though this is the standard derivation, if you study the Lagrangian~\eqref{eq:HFlagrangian} more carefully, you find that it is not enough to have $\mat{\epsilon} = \mat{\epsilon}^{\dagger}$ to ensure that the Lagrangian yields real numbers for independent choices of $\phi_i$ and $\psi_i$ in $L[\{\phi_i,\psi_i\},\mat{\epsilon}]$. The correct derivation is significantly more tedious, but leads to the same result.

In the correct derivation, one should also only take the real part of the \ac{HF} energy and also the diagonal elements, i.e.\ the constraint on the norm, should be replaced by only the real part of $\braket{\phi_i}{\psi_i}$, so the Lagrangian would become
\begin{align}\label{eq:HFlagrangianCorrect}
L[\{\phi_i,\psi_i^*\},\mat{\epsilon}]
= \Re E_{\text{HF}}[\{\phi_i,\psi_i^*\}] - 
\sum_{\mathclap{r \geq s = 1}}^N\bigl(\epsilon_{rs}\braket{\psi_s}{\phi_r} + \epsilon_{sr}\braket{\phi_r}{\psi_s} - 2\delta_{sr}\bigr) .
\end{align}
The orthogonality constraint will now enforce $\psi_i = c_i\phi_i$ and the normalization constraint gives only $\Re c_i = 1/\norm{\phi_i}$. This constraint makes the energy at the stationary point invariant under the choice of the norm of $\phi_i$, so we can choose $\norm{\phi_i} = 1$ as the most convenient one, and hence $\Re c_i = 1$. Further, the real part of the energy is invariant under arbitrary choices of $\Im c_i$, making the choice $\Im c_i = 0$ the most natural one.}

Now you can just set the first order functional derivatives w.r.t.\ $\phi_i(\vecx)$ and $\psi_i^*(\vecx)$ to zero. A slightly different but equivalent take on this is to consider the first order variation in the Lagrangian which are readily worked out as
\begin{align}
\delta L
&= \sum_{i=1}^N\brakket{\delta\psi_i}{\hat{h}}{\phi_i} + \half\sum_{\mathclap{i,j=1}}^N
\bigl(\braketx{\delta\psi_i\psi_j}{\phi_i\phi_j} + \braketx{\psi_i\delta\psi_j}{\phi_i\phi_j}\bigr) + {} \notag \\*
&\eqspace
\sum_{i=1}^N\brakket{\psi_i}{\hat{h}}{\delta\phi_i} + \half\sum_{\mathclap{i,j=1}}^N
\bigl(\braketx{\psi_i\psi_j}{\delta\phi_i\phi_j} + \braketx{\psi_i\psi_j}{\phi_i\delta\phi_j}\bigr) - {} 
\notag \\*
&\eqspace
\sum_{\mathclap{i,j=1}}^N\epsilon_{ji}
\bigl(\braket{\delta\psi_j}{\phi_i} + \braket{\psi_j}{\delta\phi_i}\bigr).
\end{align}
Collecting now all variations due to $\delta\psi^*_i$ and $\delta\phi_i$ respectively and also using that the constraint will lead to $\psi_i = \phi_i$, the first order variation in the Lagrangian can be expressed as
\begin{align}
\delta L
&= \sum_{i=1}^N\integ{\vecx}\left(\delta\phi_i^*(\vecx)\frac{\delta L}{\delta\phi^*_i(\vecx)} +
\frac{\delta L}{\delta\phi_i(\vecx)}\delta\phi_i(\vecx)\right) ,
\end{align}
where
\begin{subequations}\label{eq:HFfuncDphis}
\begin{align}\label{eq:HFfuncDphiStar}
\frac{\delta L}{\delta\phi^*_i(\vecx)}
&= \hat{h}\phi_i(\vecx) - \sum_{j=1}^N\phi_j(\vecx)\epsilon_{ji} + {} \notag \\*
&\eqspace
\sum_{j=1}^N\integ{\vecx'}\phi_j^*(\vecx')
\frac{\phi_j(\vecx')\phi_i(\vecx) - \phi_j(\vecx)\phi_i(\vecx')}{\abs{\vecr - \vecr'}}, \\
\frac{\delta L}{\delta\phi_i(\vecx)}
&= \hat{h}\phi^*_i(\vecx) - \sum_{j=1}^N\epsilon_{ij}\phi^*_j(\vecx) + {} \notag \\*
&\eqspace
\sum_{j=1}^N\integ{\vecx'}
\frac{\phi^*_i(\vecx)\phi^*_j(\vecx') - \phi^*_i(\vecx')\phi^*_j(\vecx)}{\abs{\vecr - \vecr'}}\phi_j(\vecx') .
\end{align}
\end{subequations}
So we recover the functional derivatives $\delta L / \delta\phi_i(\vecx)$ and $\delta L / \delta \phi_i^*(\vecx)$. Since the first order variation needs to vanish for any variation in $\phi_i$ and $\phi_i^*$, each of these derivatives needs to be zero. In the following exercise you show that these derivatives are each other complex conjugate, so one of the them is redundant.

\begin{exercise}
Show that $\delta L / \delta\phi_i(\vecx) = 0$ and $\delta L / \delta\phi^*_i(\vecx) = 0$ imply each other when using that $\mat{\epsilon}$ is hermitian, i.e.\ \(\mat{\epsilon} = \mat{\epsilon}^{\dagger}\). So one of them is enough to work with.
\end{exercise}

Now let us analyse the terms in the functional derivative and in particular the terms with the interaction. In the first interaction term you can recognise the classical Coulomb potential due to the interaction with the (electronic charge) density
\begin{align}\label{eq:HartreePotential}
\rho(\vecx) &= \sum_{i=1}^N\abs{\phi(\vecx)}^2 &
&\to&
v_{\text{H}}[\rho](\vecx) \isDefinedAs \integ{\vecx'}\frac{\rho(\vecx')}{\abs{\vecr - \vecr'}} .
\end{align}
Due to historical reasons, this term is called the Hartree term. Probably because Hartree~\autocite{Hartree1928} only had this interaction term in his equations, as he did not take the anti-symmetry properly into account. He started from only an orbital product.

The second part in the interaction terms in~\eqref{eq:HFfuncDphis} were recovered by Fock~\autocite{Fock1930}, so sometimes called the Fock part. He started correctly from a Slater determinant, so got this additional term which has no classical analogue. Since this term is caused by the permutation symmetry of the particles, this term is also called the exchange term. Note that as this term is caused by the permutation symmetry of the particles, the exchange term would have a plus sign for bosons.

Since the index $i$ also appears in the integral over $\vecx'$, the exchange potential can not be expressed as a simple local (multiplicative) potential. Instead, we need to express the exchange potential as an integral kernel
\begin{align}\label{eq:xIntKernel}
v_{\text{x}}[\{\phi_i\}](\vecx,\vecx') \isDefinedAs
-\sum_{j=1}^N\frac{\phi_j(\vecx)\phi^*_j(\vecx')}{\abs{\vecr - \vecr'}} .
\end{align}
The exchange part in~\eqref{eq:HFfuncDphiStar} can now be written as
\begin{align}
-\sum_{j=1}^N\integ{\vecx'}\phi_j^*(\vecx')\frac{\phi_j(\vecx)\phi_i(\vecx')}{\abs{\vecr - \vecr'}}
= \integ{\vecx'}v_{\text{x}}(\vecx,\vecx')\phi_i(\vecx')
\sAdenifeDsi \bigl(\hat{v}_{\text{x}}\phi_i\bigr)(\vecx) .
\end{align}
One-body operators which cannot be expressed as a simple multiplication are called non-local. An other example of a non-local one-body operator is the kinetic energy. Every non-local operator can be expressed as an integral kernel. For example, the integral kernel for the kinetic energy can be written as
\begin{align}\label{eq:kinKernel}
t(\vecx,\vecx') = \half\sum_r\mat{\nabla}\chi_r(\vecx)\cdot\mat{\nabla}\chi^*_r(\vecx') ,
\end{align}
where $\{\chi_r\}$ is an arbitrary complete orthonormal basis.
\mkbibfootnote{In case you are not scared of derivatives of delta distributions, you can also express the kernel as $
t(\vecx,\vecx') = \half\nabla_{\vecx'}\cdot\nabla_{\vecx}\delta(\vecx - \vecx')$.}

\begin{exercise}
Check that the kinetic energy integral kernel~\eqref{eq:kinKernel} indeed corresponds to the kinetic energy operator. So you need to check whether
\begin{align}
\integ{\vecx'}t(\vecx,\vecx')\psi(\vecx') = -\half\nabla^2\psi(\vecx) .
\end{align}
\end{exercise}

\noindent
Note that all local one-body operators can be regarded as a special case of a non-local operator with the help of the Dirac delta distribution. For example, the Coulomb potential can be written as
\begin{align}
v_{\text{H}}(\vecx,\vecx') = v_{\text{H}}(\vecx)\delta(\vecx - \vecx')
\end{align}
The non-local one-body potential are equivalent to matrices. When you act with a matrix (non-local operator) on a vector (function), you can get any vector back. However, when the matrix is diagonal (local potential), you get the same vector (function) back with scaled components (values in the points).

Now let us turn our attention back to functional derivatives of the Lagrangian. All terms coming from the \ac{HF} energy are collected in the Fock operator\index{Fock operator}
\begin{align}
\hat{f}[\{\phi_i\}] \isDefinedAs \hat{h} + \hat{v}_{\text{H}}[\{\phi_i\}] + \hat{v}_{\text{x}}[\{\phi_i\}] .
\end{align}
The functional derivative can be compactly expressed as
\begin{subequations}\label{eq:dLinFocks}
\begin{align}
\frac{\delta L}{\delta\phi^*_i(\vecx)}
&= \hat{f}\phi_i(\vecx) - \sum_{j=1}^N\phi_j(\vecx)\epsilon_{ji} , \\
\frac{\delta L}{\delta\phi_i(\vecx)}
&= \hat{f}^*\phi^*_i(\vecx) - \sum_{j=1}^N\epsilon_{ij}\phi^*_j(\vecx) .
\end{align}
\end{subequations}
Since the Lagrangian needs to be stationary with respect to any variation in the orbital, $\delta\phi_i^*(\vecx)$ and $\delta\phi_i(\vecx)$, both functional derivatives need to be zero. Just as if we were dealing with functions instead of functionals. So as first order optimality conditions, we obtain
\begin{align}\label{eq:GenFocks}
\hat{f}\phi_i(\vecx) &= \sum_{j=1}^N\phi_j(\vecx)\epsilon_{ji} ,	&
\hat{f}^*\phi^*_i(\vecx) &= \sum_{j=1}^N\epsilon_{ij}\phi^*_j(\vecx) .
\end{align}
These \ac{HF} equations look almost like Schrödinger equations for the \ac{HF} orbitals, except for $\mat{\epsilon}$-matrix on the right-hand side. We would rather like it to be diagonal.

To show that $\mat{\epsilon}$ can be made diagonal, we multiply~\eqref{eq:GenFocks} $\phi_i(\vecx)$ and $\phi_j(\vecx)$ respectively and integrating over $\vecx$, we find
\begin{align}
\brakket{\phi_j}{\hat{f}}{\phi_i} &= \epsilon_{ji} , &
\brakket{\phi_j}{\hat{f}}{\phi_i}^* &= \epsilon_{ij} ,
\end{align}
so we find that the Lagrange multiplier matrix is hermitian, $\mat{\epsilon} = \mat{\epsilon}^{\dagger}$.
We used here that at the solution point, we satisfy also the constraint $\braket{\phi_i}{\phi_j} = \delta_{ij}$. Combing now these two equations, we have
\begin{equation}
\epsilon_{ij} = \brakket{\phi_j}{\hat{f}}{\phi_i}^* = \epsilon_{ji}^* .
\end{equation}
Additionally, one can show that the \ac{HF} wavefunction does not change if we make a unitary transformation among the \ac{HF} orbitals, except for an irrelevant overall phase factor. In particular, the \ac{HF} energy will not change under such a unitary transformation and we are allowed to diagonalise $\mat{\epsilon}$. This yields the canonical \ac{HF} equations\index{canonical Hartree--Fock}
\begin{align}\label{eq:canonHF}
\hat{f}[\{\phi\}]\,\phi_i(\vecx) = \varepsilon_i\,\phi_i(\vecx) ,
\end{align}
where the eigenvalues of the Lagrange multiplier matrix, $\varepsilon_i$, are called \ac{HF} orbital energies. The \ac{HF} orbitals that also diagonalise the Fock operator are called the canonical \ac{HF} orbitals.

\begin{exercise}
Show that the \ac{HF} wavefunction remains the same up to an overall phase factor, when making a unitary transformation among the \ac{HF} orbitals. \\
Hint: The \ac{HF} wavefunction is $\Phi_{\text{HF}} = \ifbool{normSlatDet}{(N!)^{-\nhalf}}{}\det(\mat{\Phi})$, where $\mat{\Phi}$ is the matrix composed of all \ac{HF} orbitals at all possible coordinates.
Next show that the \ac{HF} wavefunction with the transformed orbitals can be written as $\Phi'_{\text{HF}} = \det(\mat{\Phi} \mat{U})$, where the transformation matrix $\mat{U}$ is defined as
\begin{align*}
\phi'_i(\vecx) = \sum_{j=1}^N\phi_j(\vecx)U_{ji} .
\end{align*}
\end{exercise}

\begin{exercise}
Check that both $\delta L / \delta\phi_i(\vecx) = 0$ as well as $\delta L / \delta\phi^*_i(\vecx) = 0$ indeed reduce to the expressions in terms of the Fock operators~\eqref{eq:dLinFocks}.
\end{exercise}

\begin{exercise}[\acf*{UHF} for H$_2$ in a minimal basis]
\label{ex:UHFminBasis}
In exercise~\ref{ex:RHF} we assumed that the \ac{HF} solution would have the same symmetry as the system itself, so we used symmetry adapted orbitals as trial orbitals. We argued that both electrons should occupy the $\sigma_g$ orbital to get the lowest energy. Since the Fock operator is \emph{not} linear, the solutions do not necessarily exhibit the symmetry of the system. Since electrons repel each other, we might be able to lower the \ac{HF} energy by allowing the electrons to localise in their own orbital. One with spin up and the other with spin down. For the two orbital (minimal basis) model for H$_2$ we will first assume that they localise in the atomic orbitals, so the orbitals would be $\chi_A(\vecr)\alpha(s)$ and $\chi_B(\vecr)\beta(s)$. Note that $\chi_A(\vecr)$ and $\chi_B(\vecr)$ are not orthogonal. For later convenience, we will perturb them to be orthogonal.

\begin{subexercise}
\item Explain why it is not necessary that $\braket{\chi_A}{\chi_B} \neq 0$.

\item Given the overlap matrix, $S_{ij} = \braket{\chi_i}{\chi_j}$, show that the orbitals
\begin{align}
\phi_i(\vecx) = \sum_j\chi_j(\vecx)S_{ji}^{-\nhalf}
\end{align}
are orthonormal. This method to generate an orthonormal basis from a non-orthonormal one is called Löwdin orthonormalisation~\autocite{Lowdin1950}.

\item In the minimal H$_2$ basis, the overlap matrix is of the form
\begin{align}
\mat{S} = \begin{pmatrix} 1 & s \\ s & 1\end{pmatrix} .
\end{align}
Calculate the inverse, $\mat{S}^{-1}$.

\item The inverse square root of the overlap matrix is
\begin{align}
\mat{S}^{-\nhalf} = \frac{1}{\sqrt{1 - s^2}}\begin{pmatrix} \cos\theta & -\sin\theta \\
-\sin\theta & \cos\theta \end{pmatrix} ,
\end{align}
where $s \sAdenifeDsi \sin(2\theta)$. Check this by showing that $\bigl(\mat{S}^{-\nhalf}\bigr)^2 = \mat{S}^{-1}$.

\item Show that
\begin{align}
\phi_{a/b}(\vecr) = \frac{1}{\sqrt{1 - s^2}}
\bigl(\cos(\theta)\chi_{A/B}(\vecr) - \sin(\theta)\chi_{B/A}(\vecr)\bigr) .
\end{align}
\end{subexercise}

\noindent
Now that we have orthogonalized the orbitals, we need to transform the integrals to the new basis.

\begin{subexercise}[resume]
\item Show that
\begin{subequations}
\begin{align}
\brakket{\phi_a}{\hat{h}}{\phi_a} = \brakket{\phi_b}{\hat{h}}{\phi_b}
&= \frac{\brakket{\chi_A}{\hat{h}}{\chi_A} - s\brakket{\chi_A}{\hat{h}}{\chi_B}}{1 - s^2} , \\
\brakket{\phi_a}{\hat{h}}{\phi_b} = \brakket{\phi_a}{\hat{h}}{\phi_b}
&= \frac{\brakket{\chi_A}{\hat{h}}{\chi_B} - s\brakket{\chi_A}{\hat{h}}{\chi_A}}{1 - s^2} .
\end{align}
\end{subequations}

\item Show that the unique two-electron integrals can be expressed as
\begin{subequations}
\begin{align}
\cbraket{aa}{aa} &= \frac{1}{(1-s^2)^2}\biggl[
\biggl(1 - \frac{s^2}{2}\biggr)\cbraket{AA}{AA} + \frac{s^2}{2}\cbraket{AA}{BB} - {} \notag \\*
&\eqspace\hphantom{\frac{1}{(1-s^2)^2}\biggl[}
2s\cbraket{AA}{AB} + s^2\cbraket{AB}{AB}\biggr] , \\
\cbraket{aa}{bb} &= \frac{1}{(1-s^2)^2}\biggl[
\biggl(1 - \frac{s^2}{2}\biggr)\cbraket{AA}{BB} + \frac{s^2}{2}\cbraket{AA}{AA} - {} \notag \\*
&\eqspace\hphantom{\frac{1}{(1-s^2)^2}\biggl[}
2s\cbraket{AA}{AB} + s^2\cbraket{AB}{AB}\biggr] , \\
\cbraket{aa}{ab} &= \frac{1}{(1-s^2)^2}\biggl[
(1 + s^2)\cbraket{AA}{AB} - s\cbraket{AB}{AB} - {} \notag \\*
&\eqspace\hphantom{\frac{1}{(1-s^2)^2}\biggl[}
\frac{s}{2}\bigl(\cbraket{AA}{AA} + \cbraket{AA}{BB}\bigr)\biggr] , \\
\cbraket{ab}{ab} &= \frac{1}{(1-s^2)^2}\biggl[
\cbraket{AB}{AB} - 2s\cbraket{AA}{AB} - {} \notag \\*
&\eqspace\hphantom{\frac{1}{(1-s^2)^2}\biggl[}
\frac{s^2}{2}\bigl(\cbraket{AA}{AA} + \cbraket{AA}{BB}\bigr)\biggr] .
\end{align}
\end{subequations}

\ifbool{MullikenIntegApprox}{
\item Show that in the Mulliken approximation, the integrals reduce to
\begin{subequations}
\begin{align}
\cbraket{aa}{aa}_M &= \cbraket{AA}{AA} +
\frac{1}{2(1 - s^2)^2}\bigl[\cbraket{AA}{AA} - \cbraket{AA}{BB}\bigr] , \\
\cbraket{aa}{bb}_M &= \cbraket{AA}{BB} +
\frac{1}{2(1 - s^2)^2}\bigl[\cbraket{AA}{BB} - \cbraket{AA}{AA}\bigr] , \\
\cbraket{aa}{ab}_M &= \cbraket{ab}{ab}_M = 0 .
\end{align}
\end{subequations}}

\end{subexercise}

\noindent
Now we put one electron with spin up in $\phi_a$ and the other electron with spin down in $\phi_b$.

\begin{subexercise}[resume]
\item\label{ex:UHFlocEnergy} Show that the \ac{HF} energy becomes
\begin{align}
E^{\text{loc}}_{\text{HF}}
&= 2\brakket{\phi_a}{\hat{h}}{\phi_a} + \cbraket{\phi_a\phi_a}{\phi_b\phi_b} + \frac{1}{R} .
\end{align}

\item Using the integrals you already calculated in exercise~\ref{ex:RHF}, optimise the exponent of the orbitals for this localised \ac{HF} solution, e.g.\ with \textsc{Mathematica}. Compare to the exponent in the \acs*{RHF} exercise~\ref{ex:RHFexpo}.

\item Plot the \ac{HF} energy with the localised orbitals. Compare with the \acs*{RHF} energy of~\ref{ex:RHFenergy}. What do you notice?

\item Put now two spin up electrons in the localised orbitals and work out the \ac{HF} energy \ifbool{MullikenIntegApprox}{in the Mulliken approximation }and compare with~\ref{ex:UHFlocEnergy} in the dissociation limit.

\item Show that occupying $\phi_a$ and $\phi_b$ with two spin up electrons yields in the dissociation limit the same \ac{HF} wavefunction (up to a sign) as occupying the $\sigma_g$ and $\sigma_u$ orbitals with two spin up electrons.
\end{subexercise}

\noindent
So far we assumed that the \ac{HF} orbitals are either localised or fully delocalised (symmetry adapted). In a fully \acf*{UHF} calculation, the \ac{HF} orbitals can also be a mixture between these two extremes, i.e.\ we need to consider linear combinations. Since the final charge density will be symmetric, the linear combinations are restricted to the form
\begin{align}
\psi(\vecr)_{a/b} = \cos(\alpha)\phi_{a/b}(\vecr) + \sin(\alpha)\phi_{b/a}(\vecr) .
\end{align}
In principle, we should now solve the \ac{HF} equations. However, since there is only one additional parameters, it is easier to write the energy as a function of $\alpha$ as well and to optimise.

\begin{subexercise}[resume]
\item Show that the \ac*{UHF} energy with the orbitals $\psi_{a/b}$ can be expressed as
\begin{align}
E_{\text{UHF}}
&= 2\bigl(\brakket{a}{\hat{h}}{a} + \sin(2\alpha)\brakket{a}{\hat{h}}{b}\bigr) + {} \notag \\*
&\eqspace
\cbraket{aa}{bb} +
\frac{\sin^2(2\alpha)}{2}\bigl(\cbraket{aa}{aa} - \cbraket{aa}{bb}\bigr) + \frac{1}{R} .
\end{align}

\item Optimise (numerically) both the orbital exponents, $\zeta$, and the orbital mixing angle, $\alpha$. Plot $\alpha$ as a function of the internuclear distance, $R$. What do you notice?

\item Plot the fully optimised \acs*{UHF} energy and compare to the previous results (\acs*{RHF} from exercise~\ref{ex:RHFenergy} and the localised solution~\ref{ex:UHFlocEnergy}).
\end{subexercise}

\end{exercise}

\begin{exercise}
The kernel of the spin-density operator can be written as
\begin{align}
\rho(\vecx,\vecx') = \sum_{i=1}^N\delta(\vecx - \vecx') .
\end{align}

\begin{subexercise}
\item Show that this expression is correct. That means, you need to show that it gives the expected expression for the expectation value of the spin-density of a \emph{general} wavefunction.

\item Show that the expectation value of the spin-density simplifies to
\begin{align}
\rho(\vecx) = \sum_{i=1}^N\abs{\phi_i(\vecx)}^2,
\end{align}
if the wavefunction is a Slater determinant.
\end{subexercise}
\end{exercise}

\section{Roothaan--Hall equations}
The canonical \ac{HF} equations derived before are differential equations which need to be solved self-consistently. As with the many-body Schrödinger equation, it is basically impossible to do this for general $\phi_i(\vecx)$. Hence, we resort to the same approach which leads to full \ac{CI}: expand the orbitals in a basis
\begin{align}\label{eq:HFbasisExpansion}
\phi_i(\vecx) = \sum_{\nu=1}^m\chi_{\nu}(\vecx)C_{\nu i} ,
\end{align}
where $m$ denotes the number of functions in our basis. This is exactly what Roothaan~\autocite{Roothaan1951} and Hall~\autocite{Hall1951} did independently in 1951. Simply insert the expansion~\eqref{eq:HFbasisExpansion} in the \ac{HF} equations~\eqref{eq:canonHF}, multiply from the left by $\chi^*_{\mu}(\vecx)$ and integrate over $\vecx$
\begin{align}
\sum_{\nu=1}^m\underbrace{\brakket{\chi_{\mu}}{\hat{f}}{\chi_{\nu}}}_{\sAdenifeDsi f_{\mu\nu}}
C_{\nu i}
= \sum_{\nu}\underbrace{\braket{\chi_{\mu}}{\chi_{\nu}}}_{= S_{\mu\nu}}C_{\nu i}\,\varepsilon_i ,
\end{align}
where we did not assume that the basis $\{\chi\}$ is orthonormal. So you see that the general Roothaan--Hall equations have the form of a generalised eigenvalue equation
\begin{align}\label{eq:RHgen}
\mat{f}\,\mat{C} = \mat{S}\,\mat{C}\diag(\mat{\varepsilon}) ,
\end{align}
where $\diag(\mat{\varepsilon})$ stands for a diagonal matrix with $\mat{\varepsilon}$ on its diagonal. In the case of an orthonormal basis, the Roothaan--Hall equations reduce to an ordinary eigenvalue equation
\begin{align}\label{eq:RHortho}
\mat{f}\,\mat{C} = \mat{C}\diag(\mat{\varepsilon}) .
\end{align}
So starting from some non-orthonormal (atomic) basis, one can either use the general form~\eqref{eq:RHgen}, or first orthonormalise the basis and use the simpler form~\eqref{eq:RHortho}.

The Fock operator itself depends on the \ac{HF} orbitals, so we should express it also in terms of the expansion coefficients
\begin{align}
f_{\mu\nu}[\mat{C}] \isDefinedAs \brakket{\chi_{\mu}}{\hat{f}[\mat{C}]}{\chi_{\nu}}
= \brakket{\chi_{\mu}}{\hat{h} + \hat{v}_{\text{H}}[\mat{C}] + \hat{v}_{\text{x}}[\mat{C}]}{\chi_{\nu}} .
\end{align}
The one-body part of the Hamiltonian, $\hat{h}$, does not depend on the \ac{HF} orbitals, so nothing to be done for that term. For the Hartree (classical Coulomb) potential, we need the spin-density. The spin-density in terms of the expansion coefficients becomes
\begin{align}\label{eq:rhoExp}
\rho(\vecx) = \sum_{i=1}^N\abs{\phi_i(\vecx)}^2
= \sum_{\mu,\nu}\chi_{\mu}(\vecx)\underbrace{\sum_{i=1}^N
C^{\vphantom{\dagger}}_{\mu i}C^{\dagger}_{i\nu}}_{\sAdenifeDsi \oneHF_{\mu\nu}}\chi_{\nu}^*(\vecx) .
\end{align}
The quantity $\oneHF_{\mu\nu}$ is called the \ac{HF} \acf*{1RDM}, or `density matrix' for short if only used in the context of \ac{HF}. Let us consider some properties of the \ac{HF} \acs*{1RDM}. As the spin-density integrates to the total number of electrons, we have
\begin{align}
N &= \sum_{i=1}^N\braket{\phi_i}{\phi_i}
= \sum_{\mu,\nu}\sum_{i=1}^NC^{\dagger}_{i\nu}\braket{\chi_{\mu}}{\chi_{\nu}}
C^{\vphantom{\dagger}}_{\mu i} \notag \\*
&= \sum_{\mu,\nu}\oneHF_{\mu\nu}S_{\nu\mu} = \Trace\{\mat{\oneHF}\mat{S}\} .
\end{align}
If we use an atomic basis set, each basis function belongs to a certain atom, so the trace can be partitioned into atomic contributions as
\begin{align}
\Trace\{\mat{\oneHF}\mat{S}\}
= \sum_A\sum_{\mu \in A}(\mat{\oneHF}\mat{S})_{\mu\mu} .
\end{align}
The combined contribution per atom is called the Mulliken population~\autocite{Mulliken1955a, Mulliken1955b}
\begin{align}
\rho_A^{\text{Mulliken}} = \sum_{\mu \in A}(\mat{\oneHF}\mat{S})_{\mu\mu} .
\end{align}
The Mulliken charge is now obtained by adding the nuclear charge
\begin{align}
Q_A^{\text{Mulliken}} = Z_A - \rho_A^{\text{Mulliken}} .
\end{align}
Note that this decomposition only makes sense for basis functions which are well localised on the individual atoms. Large basis sets with diffuse functions, this decomposition scheme breaks down. There are more sophisticated decomposition schemes which do not suffer from this, e.g.\ natural population analysis~\autocite{ReedWeinstockWeinhold1985}, Bader's atoms in molecules~\autocite{Bader1990} and Vornoi deformation density~\autocite{Voronoi1908, BickelhauptEikema-HommesFonseca-Guerra1996, Fonseca-GuerraHandgraafBaerends2004}. These more sophisticated decomposition schemes can also used in combination other basis sets, such as plane waves. The Mulliken decomposition is easy to generalise to more fine-grained populations, e.g.\ separate $s$, $p$, $d$, etc.\ contributions. The prime advantage of the Mulliken analysis is its low computational cost and simplicity to implement.

Now let us turn our attention back to the Hartree potential. Inserting the expansion for the density~\eqref{eq:rhoExp} in the Hartree potential~\eqref{eq:HartreePotential}, we find
\begin{align}
v_{\text{H}}(\vecx)
= \integ{\vecx'}\frac{\rho(\vecx')}{\abs{\vecr - \vecr'}}
= \sum_{\mu,\nu}\oneHF_{\mu\nu}\integ{\vecx'}
\frac{\chi_{\mu}(\vecx')\chi^*_{\nu}(\vecx')}{\abs{\vecr - \vecr'}} ,
\end{align}
or in terms of its matrix elements
\begin{align}
\brakket{\chi_{\mu}}{\hat{v}_{\text{H}}}{\chi_{\nu}}
= \sum_{i=1}^N\braket{\chi_{\mu}\phi_i}{\chi_{\nu}\phi_i}
= \sum_{\rho\sigma}\oneHF_{\sigma\rho}\braket{\chi_{\mu}\chi_{\rho}}{\chi_{\nu}\chi_{\sigma}} .
\end{align}
For the exchange potential~\eqref{eq:xIntKernel} we need a `density with two coordinates'
\begin{align}
\oneHF(\vecx,\vecx') = \sum_{i=1}^N\phi_i(\vecx)\phi_i^*(\vecx')
= \sum_{\mu,\nu}\chi_{\mu}(\vecx)\oneHF_{\mu\nu}\chi^*_{\nu}(\vecx') ,
\end{align}
which is the \ac{HF} \acs*{1RDM} in coordinate representation. Note that the spin-density is just the diagonal of the \acs*{1RDM}, $\rho(\vecx) = \oneHF(\vecx,\vecx)$. The exchange kernel now becomes
\begin{align}
v_{\text{x}}(\vecx,\vecx') = -\frac{\oneHF(\vecx,\vecx')}{\abs{\vecr - \vecr'}}
= -\sum_{\mu,\nu}\chi_{\mu}(\vecx)\frac{\oneHF_{\mu\nu}}{\abs{\vecr - \vecr'}}\chi^*_{\nu}(\vecx') ,
\end{align}
or in its matrix representation
\begin{align}
\brakket{\chi_{\mu}}{\hat{v}_{\text{x}}}{\chi_{\nu}}
= -\sum_{\rho,\sigma}\oneHF_{\sigma\rho}\braket{\mu\rho}{\sigma\nu} .
\end{align}
We have now seen the \ac{HF} \acs*{1RDM} in two representations: in the space-spin coordinate representation and 2) in and arbitrary basis ($\{\chi_i\}$) representation. In the canonical \ac{HF} basis the \ac{HF} \acs*{1RDM} becomes particularly simple as it is diagonal
\begin{align}\label{eq:HF1RDMinHFbasis}
\oneHF_{rs}
= \brakket{\phi_r}{\opOneHF{\hat}}{\phi_s}
= \begin{cases*}
1	&if $r = s \leq N$ \\
0	&otherwise.
\end{cases*}
\end{align}
The \ac{HF} \acs*{1RDM} is therefore idempotent, i.e.
\begin{align}\label{eq:1RDMidempotent}
\opOneHF{\mat}\mat{S}\opOneHF{\mat} = \opOneHF{\mat} .
\end{align}

\begin{exercise}\leavevmode
\begin{subexercise}
\item Prove that an idempotent matrix $\mat{M}$ can only have 0 and 1 as its eigenvalues. Assume that an orthonormal basis is used, so $\mat{S} = \mat{1}$.

\item Show that if the \ac{HF} \acs*{1RDM} is obtained in a non-orthogonal basis, so $\mat{C}^{\dagger}\mat{S}\mat{C} = \mat{1}$, that the idempotency property changes to~\eqref{eq:1RDMidempotent}. Use that the \ac{HF} \acs*{1RDM} in a general basis can be written as, $\opOneHF{\mat} = \mat{C}\mat{\rho}\mat{C}^{\dagger}$, where $\mat{\rho}$ is the \acs*{1RDM} in the canonical \ac{HF} basis~\eqref{eq:HF1RDMinHFbasis}.

\end{subexercise}
\end{exercise}

\noindent
Since the Roothaan--Hall equations have the form of a(n) (generalised) eigenvalue equation, an algorithm to find the optimal \ac{HF} solution would be to start with an initial guess for the \ac{HF} orbitals, construct the Fock matrix and diagonalise it. Then select the orbitals with the lowest orbital energies (eigenvalues) to construct a new \ac{HF} \acs*{1RDM}. That selecting the orbitals with the lowest orbital energies leads to the lowest \ac{HF} energy is called the \emph{aufbau principle}. For the completely unrestricted form described here, we will proof that the aufbau principle always works and leads to the lowest \ac{HF} energy. However, as \ac{HF} is usually implemented with additional restriction, that proof does not apply anymore. One needs to generalise the aufbau principle to handle degenerate orbital energies as well. This generalisation and its proof are beyond the scope of this course, but can be found in Ref.~\autocite{GiesbertzBaerends2010}.

The $N$ orbitals used to construct the new \ac{HF} determinant / \acs*{1RDM} are called the occupied orbitals. The other orbitals obtained from diagonalising the Fock matrix are called the unoccupied\slash{}empty\slash{}virtual orbitals.

\begin{exercise}
Make a diagram of the \ac{SCF} procedure to solve the \ac{HF} equations.
\end{exercise}

\section{Properties and specialities of the \acs*{HF} system}

\subsection{Brillouin's theorem (1934)}
Consider a singly excited determinant. To be more precise, with a singly excited determinant, $\Phi_i^a$, we mean a determinant where an occupied orbital $\phi_i(\vecx)$ is replaced by a virtual orbital, $\phi_a(\vecx)$. Brillouin's theorem states
\begin{align}\label{eq:Brillouin}
\brakket{\Phi_0}{\hat{H}}{\Phi_i^a} = 0 ,
\end{align}
where $\Phi_0$ is the unperturbed \ac{HF} determinant.
This result will be useful later in the course, when we add additional determinants to the \ac{HF} determinant to improve our approximation of the wavefunction.

\begin{exercise}
Prove Brillouin's theorem~\eqref{eq:Brillouin}. Use the Slater--Condon rules to work out the left-hand side of~\eqref{eq:Brillouin} and relate the result to the Fock matrix.
\end{exercise}

\subsection{The \acs*{HF} Hamiltonian}
Since the \ac{HF} orbitals are eigenfunctions of the Fock operator, the \ac{HF} wavefunction is an eigenfunction of the following Hamiltonian
\begin{align}
\hat{H}^{(0)} = \sum_{i=1}^N\hat{f}(\vecx_i) .
\end{align}
Note that $E_{\text{HF}} \neq \brakket{\Phi_{\text{HF}}}{\hat{H}^{(0)}}{\Phi_{\text{HF}}}$, but instead
\begin{align}
E_{\text{HF}} &= \brakket{\Phi_{\text{HF}}}{\hat{H}}{\Phi_{\text{HF}}} \\*
&= \brakket{\Phi_{\text{HF}}}{\hat{H}^{(0)}}{\Phi_{\text{HF}}} +
\brakket{\Phi_{\text{HF}}}{\hat{H} - \hat{H}^{(0)}}{\Phi_{\text{HF}}} = E^{(1)} . \notag
\end{align}
The \ac{HF} energy can therefore be regarded as the first order corrected energy, use $\hat{H}^{(0)}$ as a zeroth order Hamiltonian. It is obvious that this perturbation expansion can be pushed to higher orders. This is called \ac{MP} perturbation theory and will be explained later in the the course in more detail. One could call \ac{HF} `\ac{MP}1': first order corrected perturbation theory.

\newpage
\begin{exercise} \leavevmode
\begin{subexercise}
\item Show that
\begin{align}
\hat{H}^{(0)}\ket{\Phi_{\text{HF}}} &= E^{(0)}\ket{\Phi_{\text{HF}}} &
&\text{where} &
E^{(0)} &= \sum_{i=1}^N\epsilon_i .
\end{align}

\item Show that
\begin{align}
E_{\text{HF}} = \sum_{i=1}^N\epsilon_i - \half\sum_{i,j=1}^N\braketx{ij}{ij}
= \half\sum_{i=1}^N\bigl(\brakket{i}{\hat{h}}{i} + \epsilon_i\bigr) .
\end{align}
\end{subexercise}
\end{exercise}

\subsection{Koopmans' theorem (1934)}
This theorem on the interpretation of the \ac{HF} solution was published by Tjalling Koopmans~\autocite{Koopmans1934}. He received a Nobel prize in \emph{economics} in 1975. Koopmans' theorem states that the occupied \ac{HF} orbital energies can be regarded as approximations to ionisation energies\slash{}potentials and that the unoccupied ones serve as approximations to affinities. The assumption is that the orbitals do not relax when an electron is removed from or added to the system. Under this assumption one can readily show that
\begin{subequations}
\begin{align}
\text{IP}^{\text{HF}}_i &= E^{N-1}_{\text{HF}} - E^N_{\text{HF}} \approx -\epsilon_i , \\
\text{EA}^{\text{HF}}_a &= E^N_{\text{HF}} - E^{N+1}_{\text{HF}} \approx -\epsilon_a .
\end{align}
\end{subequations}
Apart from the intrinsic approximations in \ac{HF}, the approximation of no relaxation introduces an error of several eVs (hundreds \unitfrac{kcal}{mol}). Since ionisation energies are typically quite large, the relative error is not too large and $-\epsilon_i$ often gets you in the right ball park. On the contrary, affinities are typically small, which renders $-\epsilon_a$ as an approximation to affinities practically useless. The \ac{HF} virtuals typically have positive orbital energies, so Koopmans' theorem predicts many negatively charged ions not to be stable.

\begin{exercise}
Proof Koopmans' theorem.

\textbf{Hint:} As mentioned, the assumption here is that the HF orbitals do not change when one electron is removed or added, so you can build the HF determinant from the same set of orbitals as the $N$-particle system. In other words, to approximate the ionised \ac{HF} state, you just remove a (single) \acs{HOMO} from the Slater determinant. To approximate the $N+1$ HF state, you add the \acs{LUMO} to the Slater determinant.
\end{exercise}

\subsection{Restricted \acs*{HF}}
Often we are interested in systems for which the ground state is a singlet $\av{\hat{S}^2} = 0$, so we will have an equal amount of electrons in spin up and spin down orbitals. One can show that a single Slater determinant can only be a singlet, if we have a closed shell solution. Closed shell means that the spin up and spin down orbitals span the same spatial function space. This is easiest to implement as both having the same spatial part
\begin{align}
\phi_i(\vecx) = \begin{cases*}
\psi_{(i+1)/2}(\vecr)\alpha(\sigma)	&for $i$ odd, \\
\psi_{i/2}(\vecr)\beta(\sigma)		&for $i$ even .
\end{cases*}
\end{align}
Since the only degree of freedom is now the spatial part of the \ac{HF} orbitals, we can integrate out the spin part from the \ac{HF} expressions and only $N/2$ spatial \ac{HF} orbitals remain to be determined. This form of \ac{HF} is called \ac{RHF}. The \ac{RHF} energy becomes
\begin{align}\label{eq:ERHF}
E^{\text{RHF}} = 2\sum_{i=1}^{N/2}\brakket{\psi_i}{\hat{h}}{\psi_i} +
\sum_{i,j=1}^{N/2}\bigl(2\braket{\psi_i\psi_j}{\psi_i\psi_j} - \braket{\psi_i\psi_j}{\psi_j\psi_i}\bigr) .
\end{align}
The \ac{RHF} \ac{SCF} equations retain the same basic form
\begin{align}
\epsilon_i\psi_i(\vecr) = \hat{f}\psi_i(\vecr)
= \bigl(\hat{h} + \hat{v}_{\text{H}} + \hat{v}_{\text{x}}\bigr)\psi_i(\vecr) ,
\end{align}
though the Hartree and exchange potential are now
\begin{subequations}\label{eq:restrictedPots}
\begin{align}
v_{\text{H}}(\vecr) &= \integ{\vecr'}\frac{\rho(\vecr')}{\abs{\vecr - \vecr'}} , \\
v_{\text{x}}(\vecr,\vecr') &= -\half\frac{\oneHF(\vecr,\vecr')}{\abs{\vecr - \vecr'}} ,
\end{align}
\end{subequations}
where the spin-integrated \ac{HF} density and \ac{1RDM} are defined as
\begin{subequations} 
\begin{align}
\oneHF(\vecr,\vecr') &= \sum_{\sigma}\oneHF(\vecr\sigma,\vecr'\sigma)
= 2\sum_{i=1}^{N/2}\psi_i(\vecr)\psi^*_i(\vecr') , \\
\rho(\vecr) &= \sum_{\sigma}\rho(\vecr\sigma) = 2\sum_{i=1}^{N/2}\abs{\psi_i(\vecr)}^2
= \oneHF(\vecr,\vecr) .
\end{align}
\end{subequations}
As one is most of the time interested in closed shell systems, \ac{RHF} is the most used form of \ac{HF}. In the case of open shell there exist many variants with varying restrictions. If one uses the same spatial parts for the spin-up and spin-down orbitals, one calls this \ac{ROHF}. If one only fixes the number of occupied spin-up and spin-down orbitals (the $S_z$ value) and allows for different spatial parts of the spin-up and spin-down orbitals, the method is called \acf{UHF}. The \ac{HF} as introduced in the beginning of this chapter is even less restrictive, as it only fixes the number of electrons and finds itself the optimal distribution between the number of spin-up and spin-down electrons. One could call this completely unrestricted \ac{HF}. Typically, for small $S_z$ the \ac{UHF} and completely unrestricted \ac{HF} solutions coincide. For $S_z = 0$ both often yield the same solution as the \ac{RHF} method for organic molecules in close to their equilibrium geometry.

\begin{exercise}\leavevmode
\begin{subexercise}
\item Derive the \ac{RHF} energy expression~\eqref{eq:ERHF} .

\item Derive the \ac{RHF} equations, i.e.\ the \ac{RHF} expression for the the Fock operator and the corresponding potentials~\eqref{eq:restrictedPots}.
\end{subexercise}
\end{exercise}

\begin{exercise}\leavevmode
\begin{subexercise}
\item Consider a Slater determinant with one spin up and one spin down orbital
$\phi_1(\vecx) = \psi_1(\vecr)\alpha(s)$ and $\phi_2(\vecx) = \psi_2(\vecr)\beta(s)$. Show that this determinant is only an eigenfunction of $\hat{S}^2$ if $\psi_1(\vecr) = \psi_2(\vecr)$.

\textbf{Hint:} Write the total spin operator as
\begin{align*}
\hat{S}^2 &= \mat{\hat{S}}\cdot\mat{\hat{S}} = \sum_{i,j}\mat{S}(\sigma_i)\cdot\mat{S}(\sigma_j) \\
&= \sum_{i,j}\bigl[\hat{S}_x(\sigma_i)\hat{S}_x(\sigma_j) + \hat{S}_y(\sigma_i)\hat{S}_y(\sigma_j) +
\hat{S}_z(\sigma_i)\hat{S}_z(\sigma_j)\bigr] \\
&= \sum_{i,j}\biggl[\half\Bigl(\hat{S}_+(\sigma_i)\hat{S}_-(\sigma_j) +
\hat{S}_-(\sigma_i)\hat{S}_+(\sigma_j)\Bigr) +
\hat{S}_z(\sigma_i)\hat{S}_z(\sigma_j)\biggr] .
\end{align*}

\item \textbf{Challenge:} Show that a Slater determinant is only a singlet state, if the spin-up and spin-down parts span the same volume (determinant). The restricted solution is a particular realisation of this, since we can make arbitrary unitary transformations between functions building up the determinant without affecting its absolute value. 

\textbf{Hint:}
So you need to work out $\hat{S}^2\ket{\Phi}$, where $\Phi$ is a general determinant. First note that one can only have singlet state if the amount of spin-up and spin-down orbitals are the same.
Second, note that the $\hat{S}^2$ does not really care about the spatial part, but only about spin. So it is convenient to group the spin-up and spin-down parts and write the determinant in a more abstract manner as
\begin{align*}
\ket{\Phi} = \ket{\psi_1\alpha, \dotsc, \psi_{N/2}\alpha, \psi_{N/2+1}\beta, \dotsc, \psi_N\beta} .
\end{align*}
Now work out $\hat{S}^2\ket{\Phi}$ with the help of $\hat{S}^2 = \half\bigl(\hat{S}_+\hat{S}_- + \hat{S}_-\hat{S}_+\bigr) + \hat{S}_z^2$ and draw the conclusions.

\item Show that the \ac{UHF} solution for H$_2$ in the dissociation limit is half singlet and half triplet. You can use the \ac{UHF} dissociation limit found in exercise~\ref{ex:UHFminBasis}, as it is exact in the infinite basis limit. To do this, first show that
\begin{subequations}
\begin{align}
\Psi^1_{\text{HL}} &= \ifbool{normSlatDet}{\half}{\frac{1}{\sqrt{2}}}
\bigl(\chi_A(\vecr_1)\chi_B(\vecr_2) + \chi_A(\vecr_2)\chi_B(\vecr_1)\bigr) \notag \\*
&\eqspace\qquad\qquad\qquad
\bigl(\alpha(\sigma_1)\beta(\sigma_2) - \alpha(\sigma_2)\beta(\sigma_1)\bigr) , \\
\Psi^3_{\text{HL}} &= \ifbool{normSlatDet}{\half}{\frac{1}{\sqrt{2}}}
\bigl(\chi_A(\vecr_1)\chi_B(\vecr_2) - \chi_A(\vecr_2)\chi_B(\vecr_1)\bigr) \notag \\*
&\eqspace\qquad\qquad\qquad
\bigl(\alpha(\sigma_1)\beta(\sigma_2) + \alpha(\sigma_2)\beta(\sigma_1)\bigr)
\end{align}
\end{subequations}
are a singlet ($S=0$) and a triplet ($S=1$) respectively. Then show that the \ac{UHF} solution is a linear combination of these two Heitler--London wavefunctions.

\end{subexercise}
\end{exercise}

\begin{exercise}[Full \ac{CI} for H$_2$ in a minimal basis]
It is clear that \ac{HF} cannot give both a good energy and the correct spin state in the dissociation limit. The remedy is simple: one Slater determinant is apparently not enough. So we need to include more determinants in the description, i.e.\ to do a small \ac{CI}.

\begin{subexercise}
\item Use symmetry to argue that only the $(\sigma_u)^2$ determinant needs to be considered in the \ac{CI}.

\item Show that the required matrix elements for the full \ac{CI} calculation are
\begin{subequations}\label{eq:CIelems}
\begin{align}
E_u &= \brakket{(\sigma_u)^2}{\hat{H}}{(\sigma_u)^2}
= 2\brakket{\sigma_u}{\hat{h}}{\sigma_u} + \cbraket{\sigma_u\sigma_u}{\sigma_u\sigma_u} , \\
V &= \brakket{(\sigma_u)^2}{\hat{H}}{(\sigma_g)^2}
= \cbraket{\sigma_g\sigma_u}{\sigma_g\sigma_u} , \\
E_g &= E_{\text{RHF}}
= 2\brakket{\sigma_g}{\hat{h}}{\sigma_g} + \cbraket{\sigma_g\sigma_g}{\sigma_g\sigma_g} ,
\end{align}
\end{subequations}
where you have already evaluated $E_g$ in exercise~\ref{eq:ERHF-H2minBas}.

\item Show that these matrix elements can be expressed in the atomic basis as
\begin{subequations}
\begin{align}
\brakket{\sigma_u}{\hat{h}}{\sigma_u}
&= \frac{\brakket{A}{\hat{h}}{A} - \brakket{A}{\hat{h}}{B}}{1 - s} , \\
\cbraket{\sigma_u\sigma_u}{\sigma_u\sigma_u}
&= \frac{1}{2(1-s)^2}\bigl[\cbraket{AA}{AA} + \cbraket{AA}{BB} - {} \notag \\*
&\eqspace\hphantom{\frac{1}{2(1-s)^2}\bigl[}
4\cbraket{AA}{AB} - 2\cbraket{AB}{AB}\bigr] , \\
\ifbool{MullikenIntegApprox}{
\cbraket{\sigma_u\sigma_u}{\sigma_u\sigma_u}_M
&= \half\bigl[\cbraket{AA}{AA} + \cbraket{AA}{BB}\bigr] , \\}
\cbraket{\sigma_g\sigma_u}{\sigma_g\sigma_u}
&= \frac{1}{2(1-s^2)}\bigl[\cbraket{AA}{AA} + \cbraket{AA}{BB}\bigr] .
\end{align}
\end{subequations}

\item Setup the full \ac{CI} secular equations~\eqref{eq:fullCI} with~\eqref{eq:CIelems} and solve for the energy.

\item\label{ex:groundStateCoefs} Construct the lowest eigenvector. It is convenient to express the full \ac{CI} ground state solution as
\begin{align}\label{eq:CI-minH2groundPsi}
\ket{\Psi^{\text{CI}}_1} = \ifbool{normSlatDet}{}{\sqrt{2}\bigl(}\cos(\alpha)\ket{(\sigma_g)^2} + \sin(\alpha)\ket{(\sigma_u)^2}\ifbool{normSlatDet}{}{\bigr)}
\end{align}
and to solve for $\alpha$. This immediately ensures that your ground state is normalised.
\end{subexercise}%
\ifbool{MullikenIntegApprox}{%
For the following computer exercises use the Mulliken approximation again for the two-electron integrals.}
\begin{subexercise}[resume]
\item Optimise the orbital exponent, $\zeta$, for the full \ac{CI} ground state~\eqref{eq:CI-minH2groundPsi}. Plot the orbital exponent as a function of the bound distance. Compare with your results for \ac{RHF} and \ac{UHF}. What do you notice?

\item Plot the optimised full \ac{CI} energy as a function of the internuclear distance and compare to the \ac{RHF} and \ac{UHF} energies.

\item Plot the \ac{CI} coefficients (the coefficients of the eigenvector in exercise~\ref{ex:groundStateCoefs}) as a function of the bond distance.
\end{subexercise}
The following exercises help you to explain the behaviour of the \ac{CI} coefficients.
\begin{subexercise}[resume]
\item Write out the \ac{RHF} determinant for H$_2$ in the atomic orbital basis (minimal $1s$ basis).

\item What is wrong with the \ac{RHF} wavefunction in the dissociation limit? In other words, which terms should not be there or are missing?

\item Write out the full \ac{CI} wavefunction for H$_2$ in the dissociation limit in the atomic orbital basis (minimal $1s$ basis). How does full \ac{CI} fix the dissociation limit?

\item Argue why the equilibrium bond length of H$_2$ is predicted too short by \ac{HF}. Do you expect this trend to persist in other systems?
\end{subexercise}
\end{exercise}

\subsection{Finite gap \& aufbau in completely unrestricted \acs*{HF}}
To be able to construct the aufbau solution, it is important to have a finite gap. The gap is defined as the 
difference between the \ac{HOMO} and \ac{LUMO} energies. In Ref.~\autocite{BachLiebLoss1994} it is rigorously shown that the completely unrestricted \ac{HF} gap is always finite, i.e.\ larger than zero, so the aufbau solution always works. Additionally they showed that the aufbau solution indeed leads to the lowest completely unrestricted \ac{HF} energy. 

An other important consequence of the finite gap is that completely unrestricted \ac{HF} cannot describe metals. Because \ac{HF} is a simple orbital theory, the conductance is primarily related to the \ac{HOMO}-\ac{LUMO} gap. Only when the gap closes, we have a metal. \ac{HF} therefore predicts all material to be insulators (large gap) or semi-conductors (small gap).

Here is the proof, which actually works for any strictly positive definite interaction $w(\vecx,\vecx')$. With strictly positive definite we mean
\begin{align}
\iinteg{\vecx}{\vecx'}w(\vecx,\vecx')\abs{\psi(\vecx,\vecx')}^2 > 0 .
\end{align}
Denote by $\epsilon_1 \leq \epsilon_2 \leq \dotsb \leq \epsilon_N$ the occupied orbital energies of the corresponding \ac{HF} orbitals which constitute $\Phi_{\text{HF}}$. Further, we will introduce the notations
\begin{subequations}
\begin{align}
h_k &= \brakket{\phi_k}{\hat{h}}{\phi_k} , \\
W_{k,l} &= \half\iinteg{\vecx}{\vecx'}
\abs{\phi_k(\vecx)\phi_l(\vecx') - \phi_l(\vecx)\phi_k(\vecx')}^2w(\vecx,\vecx') .
\end{align}
\end{subequations}
Note that $W_{k,k} = 0$ and that $W_{k,l} = W_{l,k} > 0$ for $k \neq l$.

Now assume that there exists a \ac{HF} orbital $\phi_{N+1}$ with $\epsilon_{N+1} \leq \epsilon_N$, so that the minimum $E_{\text{HF}}$ was obtained with a non aufbau $\Phi_{\text{HF}}$. Let $\tilde{\Phi}$ be the Slater determinant constructed from $\phi_1,\dotsc,\phi_{N-1},\phi_{N+1}$. For the total energies we have
\begin{subequations}
\begin{align}
\brakket{\Phi_{\text{HF}}}{\hat{H}}{\Phi_{\text{HF}}}
&= \sum_{k=1}^Nh_k + \half\sum_{k,l=1}^NW_{k,l} , \\
\brakket{\tilde{\Phi}}{\hat{H}}{\tilde{\Phi}}
&= \sum_{k=1}^{N-1}h_k + \half\sum_{k,l=1}^{N-1}W_{k,l} + h_{N+1} + \sum_{l=1}^{N-1}W_{l,N+1} .
\end{align}
\end{subequations}
Since the Fock operator is constructed from $\Phi_{\text{HF}}$ we have
\begin{align}
\epsilon_k = \brakket{\phi_k}{\hat{f}}{\phi_k} = h_k + \sum_{l=1}^NW_{l,k} .
\end{align}
Now we can rewrite the energy of $\tilde{\Phi}$ in terms of the energy of $\Phi_{\text{HF}}$ as
\begin{align}
\brakket{\tilde{\Phi}}{\hat{H}}{\tilde{\Phi}}
&= \brakket{\Phi_{\text{HF}}}{\hat{H}}{\Phi_{\text{HF}}} + h_{N+1} - h_N +
\sum_{l=1}^{N-1}\bigl(W_{l,N+1} - W_{l,N}\bigr) \notag \\
&= \brakket{\Phi_{\text{HF}}}{\hat{H}}{\Phi_{\text{HF}}} + \epsilon_{N+1} - \epsilon_N - W_{N,N+1} \\
&\leq \brakket{\Phi_{\text{HF}}}{\hat{H}}{\Phi_{\text{HF}}} - W_{N,N+1} , \notag
\end{align}
where we used the assumption $\epsilon_{N+1} \leq \epsilon_N$ for the last inequality. However, we end up with a contradiction, since we find that $\tilde{\Phi}$ yields a lower energy than $\Phi_{\text{HF}}$. The assumption that an occupied \ac{HF} orbital with a lower energy than the \ac{HOMO} can exist, is wrong.

\begin{exercise}
Check each step in the proof.
\end{exercise}

\section{Basis sets}

We will first discuss some properties of basis sets on a more general level. Since quantum mechanics is more concerned about the properties of operators (spectra) than their exact representation, we will push the abstract bra-ket notation introduced by Dirac~\autocite{Grassmann1862, Dirac1939} a bit further. The elements of a Hilbert space are now represented by kets, $\ket{\phi_i}$. The inner product between two elements of the Hilbert space is now denoted as $\braket{\phi}{\psi}$. Depending on the setting, the precise formula to calculate the inner product differs. For example in an $N$-dimensional vector space we would have
\begin{align}
\braket{v}{w} = \sum_{i=1}^Nv_i^*w_i = \mat{v}^{\dagger}\mat{w} .
\end{align}
The part $\bra{v}$ is called the `bra' and includes the complex conjugation. When the bra and ket combine in the proper manner, $\braket{\cdot}{\cdot}$, their combination implies summation to form the inner product.
In the case of $L^2(\Reals^3)$ the inner product is implemented as
\begin{align}
\braket{f}{g} = \integ{\vecr}f^*(\vecr)g(\vecr) .
\end{align}
You see that the inner product is basically the same as in the vector case, except that we sum (integrate) over a continuous index.

The most important property of a basis is completeness. We will first discuss this for the usual vector space, so in the case we have a discrete index for the components. This means that the unit operator has a unique representation in the basis. In bra-ket notation this representation becomes very elegantly
\begin{align}\label{eq:unityExpansion}
\hat{1} = \sum_{i,j}\ket{\phi_i}\bigl(\mat{S}^{-1}\bigr)_{ij}\bra{\phi_j} ,
\end{align}
where $S_{ij} \isDefinedAs \braket{\phi_i}{\phi_j}$.
Its importance comes from the fact that this means that every element in the Hilbert space can be represented as a unique linear combination, since
\begin{align}\label{eq:fUnit}
\ket{f} = \hat{1}\ket{f} = \sum_{i,j}\ket{\phi_i}\bigl(\mat{S}^{-1}\bigr)_{ij}\braket{\phi_j}{f} ,
\end{align}
where $\sum_j\bigl(\mat{S}^{-1}\bigr)_{ij}\braket{\phi_j}{f}$ are the expansion (Fourier) coefficients (see exercise~\ref{ex:FourierCoefs}). There are two things which basically can go wrong for a basis to be complete. Either you miss some elements, so you do not cover the whole space. Or there are too many elements, so you cover some parts of your space multiple times. This is called overcompleteness and destroys the uniqueness of the expansion (Fourier) coefficients. In other words, some of your basis states are linearly dependent.

\begin{exercise}
Check that~\eqref{eq:unityExpansion} indeed acts as the unit operator on any arbitrary state $\ket{f}$ constructed as linear combination of the basis states $\ket{\phi_i}$, i.e.
\begin{align}
\ket{f} = \sum_k\ket{\phi_k}f_k .
\end{align}
Thus you should find that $f_i = \sum_j\bigl(\mat{S}^{-1}\bigr)_{ij}\braket{\phi_j}{f}$ in~\eqref{eq:fUnit}
\end{exercise}

An other convenient property is orthonormality
\begin{align}
\braket{\phi_i}{\phi_j} = S_{ij} = \delta_{ij} .
\end{align}
You see that the unit operator in~\eqref{eq:unityExpansion} becomes diagonal and the expansion (Fourier) coefficients simplify to $\braket{\phi_i}{f}$.
There are different ways to generate orthonormal basis sets. One way is to diagonalise a hermitian operators, since the eigenstates belonging to different eigenvalues are orthogonal. In the the degenerate subspace we can use an other hermitian operator or use one of the following techniques
\begin{description}
\item[Gramm--Schmidt orthogonalisation] Simply follow the following algorithm
\begin{align}
\text{step 1:}\qquad	\ket{u_1}	&= \ket{v_1} , \notag \\*
\text{step 2:}\qquad	\ket{u_2}	&= \ket{v_2} - \frac{\braket{u_1}{v_2}}{\braket{u_1}{u_1}}\ket{u_1} \notag , \\*
\text{step 3:}\qquad	\ket{u_3}	&= \ket{v_3} - \frac{\braket{u_1}{v_3}}{\braket{u_1}{u_1}}\ket{u_1} -
\frac{\braket{u_2}{v_3}}{\braket{u_2}{u_2}}\ket{u_2}  \\
\shortvdotswithin{=}
\text{step n:}\qquad	\ket{u_n}	&= \ket{v_n} - \sum_{i=1}^{n-1}\frac{\braket{u_i}{v_n}}{\braket{u_i}{u_i}}\ket{u_i} 
= \left(\hat{1} - \sum_{i=1}^{n-1}\frac{\ket{u_i}\bra{u_i}}{\braket{u_i}{u_i}}\ket{u_i} \right)\ket{v_n} . \notag
\end{align}
So at each step, you simply project out all the previously found components with the projector $\ket{u_i}\bra{u_i} / \norm{u_i}^2$.

\item[Löwdin orthonormalisation] which was introduced in exercise~\ref{ex:UHFminBasis}
\begin{align}
u_i = \sum_j\ket{v_j}\bigl(\mat{S}^{-\nhalf}\bigr)_{ji} .
\end{align}
The inverse square root is defined via the spectral decomposition, so $\mat{S}^{-\nhalf} \isDefinedAs \mat{U}^{\dagger}\diag\bigl(\mat{s}^{-\nhalf}\bigr)\mat{U}$, where $\mat{U}$ is the unitary matrix which diagonalises $\mat{S}$, i.e.\ $\mat{S}\mat{U} = \mat{U}\diag(\mat{s})$ and $s_i$ are its eigenvalues.

\item[Cholesky decomposition] Any hermitian positive definite matrix can be written as $\mat{S} = \mat{L}\mat{L}^{\dagger}$, where $\mat{L}$ is a lower triangular matrix and is unique. An orthonormal basis is now readily constructed as
\begin{align}\label{eq:choleskyOrthonorm}
\ket{u_i} = \sum_j\ket{v_j}\bigl(\mat{L}^{-\dagger}\bigr)_{ji} ,
\end{align}
where $\mat{L}^{-\dagger} \isDefinedAs (\mat{L}^{-1})^{\dagger} = (\mat{L}^{\dagger})^{-1}$.
The advantage is that only a Cholesky decomposition needs to be performed which is computationally more efficient than a full diagonalisation. Because $\mat{L}$ is lower triangular, the solution of~\eqref{eq:choleskyOrthonorm} is very fast.

\item[Any other decomposition of the form $\mat{S} = \mat{Q}\mat{Q}^{\dagger}$.]
\end{description}

\noindent
An additional advantage of an orthonormalisation procedure is that they automatically provide a check for linear dependency and even provide a remedy.
\begin{description}
\item[Gramm--Schmidt] breaks down if the generated $\ket{u_n} \approx \ket{0}$ and can simply be skipped.

\item[Löwdin] leave out the eigenvalues (close to) zero. To do this, one rather works with $\mat{\tilde{S}} \isDefinedAs \mat{U}^{\dagger}\diag\bigl(\mat{s}^{-\nhalf}\bigr)$ instead of $\mat{S}^{-\nhalf}$. This type of orthonormalisation is sometimes called canonical orthonormalisation.

\item[Cholesky] depends on the actual implementation, but the basic strategy is to prevent the diagonal elements of $\mat{L}$ from becoming small.
\end{description}

\begin{exercise}
Check that indeed the Löwdin, canonical and Cholesky methods yield orthonormal bases.
\end{exercise}

\begin{exercise}
A nice feature of the Löwdin orthonormalisation is that it yields the smallest perturbation (in $L^2$ norm) of the original basis which makes it orthonormal. In this sense the Löwdin orthonormalisation is superior to any other method. Proof this. That is, show that
$\norm{\mat{\chi} - \mat{\chi}\mat{T}}_2^2$ is minimised for $\mat{T} = \mat{S}^{-\nhalf}$, under the constraint that $\braket{\mat{\chi}\mat{T}}{\mat{\chi}\mat{T}} = \mat{1}$.
\end{exercise}

\noindent
All these definitions are also useful for bases with continuous indices like the position basis, $\ket{\vecr}$, or the momentum basis, $\ket{\veck}$, though we need to redefine them with continuous analogues
\begin{align}
				&&&\text{discrete, $\ket{\phi_i}$}&	&\text{continuous, $\ket{\vecr}$} \notag \\*
&\text{completeness}&	&\sum_i\ket{\phi_i}\bra{\phi_i} = \hat{1}&
	&\integ{\vecr}\ket{\vecr}\bra{\vecr} = \hat{1} \\
&\text{orthonormality}&	&\braket{\phi_i}{\phi_j} = \delta_{ij}&
	&\braket{\vecr}{\vecr'} = \delta(\vecr - \vecr') \notag
\end{align}
Basis transformations now work in exactly the same manner as before
\begin{subequations}
\begin{align}
\ket{\phi_i} &= \hat{1}\ket{\phi_i} = \integ{\vecx}\ket{\vecr}\braket{\vecr}{\phi_i}
= \integ{\vecr}\ket{\vecr}\phi_i(\vecr) , \\
\ket{\vecr} &= \hat{1}\ket{\vecr} = \sum_i\ket{\phi_i}\braket{\phi_i}{\vecr}
= \sum_i\ket{\phi_i}\phi^*_i(\vecr) .
\end{align}
\end{subequations}
So you see that the orbitals and wavefunctions we have been working with are regarded as expansion (Fourier) coefficients in the abstract bra-ket formalism.
Similarly, for the matrix elements of the operators we have
\begin{align}
\brakket{\phi_i}{\hat{O}}{\phi_j}
&= \iinteg{\vecr}{\vecr'}\braket{\phi_i}{\vecr}\brakket{\vecr}{\hat{O}}{\vecr'}\braket{\vecr'}{\phi_j} \notag \\
&= \iinteg{\vecr}{\vecr'}\phi^*_i(\vecr)\brakket{\vecr}{\hat{O}}{\vecr'}\phi_j(\vecr') .
\end{align}
The advantage of the position basis is that we typically have a good intuition how the matrix elements should be defined, e.g.\ local potentials
\begin{align}
\brakket{\vecr}{\hat{v}_{\text{loc}}}{\vecr'}
= v_{\text{loc}}(\vecr,\vecr') = v_{\text{loc}}(\vecr)\delta(\vecr - \vecr') .
\end{align}
For the kinetic energy, momentum space is easier, since the momentum and hence also the kinetic energy operator will be multiplicative
\begin{align}\label{eq:pAndTmomBas}
\brakket{\veck}{\hat{p}}{\veck'} &= \veck\delta(\veck - \veck')	&
&\text{and} &
\brakket{\veck}{\hat{T}}{\veck'} &= \frac{\abs{\veck}^2}{2}\delta(\veck - \veck') .
\end{align}
Since we know the Fourier coefficients, $\braket{\vecr}{\veck} = (2\pi)^{-\nfrac{d}{2}}\e^{\im\veck\cdot\vecr}$, where $d$ is the dimensionality of the space, we can transform these matrix elements to position space and find
\begin{subequations}\label{eq:pAndTposBas}
\begin{align}
\brakket{\vecr}{\hat{p}}{\vecr'} &= \im\mat{\nabla}_{\vecr'}\delta(\vecr - \vecr')
= -\im\mat{\nabla}_{\vecr}\delta(\vecr - \vecr') , \\
\brakket{\vecr}{\hat{T}}{\vecr'} &= -\half\nabla^2_{\vecr'}\delta(\vecr - \vecr')
= \half\mat{\nabla}_{\vecr}\cdot\mat{\nabla}_{\vecr'}\delta(\vecr - \vecr') .
\end{align}
\end{subequations}

\begin{exercise}
Derive the matrix elements for the momentum and kinetic energy operator in the position basis~\eqref{eq:pAndTposBas} starting from their matrix elements in the momentum basis~\eqref{eq:pAndTmomBas}. You need to use that
\begin{align}
\delta(\vecr - \vecr') = \frac{1}{(2\pi)^d}\integ{\veck}\e^{\im\veck\cdot(\vecr - \vecr')} .
\end{align}
\end{exercise}

\subsection{Atomic basis sets: \acfp*{STO}}
In practice we need to work with a finite basis, so we are always infinitely far from a complete basis. However, the basis does not need to be able to represent any state, but only the ground state or some low excited state. We expect the ground state of a molecule or solid to be small distortions of the ground states of the atoms. So it is natural to start from an atomic basis set.

Since the kinetic energy dominates in the outer region of the molecular Schrödinger equation, its bound states need to decay exponentially as $\e^{-\alpha r}$. As the Coulomb potential becomes infinite at the nuclei, the solutions either need a cusp to compensate with an infinite kinetic energy ($s$-orbitals) or they need to be zero at the nuclei, i.e.\ have a nodal surface ($p$, $d$, etc.). For example, when we consider the complete set of hydrogenic solutions we have
\begin{description}[style=sameline]
\item[bounded]
\begin{align}
\psi_{nlm}(r,\theta,\phi)
= N_{nl}\e^{-\rho/2}\rho^l\Laguerre^{2l+1}_{n - l - 1}(\rho)\SpherHarm^m_l(\theta,\phi) ,
\end{align}
where
\begin{itemize}
\item $\rho = 2Zr/n$,
\item $N_{nl}$ is a normalisation constant,
\item $\Laguerre^{\alpha}_n(x)$ is a generalized Laguerre polynomial,
\item $\SpherHarm^m_l(\theta,\phi)$ is a spherical harmonic,
\item $n=1,2,\dotsc$ and $l=0,1,\dotsc,n-1$ and $m=-l,-l+1,\dotsc,l$.
\end{itemize}

\item[unbounded]
\vspace{\baselineskip}
The unbounded\slash{}ionised states are often forgotten\slash{}neglected in the treatment of the hydrogen atom, but they are also part of the spectrum. The unbounded solutions are the ones with positive energies, so one needs to solve
\begin{align}
\left(-\half\nabla^2 + \frac{Z}{r}\right)\phi_{\veck}(\vecr) = \frac{k^2}{2}\phi_{\veck}(\vecr) ,
\end{align}
where $k^2 / 2 = \abs{\veck}^2/2$ is the kinetic energy of the electron far away form the nucleus. The radial part of the solutions are the Coulomb wavefunctions
\begin{align}
F_{\eta l}(\rho) = \tilde{N}_{\eta l}\rho^{l+1}\e^{\mp\im\rho}
\KummerM(l+1\mp\im\eta,2l+2,\pm2\im\rho) ,
\end{align}
where
\begin{itemize}
\item $\eta = Z / k \in [0,\infty)$ is a continuous quantum number, (continuous equivalent of $n$).
\item $\tilde{N}_{\eta l}$ is a normalisation constant,
\item $\KummerM(a,b,z)$ confluent hypergeometric function (further generalisation of the Laguerre functions).
\end{itemize}
So the Coulomb wavefunctions are similar to the radial part of the bound states, $R_{nl}(\rho)$.
The full solutions are now obtained by glueing the Coulomb wavefunctions to the spherical harmonics
\begin{align}
\tilde{\psi}_{\eta l m}(r,\theta,\phi) = F_{\eta l}(\rho)Y^m_l(\theta,\rho) .
\end{align}

\end{description}
The completeness relation for the hydrogenic solutions is therefore
\begin{align}
\hat{1}
&= \underbrace{\sum_{n,l,m}\ket{\psi_{nlm}}\bra{\psi_{nlm}}}_{\text{discrete part}} +
\underbrace{\binteg{\eta}{0}{\infty}\sum_{l,m}\ket{\tilde{\psi}_{\eta lm}}\bra{\tilde{\psi}_{\eta lm}}}_{\text{continuous part}} .
\end{align}
When people talk about \acfp{STO}, they only mean the solutions which decay exponentially. Since they only form the discrete part of the spectrum, the \ac{STO} set is not complete.

\subsection{Atomic basis sets: \acfp*{GTO}}
The main problem of the \ac{STO} basis set is not its incompleteness, but to evaluation of the 3 and 4 centre integrals. With 3\slash{}4 centre integrals we mean two-electron integrals where the different \acp{STO} are located at 3\slash{}4 atoms. In 1950 Boys therefore proposed to use \acfp{GTO} instead~\autocite{Boys1950}. Because the product of two Gaussians is a new Gaussian located between the two original Gaussians, all 3\slash{}4 centre integrals reduce to 2 centre integrals

\begin{exercise}
Show for two Gaussians centred at $\coord{A}$ and $\coord{B}$
\begin{align}
\e^{-\alpha r_A^2}\e^{-\beta r_B^2} = \e^{-\mu R_{AB}^2} \e^{-pr_P^2} ,
\end{align}
where $\vecr_A = \vecr - \coord{A}$, $r_A = \abs{\vecr_A}$ and idem for $\vecr_B$ and $\vecr_P$. Further,
\begin{align*}
p &= \alpha + \beta ,			&	\coord{P} &= \frac{\alpha\coord{A} + \beta\coord{B}}{p} , \\
\mu &= \frac{\alpha\beta}{p} ,	&	R_{AB} &= \abs{\coord{A} - \coord{B}} .
\end{align*}
\end{exercise}

\noindent
Now you might be concerned how to deal with a Gaussian 2 centre integral. We will not consider this in detail, but the basic trick is to express also the Coulomb interaction as a Gaussian integral
\begin{align}
\frac{1}{r_C} = \frac{1}{\sqrt{\pi}}\binteg{t}{-\infty}{\infty} \e^{-r_C^2t^2} .
\end{align}
The Gaussian product rule can then be used to reduce all Coulomb integrals to one special function, the Boys function
\begin{align}
F_n(x) = \binteg{t}{0}{1}t^{2n}\e^{-xt^2} ,
\end{align}
which can be integrated numerically\slash{}fitted\slash{}tabulated.

An additional formal advantage of the \acp{GTO} is that this set is complete. Since Gaussians are the eigenfunctions of the harmonic oscillator, the electrons cannot escape to infinity. So there are no ionisation\slash{}continuum states to worry about.

The disadvantages are obvious from the discussion on the \acp{STO}
\begin{itemize}
\item too fast decay ($\e^{-\alpha r^2}$ instead of $\e^{-\alpha r}$),
\item no cusp at the nuclei.
\end{itemize}
Therefore, one needs typically more \acp{GTO} than \acp{STO} to get the same accuracy.

\subsection{Read Atkins section 9.4~\autocite{AtkinsFriedman2010}}
The important points are
\begin{itemize}
\item contractions
\item split valence
\item polarization functions. They are important to give directions to bonds. For example a basis with only $s$ \& $p$ functions predicts ammonia to be planar. With $d$ functions it gets the correct umbrella shape.
\item The counterpoise correction is not always the right thing to do. A more complete error analysis shows that the problem is more delicate. There is an additional error of opposite sign due to the basis incompleteness in describing the chemical bond. So do not take this blindly. A more detailed account can be found in Ref.~\autocite{ShengMentelGritsenk2011}.
\end{itemize}

\chapter{Density Functional Theory}

\section*{Prologue}
These lecture notes provide a concise introduction to \acf{DFT}. As with most theories, the historical developments are not in a logical order, so the lecture notes do not follow the historical time-path of the development of \ac{DFT} to give a more logical presentation.

\section{Introduction}
The classical approach to quantum mechanics is to solve for the wavefunction, $\Psi$, via the Schrödinger equation. However, in general we are only interested in a reduced quantity in the sense of the amount of information it contains. Examples are:
\begin{itemize}
\item the density, $\rho(\vecr)$
\item the energy, $E$
\item
\item
\item
\end{itemize}

\begin{exercise}
Add some more observables of interest.
\end{exercise}

\noindent
It would be convenient to calculate these reduced quantities directly, instead of having to calculate the full $\Psi$ first. Important reduced quantities from which a lot of other reduced quantities can be calculated are
\begin{itemize}
\item the density, $\rho(\vecr)$
\item the \acf{1RDM}, $\gamma(\vecr,\vecr')$
\item the pair-density, $P(\vecr_1,\vecr_2)$
\end{itemize}
The density is a well-known quantity, but the other two reduced quantities might be less familiar. We will define them shortly.
These quantities can be used to calculate all the individual components of the total energy
\begin{align}
E = T[\gamma] + V[\rho] + W[P].
\end{align}
As an example consider the kinetic energy, $T$. We will use $\vecx \isDefinedAs \vecr\sigma$ as a combined space-spin coordinate and the integration over $\vecx$ implies also the summation over the spin-variable
\begin{align}
\integ{\vecx} \isDefinedAs \sum_{\sigma}\integ{\vecr}.
\end{align}
The kinetic energy can now be worked out as
\begin{align}\label{eq:kinIn1RDM}
T &= \ifbool{normSlatDet}{}{\frac{1}{N!}}
\integ{\vecx_1\dotsi\ud\vecx_N}\Psi^*(\vecx_1,\vecx_2,\dotsc,\vecx_N)
\sum_{\mathclap{i=1}}^N-\half\nabla_i^2\Psi(\vecx_1,\vecx_2,\dotsc,\vecx_N) \notag \\
&= -\ifbool{normSlatDet}{\half}{\frac{1}{2N!}}\integ{\vecx_1\dotsi\ud\vecx_N}\Psi^*(\vecx_1,\vecx_2,\dotsc,\vecx_N)
\nabla_{\vecr_1}^2\Psi(\vecx_1,\vecx_2,\dotsc,\vecx_N) + \dotsb + {} \notag \\*
&\eqspace
{-\ifbool{normSlatDet}{\half}{\frac{1}{2N!}}}
\integ{\vecx_1\dotsi\ud\vecx_N}\Psi^*(\vecx_1,\vecx_2,\dotsc,\vecx_N)
\nabla_{\vecr_N}^2\Psi(\vecx_1,\vecx_2,\dotsc,\vecx_N) \notag \\
&= -\ifbool{normSlatDet}{\frac{N}{2}}{\frac{N}{2N!}}\integ{\vecx_1\dotsi\ud\vecx_N}\Psi^*(\vecx_1,\vecx_2,\dotsc,\vecx_N)
\nabla_{\vecr_1}^2\Psi(\vecx_1,\vecx_2,\dotsc,\vecx_N) \notag \\
&= -\half\integ{\vecx_1}\bigl[\nabla_{\vecr_1}^2\gamma(\vecx_1,\vecx_1')\bigr]_{\vecx_1'=\vecx_1},
\end{align}
where $\vecx \isDefinedAs \vecr\sigma$ is the combined space-spin coordinate and we used the permutational symmetry of the wave function. The \acf{1RDM} in the last line appears naturally, which is defined as
\begin{multline}\label{eq:1RDMdef}
\gamma(\vecx_1,\vecx_1') \isDefinedAs
\ifbool{normSlatDet}{N}{\frac{1}{(N-1)!}}
\integ{\vecx_2\dotsi\ud\vecx_N}\Psi(\vecx_1,\vecx_2,\dotsc,\vecx_N) \times {} \\*
\Psi^*(\vecx_1',\vecx_2,\dotsc,\vecx_N).
\end{multline}
Since the kinetic energy operator does not depend on spin, we can also integrate out (sum over) the remaining spin-degree of freedom, which gives the spin-integrated \ac{1RDM}
\begin{align}\label{eq:spinInteg1RDMdef}
\gamma(\vecr,\vecr') \isDefinedAs \sum_{\sigma}\gamma(\vecr\sigma,\vecr'\sigma).
\end{align}
The kinetic energy can now be calculated directly from the \ac{1RDM} as
\begin{align}
\label{eq:Tfrom1RDM}
T[\gamma] = -\half\integ{\vecr}\bigl[\nabla_{\vecr}^2\gamma(\vecr,\vecr')\bigr]_{\vecr'=\vecr}
= \half\integ{\vecr}\bigl[\nabla_{\vecr}\cdot\nabla_{\vecr'}\gamma(\vecr,\vecr')\bigr]_{\vecr'=\vecr},
\end{align}
where in the last step we used partial integration. The spin-density and the spin-pair-density can be expressed in terms of the wavefunction respectively as
\begin{align}
\rho(\vecx_1) &\isDefinedAs \ifbool{normSlatDet}{N}{\frac{1}{(N-1)!}}
\integ{\vecx_2\dotsi\vecx_N}\abs{\Psi(\vecx_1,\vecx_2,\dotsc,\vecx_N)}^2, \\
P(\vecx_1,\vecx_2) &\isDefinedAs \ifbool{normSlatDet}{N(N-1)}{\frac{1}{(N-2)!}}
\integ{\vecx_3\dotsi\vecx_N}\abs{\Psi(\vecx_1,\vecx_2,\dotsc,\vecx_N)}^2 .
\end{align}

\begin{exercise}
Show that the other two components of the total energy can be written as
\begin{align}
\label{eq:locPotDens}
V[\rho] 	&= \integ{\vecr}v(\vecr)\rho(\vecr), \\
\label{eq:twoBodyPair}
W[P]		&= \half\iinteg{\vecr_1}{\vecr_2}w(\abs{\vecr_1-\vecr_2})P(\vecr_1,\vecr_2),
\end{align}
where $v(\vecr)$ is a local (external) potential, such as the Coulomb interaction with the nuclei in a molecule and $w(\abs{\vecr_1-\vecr_2})$ is the interaction, which will be the Coulomb interaction $1/\abs{\vecr_1-\vecr_2}$ for non-relativistic electrons.
\end{exercise}

\noindent
The \acf{2RDM} is defined in a similar fashion as the \ac{1RDM}
\begin{multline}\label{eq:2RDMdef}
\Gamma(\vecx_1\vecx_2,\vecx_2'\vecx_1') \isDefinedAs
\ifbool{normSlatDet}{N(N-1)}{\frac{1}{(N-2)!}}
\integ{\vecx_3\dotsi\ud\vecx_N}
\Psi(\vecx_1,\vecx_2,\vecx_3,\dotsc,\vecx_N) \times {} \\
\Psi^*(\vecx_1',\vecx_2',\vecx_3,\dotsc,\vecx_N).
\end{multline}

\begin{exercise}
Check that
\begin{align}
P(\vecx_1,\vecx_2) 			&= \Gamma(\vecx_1\vecx_2,\vecx_2\vecx_1), \\
\gamma(\vecx_1,\vecx_1')		&= \frac{1}{N-1}\integ{\vecx_2}\Gamma(\vecx_1\vecx_2,\vecx_2\vecx_1'), \\
\rho(\vecx)				&= \gamma(\vecx,\vecx)
						= \frac{1}{N-1}\integ{\vecx'}\Gamma(\vecx\vecx',\vecx'\vecx).
\end{align}%
\end{exercise}

\noindent
Since we can calculate $\rho$, $\gamma$ and $P$ from the \ac{2RDM}, we actually only need the \ac{2RDM} to calculate the total energy and not the full many-body wavefunction $\Psi$~\autocite{Lowdin1955, Mayer1955}, which is just a 4-point function. 
We say that the total energy is a functional of the \ac{2RDM}, $[\Gamma]$.

Using the variational principle, we can try to find the ground state energy by minimising the energy functional over all 4-point functions, $\Xi(\vecx_1\vecx_2,\vecx_2'\vecx_1')$ one can think of. It turns out, however, that this minimum does not exist, since the functional is not bounded from below
\begin{align}
\inf_{\Xi} E[\Xi] = -\infty.
\end{align}
The problem is that we cannot freely vary over every possible 4-point function that we can imagine. We also should guarantee that there exists a wavefunction that actually corresponds to this 4-point function, so that it is an actual \ac{2RDM}. Only when we guarantee that the 4-point functions are true \acp{2RDM}, we can invoke the variational principle to argue that $E[\Gamma]$ will be bounded from below by the true ground state energy.

Otherwise, we cannot invoke the variational principle to argue that $E[\Gamma]$ will be bounded from below by the true ground state energy. So if we do not enforce that $\Xi$ is a proper \ac{2RDM}, the calculated energy will be lower than the actual ground state energy~\autocite{Tredgold1957, MizunoIzuyama1957, Ayres1958, Bopp1959, Coleman1963}. A \ac{2RDM} which can be generated by a wavefunction via~\eqref{eq:2RDMdef} is called an $N$-representable \ac{2RDM}. Limiting our search over only $N$-representable \acp{2RDM} we actually have
\begin{align}
E_{\text{gs}} = \min_{\text{$N$-representable $\Gamma$}}E[\Gamma].
\end{align}
Unfortunately, this is not a practical solution to determine the ground state energy. It turns out that it is very hard to tell for a given \ac{2RDM} if it is $N$-representable or not~\autocite{Coulson1960, Klyachko2006}. Some necessary conditions $N$-rep\-resentability conditions are known, but not all of them. Probably, imposing all $N$-rep\-resentability conditions is equally or even more difficult than solving the Schrödinger equation itself. There are some efforts to impose only some of the $N$-representability conditions and to hope for a good energy, although there is no proof that the energy would not collapse to $-\infty$. By imposing more and more $N$-representability conditions the true ground state energy is approached from below. In this sense this strategy is complementary to \ac{CI}.

\begin{figure}[t]
\begin{center}
\begin{tikzpicture}
  \tikzset{>=Stealth}
  \draw[->] (0,0,0) -- (3.5,0,0) node[anchor=north] {$\Gamma_{1111}$};
  \draw[->] (0,0,0) -- (0,2,0) node[anchor=south] {$E[\Gamma]$};
  \draw[->] (0,0,0) -- (0,0,3.5) node[anchor=north] {$\Gamma_{1212}$};
  \coordinate (A) at (1,0,1);
  \coordinate (B) at (1,1,0);
  \coordinate (C) at (2,1.5,0);
  \coordinate (D) at (3,1.5,0.5);
  \coordinate (E) at (3,1,1);
  \coordinate (F) at (2.6,0.6,1.2);
  \filldraw[draw=blue,thick,fill=blue!20,fill opacity=0.5] (A) -- (B) -- (C) -- (D) -- (E) -- (F) -- cycle;
  \draw[thick,red] (A) + (-0.8,-0.3,-0.1) -- +(2.4,0.9,0.3);
  \draw[thick,red] (A) + (0,-0.5,0.5) -- +(0,1.25,-1.25);
  \fill[black] (A) circle (1pt);
\end{tikzpicture}
\begin{tikzpicture}
  \tikzset{>=Stealth}
  \draw[->] (0,0,0) -- (3.5,0,0) node[anchor=north] {$\Gamma_{1111}$};
  \draw[->] (0,0,0) -- (0,2,0) node[anchor=south] {$E[\Gamma]$};
  \draw[->] (0,0,0) -- (0,0,3.5) node[anchor=north] {$\Gamma_{1212}$};
  \coordinate (A) at (2,0,1);
  \coordinate (B) at (1,0,2);
  \coordinate (C) at (0,1,1);
  \coordinate (D) at (0,1.5,0);
  \coordinate (E) at (0.5,1.5,-0.5);
  \coordinate (F) at (2,0.5,0);
  \filldraw[draw=blue,thick,fill=blue!20,fill opacity=0.5] (A) -- (B) -- (C) -- (D) -- (E) -- (F) -- cycle;
  \draw[thick,red] (2.5,0,0.5) -- (0.5,0,2.5);
  \draw[very thick] (A) -- (B);
\end{tikzpicture}
\end{center}
\caption{An artistic impression of the 2RDM optimization.
In the left figure, the external potential in $\hat{L}_1$ is set such that there is a unique minimum, dictated by the boundaries. The active ones are shown in red. In the right figure, a different external potential leads to a different $\hat{L}_2$, i.e. orientation of the plane. In this case, only one constraint is active and we have a degenerate minimum, i.e.\ a set of degenerate ground state 2RDMs.}
\label{fig:linearOptimisation}
\end{figure}

That the minimum values is now determined by constraints rather than the functional itself, becomes quite obvious when we inspect the functional $E[\Gamma]$ more closely
\begin{equation}
E[\gamma] = \Trace\bigl\{\hat{L}\,\hat{\Gamma}\bigr\}
= \iinteg{\vecx_1}{\vecx_2}\hat{L}
\Gamma(\vecx_1\vecx_2,\vecx_2'\vecx_1')\bigr]_{\vecx_1' = \vecx_1; \vecx_2'=\vecx_2} ,
\end{equation}
where the operator $\hat{L}$ can be defined in different ways. The reasonably symmetric manner is
\begin{equation}
\hat{L} = \frac{1}{2(N-1)}\left(-\half\bigl(\nabla^2_{\vecr_1} + \nabla^2_{\vecr_2}\bigr) + v(\vecr_1) + v(\vecr_2)\right) + \frac{1}{\abs{\vecr_1 - \vecr_2}} .
\end{equation}
The key point here is that the functional is just linear in the 2RDM, which means that the functional itself is a straight hyper plane, and the operator $\hat{L}$ is the normal of this hyper plane in some advanced mathematical sense. Since the functional is just a plane, you can just keep on sliding down till you finally hit a boundary. This is illustrated in Fig.~\ref{fig:linearOptimisation} using a low dimensional representation of the 2RDM.

The $N$-representability conditions for the density and the \ac{1RDM}~\autocite{Coleman1963} are actually known and quite simple, so one can hope that the functionals $E[\rho]$ and $E[\gamma]$ might exist. They will definitely be more complicated than $E[\Gamma]$, but it might be possible to find some good approximations to parts of the energy which are not readily expressible in terms of the density or \ac{1RDM} respectively.

\begin{exercise}
Some of the $N$-representability conditions for the \ac{2RDM} are actually quite easy to derive directly from its definition~\eqref{eq:2RDMdef}. Find these $N$-representability conditions for $\Gamma(\vecx_1\vecx_2,\vecx_2'\vecx_1')$. You can find four different kind of conditions in this manner. Consider the permutation symmetry of the wavefunction (2 conditions) and complex conjugation (1 condition). Also consider the `diagonal', $\vecx_i = \vecx_i'$. This should give you a positivity condition (an inequality).
\end{exercise}

\section{Hohenberg--Kohn theorems (1964)}
The existence of the functional $E[\rho]$ has been proved
\mkbibfootnote{on a mathematical level only conjectured} by Hohenberg and Kohn and additionally that the potential generating the ground state density is unique up to a constant shift~\autocite{HohenbergKohn1964}. These theorems form therefore the basis of \ac{DFT} and will be considered in detail. First consider the composite mapping
\begin{align}
v(\vecr) \xmapsto{\hat{H}\Psi = E\Psi} \Psi \xmapsto{\integ{\vecr}} \rho(\vecr).
\end{align}
Now we ask ourselves the question if these maps are invertible. If this is the case, then we can always go back to the potential and reconstruct everything we need to know. The proof consists of two parts, showing the invertibility of each map separately. For simplicity we only consider non-degenerate ground states, though the results can be generalised to degenerate ground states without too much difficulty~\autocite{Kohn1985,DreizlerGross1990}. You are asked to do this yourself in Exercise~\ref{ex:HKdegenerate}.

Hohenberg and Kohn only treat the second map in detail and only spend one sentence on the first map. However, this is actually the most tricky part and has only been proved to be correct in 2018 for a class of systems which includes Coulombic systems \autocite{Garrigue2018}. Therefore, we treat it here in a separate theorem, though we will not completely prove it due to the heavy math required. But we will pinpoint the problem in the proof.

\begin{theorem}[HK-I]
\label{theo:HK-I}
The map from local potentials to ground states, $v \mapsto \Psi$ is invertible modulo a constant (shift) in the potential.
\end{theorem}
\begin{proof}
Suppose that there are two potential $v_1$ and $v_2$ which both yield the same ground state $\Psi$, then from the Schrödinger equation we have
\begin{subequations}
\begin{align}
\bigl(\hat{T} + \hat{V}_1 + \hat{W}\bigr)\Psi &= E_1\Psi, \\
\bigl(\hat{T} + \hat{V}_2 + \hat{W}\bigr)\Psi &= E_2\Psi .
\end{align}
\end{subequations}
Subtracting both equations from each other, we find
\begin{align}
(E_1 - E_2)\Psi
= \bigl(\hat{V}_1 - \hat{V}_2\bigr)\Psi = \sum_{i=1}^N\bigl(v_1(\vecr_i) - v_2(\vecr_i)\bigr)\Psi .
\end{align}
Under the assumption that $\Psi$ does not vanish on any finite region in space,
\mkbibfootnote{To show this, one would need to prove a unique continuation property of the Schrödinger equation \autocite{Lieb1983}. More than half a century later this has been proved for a very general class of potentials by Garrigue which includes Coulomb potentials \autocite{Garrigue2018}.}
we can divide by $\Psi$ and obtain
\begin{align}
\sum_{i=1}^N\bigl(v_1(\vecr_i) - v_2(\vecr_i)\bigr) = E_1 - E_2 = \text{constant}.
\end{align}
So we find that potentials which yield the same ground state, that they can only differ by a constant.
\end{proof}

\begin{theorem}[HK-II]
\label{theo:HK-II}
The map from non-degenerate ground states, generated by local potentials, to ground state densities, $\Psi \mapsto \rho$ is invertible.
\end{theorem}

\begin{proof}
The proof goes by reductio ad absurdum. Suppose that the statement is not true, so that there exist two different non-degenerate ground state wavefunctions, $\Psi_1 \neq \Psi_2$, that both yield the same ground state density $\rho$. Then using the variational principle, we have
\begin{align}
E_1 &= \brakket{\Psi_1}{\hat{T} + \hat{V}_1 + \hat{W}}{\Psi_1} \notag \\*
&< \brakket{\Psi_2}{\hat{T} + \hat{V}_1 + \hat{W}}{\Psi_2}
= \brakket{\Psi_2}{\hat{T} + \hat{V}_2 + \hat{W}}{\Psi_2} +
\brakket{\Psi_2}{\hat{V}_1 - \hat{V}_2}{\Psi_2} \notag \\*
&\hphantom{{}< \brakket{\Psi_2}{\hat{T} + \hat{V}_1 + \hat{W}}{\Psi_2}}
{} = E_2 + \integ{\vecr}\rho(\vecr)\bigl(v_1(\vecr) - v_2(\vecr)\bigr).
\end{align}
Turning the role of the indices 1 and 2 around, we additionally find
\begin{align}
E_2 < E_1 + \integ{\vecr}\rho(\vecr)\bigl(v_2(\vecr) - v_1(\vecr)\bigr).
\end{align}
Adding both inequalities we find
\begin{align}
E_1 + E_2 < E_1 + E_2
\end{align}
and therefore, our initial assumption that both ground states $\Psi_1$ and $\Psi_2$ can yield the same density is incorrect.
\end{proof}

\noindent
The combination of the \acf{HK} theorems (\nameref{theo:HK-I} and~\nameref{theo:HK-II}) forms the foundation of \acf{DFT}, since they show that
\begin{align}
v(\vecr) \xmapsback{\text{\nameref{theo:HK-I}}} \Psi \xmapsback{\text{\nameref{theo:HK-II}}} \rho(\vecr),
\end{align}
so that we are allowed to write $v[\rho]$ as well as $\Psi[\rho]$. Since the ground state wave function is a functional of the density, $\Psi[\rho]$, also every observable is a functional of the density
\begin{align}
O[\rho] = \brakket{\Psi[\rho]}{\hat{O}}{\Psi[\rho]}.
\end{align}
In particular the ground state energy is a functional of the density
\begin{align}
E[\rho] = \brakket{\Psi[\rho]}{\hat{H}}{\Psi[\rho]}
= F_{\text{\acs{HK}}}[\rho] + \integ{\vecr}\rho(\vecr)v(\vecr),
\end{align}
where $F_{\text{\acs{HK}}}[\rho] \isDefinedAs \brakket{\Psi[\rho]}{\hat{T} + \hat{W}}{\Psi[\rho]}$ is the \ac{HK} functional and simply collects the parts of the energy which are not explicit density functionals. The \ac{HK} functional is often called a `universal' functional, because it does not depend on the particular system considered. The system (the positions of the nuclei and local external fields) only enters via the local potential $v$.

\begin{exercise}
Try to proof the Hohenberg--Kohn theorems in  the same manner for the \ac{1RDM}~\eqref{eq:1RDMdef} and non-local one-body potentials. Non-local one-body potentials are similar to the exchange potential in the sense that they act via an integral kernel on one-body states
\begin{align}
\hat{v}\psi(\vecx) = \bigl(\hat{v}\psi\bigr)(\vecx) = \integ{\vecx'}v(\vecx,\vecx')\psi(\vecx') \, .
\end{align}
On a many-body state these non-local potentials act on each coordinate individually. So for a many-body wave function the non-local one-body potential becomes
\begin{align}
\hat{V} = \sum_{i=1}^N\hat{v}(\vecx^{\vphantom{'}}_i,\vecx'_i) .
\end{align}
Due to the one-body nature, the expectation value reduces to a contraction with the \ac{1RDM}
\begin{align}
\brakket{\Psi}{\hat{V}}{\Psi} = \iinteg{\vecx}{\vecx'}v(\vecx',\vecx)\gamma(\vecx,\vecx').
\end{align}
Consider now the mappings $\hat{v} \mapsto \Psi \mapsto \gamma$.
Investigate whether you can reuse the proofs for \nameref{theo:HK-I} and~\nameref{theo:HK-II} that we used to establish \ac{DFT}.
If you can not reuse the proofs of the \ac{HK} theorems to establish a \ac{1RDM} functional theory, why does it not work?
\end{exercise}

\begin{exercise}\label{ex:HKdegenerate}
In this exercise we consider in which sense the Hohenberg--Kohn theorems can be generalised to degenerate ground states.
\begin{subexercise}
\item What changes in \nameref{theo:HK-I} if we allow for degenerate ground states?
\item Is it possible to generalise \nameref{theo:HK-II} to degenerate states? Consider the two different cases separately: the non-degenerate case ($E_1 \neq E_2$) and the degenerate case ($E_1 = E_2$).
\item Is it still possible to establish \ac{DFT} when allowing for degenerate states? In particular, are we still allowed to write $E[\rho]$?
\end{subexercise}
\end{exercise}

\begin{exercise}
Show $E[\rho_{\text{gs}}] \leq E[\rho]$, where $\rho_{\text{gs}}$ is the ground state density.
\end{exercise}

\subsection{Constrained-search formulation (1979)}
Though we have shown that the functional $E[\rho]$ exists on a formal level, we do not have an explicit form for practical calculations. The \ac{HK} functional, $F_{\text{\acs{HK}}}[\rho]$, only provides a very abstract expression for the universal functional. Here we will construct a more explicit expression for the universal functional, which has better mathematical properties and serves as a more convenient starting point to derive approximations. In fact, we will always want to resort to approximations to make \ac{DFT} useful, as an exact functional should always be too complicated (more complicated than the Schrödinger equation probably) to use in practice.

\begin{figure}[t]
  \centering
  \includegraphics[width=0.7\textwidth]{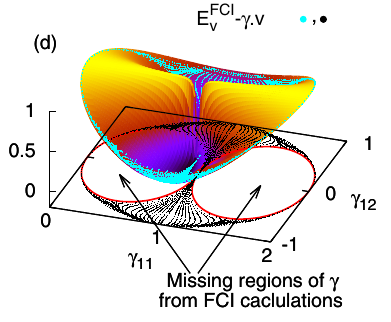}
  \caption{Taken from \autocite{CohenMori-Sanchez2016} without permission. The 3D lanscape is the constrained-search universal \acs{1RDM} functional as a function of the two independent 1RDM components for a two-orbital Hubbard system. The projection indicates the regions where the 1RDM is $v$-representable, i.e.\ the 1RDMs for which there exists a (non-local) potential generating this 1RDM via the ground state. In the two red oval encircled regions [Missing regions of $\gamma$ from FCI caclulations], the 1RDM is not $v$-representable.}
  \label{fig:1RDMlandscape}
\end{figure}

The mathematical motivation to construct a different expression for the universal functional is that the \ac{HK} functional $F_{\text{\acs{HK}}}[\rho]$ is only defined for $v$-representable densities. With $v$-representable densities we mean densities which can be generated by a \emph{ground state}, so that $\Psi[\rho]$ in HK construction exists. This is important when we want to minimise the energy by making variations in the density. We should never hit a density for which $F_{\text{\acs{HK}}}[\rho]$ does not exists, since then we would get stuck in our optimisation attempt. One would expect that all reasonably well-behaved densities (positive, smooth and normalisable) are $v$-representable, but this turns out to be not the case unfortunately~\autocite{EnglischEnglisch1983}.

An example is shown in Fig.~\ref{fig:1RDMlandscape} taken from Ref.~\autocite{CohenMori-Sanchez2016}. For technical reasons, this is not the DFT functional, but the 1RDM functional. In the $xy$-plane, the region of 1RDMs is plotted for which a (non-local) potential could be found which generates that particular 1RDM. In the two red encircled oval regions, no such potential can be found, so the 1RDM is not $v$-representable there and the \ac{HK} functional does not exist in those regions. It nicely demonstrates that the topology of the domain of the \ac{HK} functional can be quite nasty.

To avoid this problem, the domain of the universal functional was extended by Levy~\autocite{Levy1979}. He observed that the search for the ground state energy can be split as
\begin{align}
E_{\text{gs}} &= \min_{\Psi}\brakket{\Psi}{\hat{T} + \hat{V} + \hat{W}}{\Psi} \notag \\
&= \min_{\rho}\left(\min_{\Psi \to \rho}\brakket{\Psi}{\hat{T} + \hat{W}}{\Psi} + \integ{\vecr}\rho(\vecr)v(\vecr)\right) .
\end{align}
So the only thing we did is to write the minimisation over all wave functions into two parts. First we vary over all densities and inside these density variations we consider all wave functions which yield this density, $\Psi \to \rho$.
The universal functional is now readily extracted as
\begin{align}\label{eq:FLLdef}
F_{\text{LL}}[\rho] \isDefinedAs \min_{\mathclap{\Psi \to \rho}} \brakket{\Psi}{\hat{T} + \hat{W}}{\Psi}.
\end{align}
The advantage is now that the domain of this functional consists of all $N$-representable densities, i.e.\ densities that can be produced by some wave function. In fact all reasonable densities (positive, smooth and normalisable) are $N$-representable~\autocite{Harriman1981}. Since the characterization of $N$-representable densities is so much easier than for $v$-representable densities, making proper variations also becomes easier.

The $v$-representability of the density is still an issue when taking a functional derivative of $F_{\text{LL}}$ with respect to the density and therefore remains a topic of research~\autocite{Lieb1983, EnglischEnglisch1984a, EnglischEnglisch1984b, Leeuwen2003, Eschrig2003, Lammert2006a, Lammert2006b, Lammert2010, KvaalEkstromTeale2014}. In physics, lattice systems (only a grid of specific points is used) are also of much interest and these come with their own type of challenges concerning $v$-representability \autocite{ChayesChayesRuskai1985, PenzLeeuwen2021}. For a more thorough introduction to these technical matters, consult the excellent text-book by Dreizler and Gross~\autocite{DreizlerGross1990} or the review by Van Leeuwen~\autocite{Leeuwen2003}.

As mentioned before, this formulation of the exact functional is not only useful from a theoretical point of view, but can also be used as a starting point to formulate different types of approximations. Most well-known is the \acf{KS} approximation to the kinetic energy, which will be discussed in Sec.~\ref{sec:KSsystem}. The constrained-search formulation can also be used to approximate the interaction part, though this is quite involved~\autocite{Lieb1983, SeidlPerdewLevy1999, Seidl1999, SeidlGori-GiorgiSavin2007, Gori-GiorgiVignaleSeidl2009, SeidlGori-GiorgiSavin2007}. For the interaction part, we will only consider the traditional route.

Since the interaction term is mainly electrostatic, a reasonable starting point is the classical Coulomb term
\begin{align}\label{eq:HartreeDef}
W_{\text{H}}[\rho] \isDefinedAs \half\iinteg{\vecr_1}{\vecr_2}\frac{\rho(\vecr_1)\rho(\vecr_2)}{\abs{\vecr_1 - \vecr_2}},
\end{align}
This term is commonly called the Hartree term, because Hartree only took this classical interaction term into account~\autocite{Hartree1928}.
\mkbibfootnote{Probably the most important reason not to call this term the Coulomb term is that Coulomb starts with a `c' which is already in use as an abbreviation for correlation (see \acl{HF} and later in these lecture notes).}
We will consider refinements to the Hartree term, the \acf{xc} term, $W_{\text{\acs{xc}}}$, later. First we will turn our attention to the kinetic energy, since it is not so easy to find a reasonable approximation for the kinetic energy which captures the important quantum effects.

\begin{exercise}
In \ac{1RDM} functional theory, the kinetic energy is an explicit functional of the \ac{1RDM}, cf.~\eqref{eq:kinIn1RDM}, so only the interaction part of the energy has no explicit form.
Use the construction by Levy to write down the two-body term of the energy, $W$~\eqref{eq:twoBodyPair}, as an exact functional of the \ac{1RDM}.
\end{exercise}

\section{Approximating the kinetic energy}
The most important part to approximate reasonably well is the kinetic energy term, because this term is fundamentally different from its classical counterpart. The interactions with the nuclei and between the electrons is still simply the Coulomb interaction in non-relativistic quantum mechanics and crude approximations often suffice.

\subsection{Thomas--Fermi Theory (1926)}
Already in the early days of quantum mechanics, physicists tried to simplify quantum mechanics to make calculations with pen and paper more feasible.
\mkbibfootnote{The computer did not exist at that time. Programmable computers were only developed in during the 2\textsuperscript{nd} world war mainly to break the Enigma code used by the Nazis. The computers around that time were on par with human speed to perform basic mathematical operations (multiplication, addition, etc.). The main advantage was that computers do not tire and do not need coffee breaks and sleep~\autocite{Feynman2005}.}
The first \ac{DFT} approximations, therefore, date back long ago before the \ac{HK} theorems were formulated and were even formulated before \ac{HF}. The oldest approximation to the kinetic energy in terms of the density alone is due to Thomas~\autocite{Thomas1927} and Fermi~\autocite{Fermi1927}. Around 1926 they derived independently the following approximation for the kinetic energy
\begin{align}\label{eq:TFapproximation}
T_{\text{TF}}[\rho] &\isDefinedAs C_{\text{TF}}\integ{\vecr}\rho(\vecr)^{\nfrac{5}{3}}, &
C_{\text{TF}} &\isDefinedAs \frac{3}{10}\bigl(3\pi^2\bigr)^{\nfrac{2}{3}}.
\end{align}
This expression might look very strange, but it is simply the kinetic energy of a non-interacting \ac{HEG}, where the homogeneous density has been replaced by the inhomogeneous density, $\rho \to \rho(\vecr)$. The \ac{HEG} is a fictitious system with a constant electron density in an infinite box. The `nuclei' are also smeared out as a homogeneous background charge over the whole space, to counteract the Coulomb repulsion between the electrons.  so the \ac{HEG} is also sometimes called the Jellium model, as the smeared out nuclei resemble a jelly. The \ac{HEG} provides a reasonable model for the electron delocalization in metals and qualitatively reproduces features of real metals, e.g.\ plasmons and Wigner crystallization.

Combining the Thomas--Fermi kinetic energy with the Hartree approximation to the interaction term, we have a very simply approximation to the energy functional
\begin{align}
E[\rho] \approx E_{\text{TF}}[\rho] \isDefinedAs T_{\text{TF}}[\rho] + V[\rho] + W_{\text{H}}[\rho].
\end{align}
Exchange and correlation effects were not (well) known at that time, so those terms were neglected. Unfortunately, this approximation performs very poorly. The main failures of $E_{\text{TF}}[\rho]$ when solved self-consistently are
\begin{itemize}
\item too low total energy for atoms (54\% for hydrogen),
\item density decays as $r^{-6}$ instead of $\e^{-\alpha r}$,
\item no shell structure in atoms,
\item all negative ions are predicted to be unstable,
\item molecules do not exist (Teller non-binding theorem~\autocite{Teller1962}).
\end{itemize}
The main term to blame for its bad performance is the Thomas--Fermi approximation to the kinetic energy. This statement can be validated by comparing with Hartree calculations (\acs{HF} without exchange)~\autocite{Hartree1928}.

\begin{exercise}\label{ex:Theg}
Derive the Thomas--Fermi approximation for the kinetic energy by deriving the kinetic energy for the \ac{HEG}. If you have never seen calculations on the \ac{HEG}, this can be quite a challenge, so here are some steps to help you out.
\begin{subexercise}
\item The first step is to solve the non-interacting Schrödinger equation for the electrons in a finite box with sides of length $L$ and periodic boundary conditions. The one-electron solutions are $\psi_{\veck}(\vecr) = \Omega^{-\nhalf}\e^{\im\veck\cdot\vecr}$, where $\Omega = L^3$ is the volume and the wave vectors, $\veck$, are quantized as
\begin{align}\label{eq:kVecDef}
k_i &= \frac{2\pi n_i}{L}	&
i &= x, y, z &
n_i &= 0, \pm 1, \pm 2, \dotsc
\end{align}
with energies $\epsilon(\veck) = \half\abs{\veck}^2$. Since I have already given you the solutions, you only need to check whether they are correct.

\item Check that the density is constant.

\item Calculate the total number of electrons in the limit of a large box using Aufbau, i.e.\ we occupy all states with a wave vector $\abs{\veck} \leq k_F$. Assume that the system is spin compensated, so there is an equal amount of spin up and spin down electrons. In the limit of a large box, you can replace the summation over the states by integrals, so
\begin{align}
\sum_{n_x,n_y,n_z} = \frac{\Omega}{(2\pi)^3}\sum_{k_x,k_y,k_z}
\xrightarrow{L \to \infty} \frac{\Omega}{(2\pi)^3}\integ{\veck}.
\end{align}
In the first equality we compensated for the fact that the interval between $k$-points is $\Delta k_i = 2\pi/L$ instead of unity, so the volume in the second summation increased by $(2\pi/L)^3$ which is compensated for by the prefactor. For large volumes, $L \to \infty$, we replaced the sum by an integral, because $\Delta k_i = k_{i+1} - k_i \to 0$, cf.~\eqref{eq:kVecDef}. Note that we still retained the volume, $\Omega$, in front of the integral. All expectation values of size-extensive operators diverge due to this volume term. We will therefore retain this volume factor and divide by it on both sides, to obtain a finite expectation value for density-like (size-intensive) quantities.

Further note that the orbital energies behave as $\epsilon(\veck) = \half \abs{\veck}^2$, so the lowest (so occupied) states will have $\veck$-vectors contained in a sphere with some radius $k_F$.
Show that given a density $\rho = N/\Omega$, the length of the maximum occupied wave vector, the Fermi wave vector, is related to the density as $k_F \isDefinedAs (3\pi^2\rho)^{\nfrac{1}{3}}$.

\item Calculate the kinetic energy density $\tau = T/\Omega$ in the limit of a large box.

\item Construct the Thomas--Fermi approximation to the kinetic energy~\eqref{eq:TFapproximation} from the \ac{HEG} kinetic energy density. The assumption you need to use is that the kinetic energy density at a point in space can be approximated by the kinetic energy of a \ac{HEG} with the electron density at that point. To get the total kinetic energy, integrate over the kinetic energy density.
\end{subexercise}
\end{exercise}

\subsection{The Kohn--Sham system (1965)}
\label{sec:KSsystem}

The \acf{KS} system is a system composed of non-interacting particles with a prescribed density~\autocite{KohnSham1965}. The energy functional for the \ac{KS} system is simply the functional $F_{\text{LL}}$~\eqref{eq:FLLdef} without the interaction term, so only the kinetic energy is left
\begin{align}\label{eq:TsDef}
T_s[\rho] \isDefinedAs \min_{\mathclap{\Psi_s \to \rho}}\brakket{\Psi_s}{\hat{T}}{\Psi_s}.
\end{align}
This functional is simple enough to allow us to actually perform the minimisation. Since only a one-body part is present, we suspect that the wavefunction $\Phi$ which achieves this minimisation will be a Slater determinant composed of single-particle orbitals
\mkbibfootnote{Since we do now have the constraint that $\Psi_s$ should yield a prescribed density $\rho$, we are not certain anymore that the minimiser will be a Slater determinant. Nevertheless, it seems that the densities typically under consideration allow for a Slater determinant as minimiser, though exceptions are known~\autocite{Levy1982, SchipperGritsenkoBaerends1998}}
\begin{align}\label{eq:SlaterDetDef}
\Psi_s(\vecx_1,\dotsc,\vecx_N) = \ifbool{normSlatDet}{\frac{1}{\sqrt{N!}}}{}\begin{vmatrix}
\phi_1(\vecx_1)	&\phi_2(\vecx_1)	&\ldots	&\phi_N(\vecx_1)	\\
\phi_1(\vecx_2)	&\phi_2(\vecx_2)	&\ldots	&\phi_N(\vecx_2)	\\
\vdots		&\vdots			&\ddots	&\vdots			\\
\phi_1(\vecx_N)	&\phi_2(\vecx_N)	&\ldots	&\phi_N(\vecx_N)
\end{vmatrix}.
\end{align}
The \ac{1RDM} of a Slater determinant is readily evaluated to be
\begin{align}\label{eq:Slater1RDM}
\gamma_s(\vecx,\vecx') = \sum_{i=1}^N\phi_i(\vecx)\phi^*_i(\vecx').
\end{align}
So the kinetic energy can be expressed in terms of the orbitals as
\begin{align}
\tilde{T}_s[\{\phi,\phi^*\}] = -\half\integ{\vecx}\bigl[\nabla_{\vecr}^2\gamma_s(\vecx,\vecx')\bigr]_{\vecx'=\vecx}
= -\half\sum_{i=1}^N\brakket{\phi_i}{\nabla^2}{\phi_i},
\end{align}
where $[\{\phi,\phi^*\}]$ means that $\tilde{T}_s$ is a functional depending on the orbitals constituting the Slater determinant and their complex conjugates. To calculate the functional $T_s[\rho]$, we need to minimise $\tilde{T}_s$ not only under the constraint that the orbitals $\{\phi\}$ are orthonormal, but also under the constraint of a prescribed density, so we introduce the following Lagrangian
\begin{multline}
L_{\rho}[\{\phi,\phi^*\},v_s,\mat{\epsilon}] \isDefinedAs \tilde{T}_s[\{\phi,\phi^*\}] -
\sum_{ij}\epsilon_{ji}\bigl(\braket{\phi_i}{\phi_j} - \delta_{ij}\bigr) + {} \\
\integ{\vecr}v_s(\vecr)\biggl(\sum_{i,\sigma}\abs{\phi_i(\vecr\sigma)}^2 - \rho(\vecr)\biggr).
\end{multline}
A necessary condition for a minimum is that if we make small variations, $\delta L_{\rho} = 0$ to first order, so if we consider variations due to perturbations in the orbitals, we have
\begin{align}
0 = \delta L_{\rho} = \integ{\vecx}\sum_{i=1}^N\left(\delta\phi^*_i(\vecx)\frac{\delta L_{\rho}}{\delta\phi^*_i(\vecx)} +
\frac{\delta L_{\rho}}{\delta\phi_i(\vecx)}\delta\phi_i(\vecx)\right).
\end{align}
We have seen a similar expression in the derivation of the \ac{HF} equations. Since $\delta L_{\rho}$ needs to vanish for arbitrary variations, the functional derivatives need to vanish. For the derivative with respect to $\phi_i^*(\vecx)$, we find
\begin{align}\label{eq:dLdpsiStar}
0 = \frac{\delta L_{\rho}}{\delta\phi^*_i(\vecx)} = -\half\nabla^2\phi_i(\vecx) - \sum_{j=1}^N\phi_j(\vecx)\epsilon_{ji} + v_s(\vecr)\phi_i(\vecx).
\end{align}
which can be rearranged to
\begin{align}
\Bigl(-\half\nabla^2 + v_s(\vecr)\Bigr)\phi_i(\vecx) = \sum_{j=1}^N\phi_j(\vecx)\epsilon_{ji}.
\end{align}
The derivative with respect to $\phi_i(\vecx)$ is less straightforward, due to the Laplacian. There are different ways to proceed.
\mkbibfootnote{An alternative is to use the the functional derivative of a functional of the simple form $\integ{\vecs}L[f(\vecs),\nabla f(\vecs),\nabla^2 f(\vecs),\dotsc]$ can be calculated as
\begin{align*}
\frac{\delta L_{\rho}}{\delta f(\vecs)} = \frac{\du L}{\du f(\vecs)} - \nabla\cdot\frac{\du L}{\du \nabla f(\vecs)} + \nabla^2\frac{\du L}{\du \nabla^2 f(\vecs)} +\dotsb,
\end{align*}
where $\vecs$ is some arbitrary vector. This can easily be established with successive partial integrations. Note that for $\vecs = t$, $f(\vecs) = q(t)$, $L = T - V$ and only derivatives up to first order, one gets the usual Euler--Lagrange equations from classical mechanics.}
Probably the easiest way is to proceed as we did for \ac{HF} and use that the kinetic energy operator is self-adjoint, so the kinetic energy can alternatively be expressed as
\begin{align}
\tilde{T}_s = \half\sum_{i=1}^N\braket{\nabla\phi_i}{\nabla\phi_i} = -\half\sum_{i=1}^N\braket{\nabla^2\phi_i}{\phi_i}.
\end{align}
Now it is easy to take the derivative with respect to $\phi_i(\vecx)$ which gives
\begin{align}\label{eq:dLdpsi}
\Bigl(-\half\nabla^2 + v_s(\vecr)\Bigr)\phi^*_i(\vecx) = \sum_{j=1}^N\epsilon_{ij}\phi^*_j(\vecx)
\end{align}
when equated to zero.
In both equations~\eqref{eq:dLdpsiStar} and~\eqref{eq:dLdpsi} we still have the sum over the Lagrange multiplier matrix $\mat{\epsilon}$, which we would rather like to be diagonal to interpret them as orbital energies. We can proceed in exactly the same manner as in the derivation of the \ac{HF} SCF equations (Sec.~\ref{sec:HF-SCF}).
Diagonalising the Lagrange multiplier matrix brings the stationarity equations to canonical form
\begin{align}\label{eq:KSequations}
\Bigl(-\half\nabla^2 + v_s(\vecr)\Bigr)\phi_k(\vecx) = \varepsilon_k\phi_k(\vecx).
\end{align}
These are the \ac{KS} equations which yield \ac{KS} orbitals and \ac{KS} orbital energies.

Some remarks are in order.
\begin{itemize}
\item To obtain the minimum expectation value for $T_s$, we select the \ac{KS} orbitals with the lowest orbital energies to constitute our Slater determinant~\eqref{eq:SlaterDetDef}. This is the Aufbau principle, which intuitively makes sense, but it takes a more careful derivation to mathematically justify this assumption and to handle the possibility of fractionally occupied orbitals~\autocite{GiesbertzBaerends2010}. The \ac{KS} orbitals from which the Slater determinant is composed~\eqref{eq:SlaterDetDef} are called occupied orbitals and unused orbitals are called unoccupied\slash{}virtual orbitals.

\item The Lagrange multiplier for the density, $v_s(\vecr)$, has the form of a potential (\ac{KS} potential). This is not so strange, since by modifying the potential we can influence the density profile: too much density $\to$ increase the potential to push the particles away and visa versa. In this manner the \ac{KS} potential, $v_s(\vecr)$, can be obtained self-consistently by solving the \ac{KS} equations~\eqref{eq:KSequations}, calculating the density from the orbitals and by comparing to the target density, one increases\slash{}decreases to potential accordingly till the orbitals yield the required density. Further, the \ac{HK} theorems tell us that this potential, $v_s(\vecr)$, is unique for a given density modulo a constant, so there is a unique solution for the \ac{KS} potential up to a shift which we do not need to worry about. Note that this shift in the potential exactly corresponds to an overall phase factor of the wave function, whose relevance was discarded before when dealing with the Lagrange multipliers $\mat{\epsilon}$.

\item Often we deal with closed shell systems, so the amount of spin-up and spin-down electrons is the same. If there are no magnetic interactions, the wave function will be an eigenstates of $\hat{S}_z$ operator, so the orbitals come in pairs (spin-up and spin-down) with the same spatial part
\begin{align}\label{eq:restrictedOrbs}
\phi_k(\vecx) = \phi_k(\vecr\sigma) = \begin{cases}
\psi_{(k+1)/2}(\vecr)\alpha(\sigma)	&\text{for $k$ odd} \\
\psi_{k/2}(\vecr)\beta(\sigma)		&\text{for $k$ even}.
\end{cases}
\end{align}
In all expressions, the summation over the spin can now be performed explicitly, so all expressions simplify somewhat, because only half of the orbitals needs to be calculated. For example, the spin-integrated \ac{1RDM}~\eqref{eq:spinInteg1RDMdef} becomes
\begin{align}\label{eq:equalSpin1RDM}
\gamma_s(\vecr,\vecr') = 2\sum_{i=1}^{N/2}\psi_i(\vecr)\psi^*_i(\vecr').
\end{align}
\end{itemize}
It turns out that $T_s[\rho]$ provides a very good approximation to the real kinetic energy, $T[\rho]$. The main reason is that the quantum nature of the kinetic energy operator is properly taken into account. Though originally not intended by Kohn and Sham in 1965, their approximation to the kinetic energy~\autocite{KohnSham1965} has been crucial for the success of \ac{DFT} in practice.

\begin{exercise}
Check that the \ac{1RDM} for a Slater determinant has indeed the simple form as is stated in~\eqref{eq:Slater1RDM}. Note that the task is almost identical to the derivation of the Slater--Condon rules for one-body operators, so you could follow the same procedure.
\end{exercise}

\begin{exercise}
Check that the noninteracting, spin-integrated \ac{1RDM}~\eqref{eq:spinInteg1RDMdef} for a non-magnetic closed shell system is indeed given by~\eqref{eq:equalSpin1RDM}.
\end{exercise}

\begin{exercise}
Show that the functional derivative of the classical Coulomb interaction~\eqref{eq:HartreeDef} with respect to the density, gives the classical Coulomb potential. There are two ways of obtaining this derivative. The first one is to follow the same approach as I used in class for the energetic contribution from the local potential, i.e.\ to consider the functional as the continuous analogue of the gradient. The other option is to work out $\delta W$ due to variations in the density and to collect the first order terms in the form
\begin{align}
\integ{\vecr}\frac{\delta W_{\text{H}}}{\delta\rho(\vecr)}\delta\rho(\vecr)
\end{align}
\end{exercise}

\begin{exercise}\label{ex:singleOrbital}
Calculate the \ac{KS} potential, $v_s(\vecr)$, for a singlet two-electron system. Note that only one spatial \ac{KS} orbital is occupied in this case, so you can express the \ac{KS} potential in the terms of the density and the orbital energy.
\end{exercise}

\subsection{Connecting the Kohn--Sham system to the real system}
Our current formulation of the \ac{KS} system requires an input density and this should be the density of the real interacting system of course. We will assume that every interacting $v$-representable density is also non-interacting $v$-representable, i.e.\ that a potential $v_s(\vecr)$ exists which is able to make the two densities equal. This is still an open question for the $T_s[\rho]$ under consideration here~\eqref{eq:TsDef}. However, there exists a suitable generalisation for the kinetic energy functional, which avoids this potential problem partially~\autocite{Lieb1983, DreizlerGross1990}.

Assuming that always a $v_s(\vecr)$ can be found which makes the density of the non-interacting system equal to the one of the interacting system $\rho_s(\vecr) = \rho(\vecr)$, we will derive an expression for the \ac{KS} potential that takes care of this. We first note that for a given density, the optimal orbitals and Lagrange multipliers for the Lagrangian $L_{\rho}$ are functionals of the density, so we write these optimal quantities as $\phi_k[\rho](\vecx)$, $v_s[\rho](\vecr)$ and $\mat{\epsilon}[\rho]$. Further note that the Lagrangian at these optimal values exactly equals $T_s[\rho]$, i.e.
\begin{align}
T_s[\rho] = L_{\rho}\bigl[\{\phi_i[\rho],\phi^*_i[\rho]\}, v_s[\rho], \mat{\epsilon}[\rho]\bigr].
\end{align}
Since the values of $T_s$ and $L_{\rho}$ are at each density the same, also the derivatives with respect to the density are the same, so using the chain rule we have
\begin{align}\label{eq:TsVsRel}
\frac{\delta T_s}{\delta\rho(\vecr)} &=
\sum_i\integ{\vecx'}\!\!\left(\left.\frac{\du L_{\rho}}{\du \phi_i(\vecx')}\right|_{\rho}\frac{\delta\phi_i(\vecx')}{\delta\rho(\vecr)} +
\left.\frac{\du L_{\rho}}{\du \phi^*_i(\vecx')}\right|_{\rho}\frac{\delta\phi^*_i(\vecx')}{\delta\rho(\vecr)}\right) + {} \\
&\eqspace
\integ{\vecr'}\!\!\left.\frac{\du L_{\rho}}{\du v_s(\vecr')}\right|_{\rho}\frac{\delta v_s(\vecr')}{\delta\rho(\vecr)} +
\left.\sum_{i}\frac{\du L_{\rho}}{\du \epsilon_{ij}}\right|_{\rho}\frac{\delta\epsilon_{ij}}{\delta\rho(\vecr)} + 
\left.\frac{\du L_{\rho}}{\du\rho(\vecr)}\right|_{\rho}
= -v_s[\rho](\vecr). \notag
\end{align}
We used here that the Lagrangian $L_{\rho}$ is stationary (derivatives zero) with respect to the orbitals and Lagrange multipliers at the optimum values $\phi_k[\rho](\vecx)$, $v_s[\rho](\vecr)$ and $\varepsilon_k[\rho]$, so only the term where the density appears explicitly survives.

The next step is to rewrite the total energy of the real interacting system in terms of $T_s$ as
\begin{align}
E[\rho] = F[\rho] + V[\rho] = T_s[\rho] + V[\rho] + \underbrace{F[\rho] - T_s[\rho]}_{\sAdenifeDsi E_{\text{Hxc}}[\rho]},
\end{align}
where $E_{\text{Hxc}}[\rho]$ is the Hartree-exchange-correlation energy.
Often the classical Coulomb (Hartree) part~\eqref{eq:HartreeDef} is treated explicitly and the remaining \acf{xc} energy
\mkbibfootnote{Note that \ac{xc} energy is just a fancy name for all the difficult parts of the energy which are hard to calculate for us, so it is a measure of our inability to do the real calculation. Richard Feynman therefore prefers to call the \ac{xc} energy the stupidity energy~\autocite{Feynman1972}.}
 is decomposed in a kinetic and an interaction part
\begin{align}
E_{\text{\acs{xc}}}[\rho] \isDefinedAs E_{\text{Hxc}}[\rho] - W_{\text{H}}[\rho]
= \underbrace{T[\rho] - T_s[\rho]}_{\sAdenifeDsi T_c[\rho]} + W_{\text{\acs{xc}}}[\rho].
\end{align}
Since we want to minimise the total energy of the fully interacting system, we actually want to optimise the energy of the interacting system with respect to density variations, $\delta\rho(\vecr)$. We should keep in mind, however, that we fixed the number of particles, so we have the following constraint on the allowed density variations
\begin{align}
\integ{\vecr}\delta\rho(\vecr) = 0 \, .
\end{align}
This means that if we consider first order variations in the energy due to such variations in the density, we find
\begin{align}
0 &= \delta E = \integ{\vecr}\frac{\delta E}{\delta \rho(\vecr)} \delta\rho(\vecr) &
&\Rightarrow &
c &= \frac{\delta E}{\delta\rho(\vecr)} = \frac{\delta F}{\delta\rho(\vecr)} + v(\vecr) \, .
\end{align}
So the functional derivative of $E$ with respect to the density needs to be a constant. This constant reflects the fact that the potential is uniquely determined by the density up to a constant. This called gauge freedom: shifting the potential by a constant does not change the physics of the system. A convenient convention to fix the value of the constant in the potential is to demand that $v(\vecr) \to 0$ if $\abs{\vecr} \to \infty$.

Working out the stationarity condition in terms of \ac{KS} quantities, we have
\begin{align}
c = \frac{\delta E}{\delta\rho(\vecr)}
= \frac{\delta T_s}{\delta\rho(\vecr)} + v(\vecr) + \frac{\delta E_{\text{Hxc}}}{\delta\rho(\vecr)}.
\end{align}
Combining with~\eqref{eq:TsVsRel}, we find that the \ac{KS} potential should be set to
\begin{align}
v_s[\rho](\vecr) = v(\vecr) + v_{\text{H}}[\rho](\vecr) + v_{\text{\acs{xc}}}[\rho](\vecr),
\end{align}
where we assumed now that all the potentials involved vanish at infinity, so we can set $c = 0$.
Additionally we used the classical Coulomb (Hartree) potential and the \acf{xc} potential, defined respectively as
\begin{align}
v_{\text{H}}[\rho](\vecr) &\isDefinedAs \frac{\delta W_{\text{H}}}{\delta\rho(\vecr)}
= \integ{\vecr'}\frac{\rho(\vecr')}{\abs{\vecr - \vecr'}}, &
v_{\text{\acs{xc}}}[\rho](\vecr) &\isDefinedAs \frac{\delta E_{\text{\acs{xc}}}}{\delta\rho(\vecr)}.
\end{align}
The exact $E_{\text{xc}}$ will be a very complicated functional of the density and the \ac{xc} potential will be even more complicated. Nevertheless, if we can find a good approximation to the \ac{xc} energy, we are in business and solve the \ac{KS} equations directly without solving the fully interacting Schrödinger equation.


\section{The exchange-correlation energy}
To transform the \ac{KS} equations into a practical scheme, we need to approximate the \ac{xc} energy. Therefore, we need to understand the exact $E_{\text{\ac{xc}}}$ better to rationalise the performance of various approximations, which will be treated later. To make the discussion simpler, we will first assume that the kinetic energy correction, $T_c$, is small and can be neglected. Later, we will show how the contribution from the kinetic energy can be included again.

\subsection{Holes}
To introduce the concept of holes, we will first only consider the interaction part of the \ac{xc} energy, $W_{\text{\ac{xc}}}$. The full interaction written in terms of the pair-density is
\begin{align}
W = \half\iinteg{\vecr_1}{\vecr_2}\frac{P(\vecr_1,\vecr_2)}{\abs{\vecr_1 - \vecr_2}}.
\end{align}
Since the pair-density integrates to $N(N-1)$ particles, this term describes the interaction between all the particles. So that are in total $\half N(N-1)$ interactions, since the particles do not interact with themselves and the half takes care that we only count the \emph{unique} pairs.

To express more explicitly that the particles only interact with other particles, we introduce the conditional probability density
\begin{align}
\rho(\vecr | \rref) \isDefinedAs \frac{P(\vecr,\rref)}{\rho(\rref)} .
\end{align}
The conditional probability is the probability to find a particle at $\vecr$, if we know that an other particle is located at the reference position $\rref$.

\begin{exercise}
Check that the conditional probability is normalised as you would expect ($N-1$).
\end{exercise}

The interaction term can now be expressed with the help of the conditional probability as
\begin{align}
W = \half\iinteg{\vecr}{\rref}\frac{\rho(\vecr | \rref)\rho(\rref)}{\abs{\vecr - \rref}},
\end{align}
So given the density profile of the electrons, only the interactions with the \emph{other} electrons should be taken into account, which is exactly what the conditional density achieves.

The classical Hartree part~\eqref{eq:HartreeDef}, however, takes the interaction of $\rho(\rref)$ with the full density $\rho(\vecr)$ into account instead of only the conditional density $\rho(\vecr | \rref)$. The Hartree term therefore does not only contain the interaction between the particles, but also contains an interaction of the particles with themselves. The main task of $W_{\text{\ac{xc}}}$ is to remove this self-interaction. We can imagine this correction as an interaction of the particles with a $-1$ particle, so we rewrite the \ac{xc} interaction part as
\begin{align}
W_{\text{xc}} \isDefinedAs W - W_{\text{H}}
&= \half\iinteg{\vecr}{\rref}\frac{\rho(\vecr | \rref)\rho(\rref) - \rho(\vecr)\rho(\rref)}{\abs{\vecr - \rref}} \notag \\
\label{eq:WxcWorkedOut}
&= \half\iinteg{\vecr}{\rref}\rho(\rref)\frac{\rho_{\text{xc}}(\vecr|\rref)}{\abs{\vecr - \rref}},
\end{align}
where the \ac{xc}-hole is defined as the correction to the density to obtain the conditional density
\begin{align}\label{eq:xcHoleDef}
\rho_{\text{xc}}(\vecr|\rref) \isDefinedAs \rho(\vecr | \rref) - \rho(\vecr)
= \frac{P(\vecr,\rref)}{\rho(\rref)} - \rho(\vecr).
\end{align}
The \ac{xc}-hole has the property that it contains exactly $-1$ particle
\begin{align}\label{eq:xcHoleInteg}
\integ{\vecr}\rho_{\text{xc}}(\vecr|\rref) = -1
\end{align}
as would be expected from our discussion on the main purpose of $W_{\text{\ac{xc}}}$.

One could forget about the conditional density and consider the \ac{xc}-hole as a quantity which discribes (minus) 
the shape of one electron at the reference position, among all the other electrons. Since the electrons are quantum particles, the electron is not localized at the reference position, but delocalized. Subtracting the interaction of the electrons with their \ac{xc}-hole from the Hartree term therefore eliminates the self-interaction.

The definition for the \ac{xc}-hole~\eqref{eq:xcHoleDef} can also be used to define other holes. In particular, the hole corresponding to the \ac{KS} wave function is called the \ac{x-hole}. Since the \ac{KS} wave function is consists of only one Slater determinant, its pair-density simplifies to
\begin{align}
P_s(\vecr_1,\vecr_2) = \rho(\vecr_1)\rho(\vecr_2) - \half\abs{\gamma_s(\vecr_1,\vecr_2)}^2.
\end{align}
Using this pair-density in the definition of the \ac{xc}-hole, we find that the \ac{x-hole} has the simple form
\begin{align}\label{eq:xHoleDef}
\rho_{\text{x}}(\vecr|\rref) \isDefinedAs -\half\frac{\abs{\gamma_s(\vecr,\rref)}^2}{\rho(\rref)}.
\end{align}
The \ac{c-hole} is now simply defined as the difference between the \ac{xc}-hole and the \ac{x-hole}
\begin{align}
\rho_{\text{xc}}(\vecr|\rref) \sAdenifeDsi \rho_{\text{x}}(\vecr|\rref) + \rho_{\text{c}}(\vecr|\rref).
\end{align}
In the next section we derive how the kinetic energy can be included in the \ac{xc} hole description. This has not been treated in the lecture, so will not be part of the exam. If you are interested, you can read the following section, otherwise, you can perfectly skip it. The holes including the kinetic energy effects are indicated with an additional bar, $\bar{\rho}_{\text{c}}$ and $\bar{\rho}_{\text{xc}}$ for the \ac{c-hole} and \ac{xc}-hole respectively.

\begin{exercise}
Check that
\begin{align}
\iinteg{\vecr_1}{\vecr_2}P(\vecr_1,\vecr_2) = N(N-1).
\end{align}
\end{exercise}

\begin{exercise}
Check that the definition of $\rho_{\text{xc}}$~\eqref{eq:xcHoleDef} is consistent with $W_{\text{xc}}$ and the \ac{xc}-hole integrates to exactly $-1$ particle~\eqref{eq:xcHoleInteg}.
\end{exercise}

\begin{exercise}
In this exercise we will check some properties of the pair density when the wave function is a simple Slater determinant.
\begin{subexercise}
\item Show that $P_s(\vecx_1,\vecx_2) = \rho(\vecx_1)\rho(\vecx_2) - \abs{\gamma_s(\vecx_1,\vecx_2)}^2$. This relation is only valid if the pair-density is calculated from a single Slater determinant, $\Psi_s$, so you need to use this fact.

\item Using the previous result, show $P_s(\vecr_1,\vecr_2) = \rho(\vecr_1)\rho(\vecr_2) - \half\abs{\gamma_s(\vecr_1,\vecr_2)}^2$ in the restricted case, so the spin-up and spin-down orbitals share the spatial parts as in~\eqref{eq:restrictedOrbs}.
\end{subexercise}
\end{exercise}

\begin{exercise}
Check that the \ac{x-hole}, $\rho_{\text{x}}(\vecr|\rref)$, integrates to $-1$ particle. Assume that the \ac{KS} \ac{1RDM} is of restricted form~\eqref{eq:equalSpin1RDM}. To what number of particles does the \ac{c-hole} integrate?
\end{exercise}

\begin{exercise}\label{ex:xHoleTwoElec}
Show that the \ac{x-hole} for a singlet two-electron system can be calculated to be $\rho_{\text{x}}(\vecr|\rref) = -\rho(\vecr)/2$. What do you notice?
\end{exercise}

\subsection{Including the kinetic energy}
What about the kinetic energy part of $E_{\text{xc}}$? We expect $T_{\text{c}}$ to be small, so we will aim to write it as a small modification of the \ac{xc}-hole introduced before. To calculate $T_{\text{c}}$ we need to connect the non-interacting \ac{KS} system with the fully interacting system. We do this by considering systems with a rescaled interaction, $\lambda\hat{W}$, and a local potential, $v_{\lambda}$, which is adjusted such that the density is equal to the fully interacting density at all coupling strengths, $\rho_{\lambda} = \rho_1 = \rho$. Note that for $\lambda = 0$ we exactly recover the \ac{KS} system. The energy at arbitrary $\lambda$ is
\begin{align}
E_{\lambda} = \brakket[\big]{\Psi_{\lambda}}{\hat{H}_{\lambda}}{\Psi_{\lambda}}
= \brakket[\big]{\Psi_{\lambda}}{\hat{T} + \hat{V}_{\lambda} + \lambda\hat{W}}{\Psi_{\lambda}}.
\end{align}
The variation in the energy when changing the coupling strength, $\lambda$, can be evaluated as
\begin{align}\label{eq:dEdLambda}
\frac{\du E_{\lambda}}{\du\lambda}
&= \brakket[\bigg]{\Psi_{\lambda}}{\frac{\du\hat{H}_{\lambda}}{\du\lambda}}{\Psi_{\lambda}}
\quad \text{(Hellman--Feynman)}\notag \\
&= \brakket[\bigg]{\Psi_{\lambda}}{\frac{\du\hat{V}_{\lambda}}{\du\lambda}}{\Psi_{\lambda}} +
\brakket[\big]{\Psi_{\lambda}}{\hat{W}}{\Psi_{\lambda}} \notag \\
&= \integ{\vecr}\rho(\vecr)\frac{\du v_{\lambda}(\vecr)}{\du\lambda} +
\half\iinteg{\vecr_1}{\vecr_2}\frac{P_{\lambda}(\vecr_1,\vecr_2)}{\abs{\vecr_1 - \vecr_2}}
\end{align}
Now we use the fundamental theorem of calculus~\autocite{Almbladh1972, LangrethPerdew1975, GunnarssonLundqvist1976}. To write the energy difference between the \ac{KS} system and the interacting system as
\begin{align}
E_1 - E_0
&= \binteg{\lambda}{0}{1}\frac{\du E}{\du \lambda} \notag \\
&= \integ{\vecr}\rho(\vecr)\bigl(v_1(\vecr) - v_0(\vecr)\bigr) +
\half\iinteg{\vecr_1}{\vecr_2}\frac{\bar{P}(\vecr_1,\vecr_2)}{\abs{\vecr_1 - \vecr_2}},
\end{align}
where the coupling constant integrated\slash{}averaged pair-density is defined as
\begin{align}
\bar{P}(\vecr_1,\vecr_2) \isDefinedAs \binteg{\lambda}{0}{1}P_{\lambda}(\vecr_1,\vecr_2).
\end{align}
Subtracting the integral over the potential difference on both sides of this equation, gives the following expression for the \ac{H}\ac{xc} energy
\begin{align}
E_{\text{Hxc}} = T - T_s + W_{\text{Hxc}}
= \half\iinteg{\vecr_1}{\vecr_2}\frac{\bar{P}(\vecr_1,\vecr_2)}{\abs{\vecr_1 - \vecr_2}}.
\end{align}
Subtracting the classical Coulomb (Hartree) term from both sides, we fin the following expression for the \ac{xc} energy
\begin{align}
E_{\text{xc}}
= \half\iinteg{\vecr_1}{\vecr_2}\frac{\bar{P}(\vecr_1,\vecr_2) - \rho(\vecr_1)\rho(\vecr_2)}{\abs{\vecr_1 - \vecr_2}}.
\end{align}
Comparing this expression with $W_{\text{xc}}$~\eqref{eq:WxcWorkedOut}, we find that we only need to replace the fully interacting pair-density by the coupling constant integrated pair-density. The corresponding averaged \ac{xc}-hole, which includes the kinetic energy effects, becomes
\begin{align}
\bar{\rho}_{\text{xc}}(\vecr|\rref) \isDefinedAs \frac{\bar{P}(\vecr,\rref)}{\rho(\rref)} - \rho(\vecr).
\end{align}
Approximate functionals based on the \ac{HEG} typically include the the kinetic energy effects. The kinetic energy effects on the shape of the hole in inhomogeneous systems (anything else than the \ac{HEG}) such as molecules are not well known actually. They are expected to be small nevertheless.

\begin{exercise}
Check the Hellman--Feynman step in~\eqref{eq:dEdLambda}.
\end{exercise}

\begin{figure}[t]
  \includegraphics[width=\textwidth]{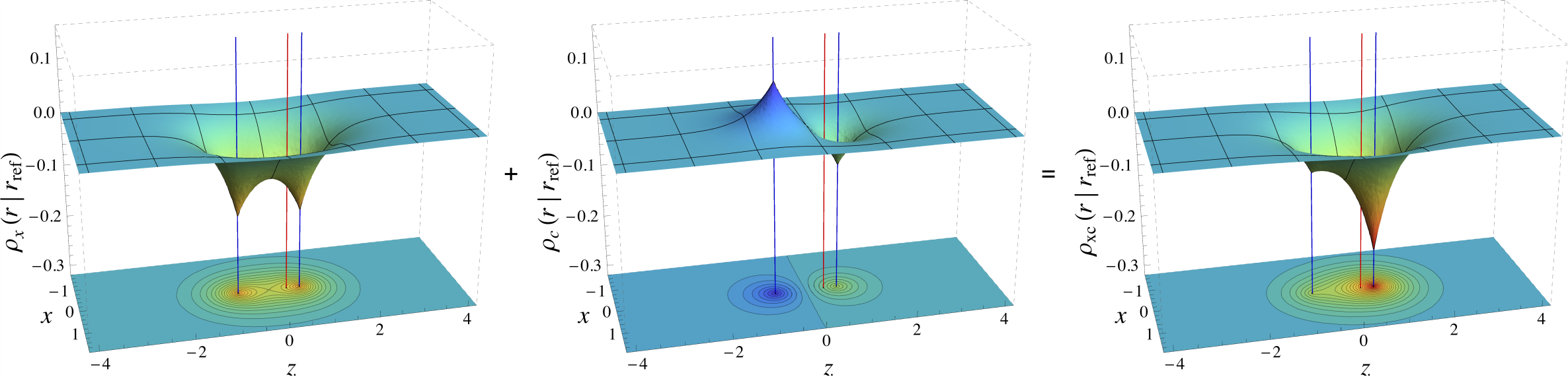}
  \caption{The different holes of the H$_2$ molecule at $R_{\text{H--H}} = 1.4$ Bohr and the reference electron at 0.3 Bohr to the left of the right nucleus along the bond axis ($\rref = (0,0,0.4)$ Bohr). The positions of the nuclei are indicated by the blue lines and the position of the reference electron is by red. The left panel shows the \ac{x-hole}, $\rho_{\text{x}}(\vecr|\rref)$, the middle panel shows the \ac{c-hole}, $\rho_{\text{c}}(\vecr|\rref)$, which provides a small correction to have the more localized real hole, $\rho_{\text{xc}}(\vecr|\rref)$.}
  \label{fig:holesEqui}
  \includegraphics[width=\textwidth]{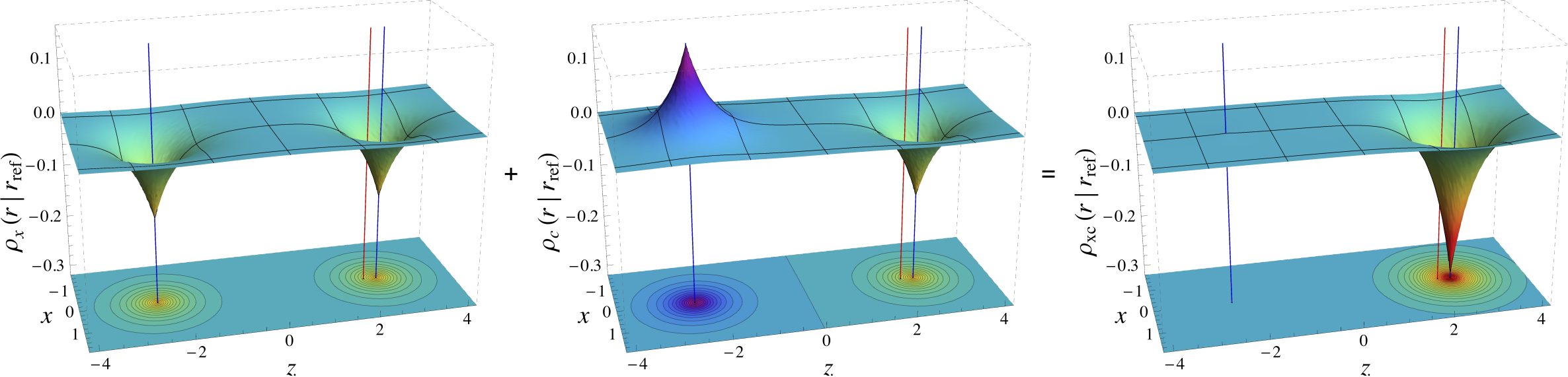}
  \caption{Similar to the previous plots, but now for $R_{\text{H--H}} = 5.0$ Bohr. The reference electron is still at 0.3 Bohr to the left of the right nucleus along the bond axis ($\rref = (0,0,2.2)$ Bohr now).}
  \label{fig:holes5}
\end{figure}

\subsection{Holes of the H$_2$ molecule}
In Exercise~\ref{ex:xHoleTwoElec} you have already calculated the \ac{x-hole} of a singlet two-electron system, so also H$_2$, to be
\begin{align}
\rho_{\text{x}}(\vecr|\rref) = -\half\rho(\vecr) \, ,
\end{align}
so the \ac{x-hole} of a singlet two-electron system is independent of the reference position. The \ac{xc}-hole can be obtained from accurate \ac{CI} calculations.
\mkbibfootnote{To include the kinetic energy effects, the averaged \ac{xc}-hole should have been calculated. This is actually quite involved procedure and this data is currently not yet available. As mentioned before, the effect of the kinetic energy is probably small and the $\lambda = 1$ hole already shows the most important physics.}
The holes for the H$_2$ molecule at equilibrium distance are shown in Fig.~\ref{fig:holesEqui}. The \ac{c-hole} only provides a small correction to the \ac{x-hole}. The reference electron is located close to the right nucleus (0.3 Bohr to the left). The main effect of the \ac{c-hole} is to localize the full \ac{xc}-hole more on the right nucleus. This is a very intuitive effect, since if the reference electron is located near the right nucleus, you also expect that the hole is located near the reference electron. Some delocalization remains due to the quantum nature of the electrons.

The localization effect of the \ac{c-hole} becomes more pronounced when the bond of the hydrogen molecule is stretched. In Fig.~\ref{fig:holes5} the holes are now shown for a H$_2$ molecule with $R_{\text{H--H}} = 5.0$ Bohr. The \ac{c-hole} is not a small correction to the \ac{x-hole} anymore. The \ac{c-hole} actually needs to be of equal magnitude to completely eliminate the hole amplitude on the left nucleus. The \ac{c-hole} has the same peak on the right nucleus with opposite sign to ensure that the full \ac{xc}-hole correctly integrates to $-1$ electron~\eqref{eq:xcHoleInteg}.

The holes give a different view on the failure of restricted \ac{HF} to describe the dissociation of the H$_2$ molecule. Since \ac{HF} only includes exchange, the hole by the \ac{HF} model remains completely delocalized, even when the hydrogen molecule is dissociated. Simple physical intuition immediately tells you that this is an incorrect description, since it is energetically unfavourable for two electrons to be near the same nucleus simultaneously. Instead, when one electron is on the left nucleus, the other should be at the right nucleus and visa versa. A fancy name for this strong correlation between the whereabouts of the electrons is `quantum entanglement'.

\begin{exercise}
Use the \ac{x-hole} to argue that the \ac{HF} energy behaves asymptotically as
\begin{align}
E_{\text{\acs{HF}}}(R_{\text{H--H}} \to \infty) = C - \frac{1}{2R_{\text{H--H}}} \, .
\end{align}
Do not forget to include the interaction between the nuclei!

Assuming that the \ac{HF} orbital becomes just a linear (gerade) combination of the hydrogenic 1s orbitals on the two hydrogen atoms, argue that the constant can be approximated as
\begin{align}
C = \lim_{\mathclap{R\to\infty}} E_{\text{RHF}}(R)
\approx 2\brakket{A}{\hat{h}}{A} + \half(AA | AA) = \zeta^2 - 2\zeta + \frac{5\zeta}{16} \, ,
\end{align}
where in the last step you can reuse the results from exercise 7 (e + k) from the HF part. Since this is the restricted \ac{HF} energy in the dissociation limit, you can simply optimise the exponent to get the numbers ($\zeta_{\text{opt}} = 27/32$ and $C_{\text{opt}} = -729/1024 \approx -0.71$).
\end{exercise}

\section{Approximations to the \acs{xc} energy}
Now that we have some understanding of what the exact \ac{xc} energy should do, we can take a look at different approximations which are used for $E_{\text{\ac{xc}}}[\rho]$. One of the main disadvantages of current \ac{DFT} is that there does not an ultimate \ac{xc} functional, which works in all cases. Therefore, hundreds of different approximations have been published all optimised for different systems and physical situations. Although many of these functionals have intimidating acronyms based on the author names, they are often just some slight reparametrizations without any newly captured physics. We will limit the discussion to the most basic \ac{xc} functional classes, which encompasses most of the functionals used in daily practice.

\subsection{The \acf{LDA}}
The \acf{LDA} is the oldest of the density functionals and its history actually dates back far before the foundations of \ac{DFT} were laid. Some people like to refer to the \ac{LDA} as the mother of all functionals. The \ac{LDA} based on the idea that if the density does not vary too strongly, we can assume that $\bar{\rho}_{\text{xc}}$ closely resembles the \ac{xc}-hole of the \ac{HEG} at the reference position. The \ac{xc}-hole of the \ac{HEG} only depends on the (constant) density, $\rho$, of the gas and the distance between the electron, $\abs{\vecr - \rref}$, since the \ac{HEG} is an isotropic system. The \ac{LDA} \ac{xc}-hole is defined as
\begin{align}
\bar{\rho}^{\text{\ac{LDA}}}_{\text{\ac{xc}}}(\vecr|\rref)
\isDefinedAs \bar{\rho}^{\text{\acs{HEG}}}_{\text{\ac{xc}}}\bigl(\rho(\rref), \abs{\vecr - \rref}\bigr) \, .
\end{align}
The remaining task is to calculate $\bar{\rho}^{\text{\acs{HEG}}}_{\text{\ac{xc}}}$. The exchange part is not too difficult to obtain, since you can use the non-interacting solution you already found in Exercise~\ref{ex:Theg}, where you calculated the kinetic energy of the \ac{HEG}. The \ac{1RDM} of the non-interacting \ac{HEG} can be worked out to be
\begin{align}\label{eq:HEG1RDM}
\gamma_s^{\text{\acs{HEG}}}(k_F,r_{12}) &= \frac{k_F^3}{\pi^2}J(k_Fr_{12}) \, , &
J(y) = \frac{j_1(y)}{y} = \frac{\sin(y) - y\cos(y)}{y^3} \, ,
\end{align}
where $j_1(y)$ is a spherical Bessel function of the first kind and $k_F^3 \isDefinedAs 3\pi^2\rho$ is the Fermi wave vector. From the non-interacting \ac{1RDM}, the \ac{x-hole} of the \ac{HEG} can readily be calculated~\eqref{eq:xHoleDef}. In fact, the exchange part of the \ac{LDA} energy was already evaluated in 1930 by Paul Dirac \autocite{Dirac1930} to be
\begin{align}\label{eq:LDAexchange}
W_{\text{x}}^{\text{\acs{LDA}}}[\rho] &= -C_{\text{x}}\integ{\vecr}\rho(\vecr)^{\nfrac{4}{3}} \, , &
C_{\text{x}} &\isDefinedAs \frac{3}{4}\biggl(\frac{3}{\pi}\biggr)^{\nfrac{1}{3}} \, .
\end{align}
Since the original treatment by Thomas and Fermi did not include exchange he proposed add this term to take exchange approximately into account. Unfortunately, the additional term $W^{\text{\ac{LDA}}}_{\text{x}}$ only makes the result from Thomas--Fermi theory even worse \autocite{Eschrig2003}, indicating that the poor kinetic energy functional, $T_{\text{TF}}$ is the main source of the bad performance and not the lack of exchange and correlation effects.

History demands we should mention a resurgence of interest in $W^{\text{\ac{LDA}}}_{\text{x}}$ already in 1951. The first electronic computers started to appear around that time and performing an actual \ac{HF} calculation for small chemical systems became feasible. The main bottleneck however, was the calculation of the non-local exchange potential of \ac{HF}. Therefore, J.C. Slater proposed to use a local approximation to the exchange potential based on the \ac{HEG} in the same way as the \ac{LDA} \autocite{Slater1951}. He called this method the X$\alpha$ method, where $\alpha$ refers to a constant in his expression for the exchange energy
\begin{align}
W_{\text{x}}^{\text{X$\alpha$}}[\rho] = \frac{3}{2}\alpha W_{\text{x}}^{\text{\acs{LDA}}}[\rho] \, ,
\end{align}
so $\alpha = 2/3$. Strangely enough, the X$\alpha$ approximation by Slater actually gave better results than \ac{HF} and got even better if the value of the parameter $\alpha$ was set to $\alpha = 0.7$. At first sight it is counterintuitive that an approximation gives better numbers than the method it is supposed to approximate. This has upset many scientists and the use of the X$\alpha$ method has remained controversial for a very long time. Only when the \ac{LDA} functional was studied in more detail, people started to understand why the simplistic X$\alpha$ performed better than \ac{HF} and could even explain why raising the value of $\alpha$ would improve the results even further.

Before we explain the superior performance of \ac{LDA} and X$\alpha$ over \ac{HF}, we should mention that an analytic expression for the correlation part is not available for the \ac{HEG}. Although the \ac{HEG} appears to be a very simple system with its constant density, all the many-body effects responsible for correlation turn out to be very complicated, even in this `simple' isotropic system. The main advantage of the \ac{HEG} is that we do not have to deal with real density functionals, but only functions of the density, since the density is constant in the \ac{HEG}. The asymptotic behaviour of the \ac{xc}-hole of the \ac{HEG} has been studied to great detail and is well understood \autocite{Gori-GiorgiPerdew2001}. Accurate quantum Monte Carlo calculations
\mkbibfootnote{Quantum Monte Carlo is a different approach to find the ground state. Instead of using Slater determinants, one uses much more complicated ansatz forms which captures important analytic features of the wave function exactly. The resulting integral with the Hamiltonian remains high dimensional and is solved by stochastic generation of integral points. This stochastic manner of solving integrals is called Monte Carlo, a name the method got during the Manhattan project (the atomic bomb).}
have supplied \ac{xc} holes for intermediate values \autocite{CeperleyAlder1980, OrtizHarrisBallone1999}, which have been combined with the asymptotic behaviour to construct accurate fits of the correlation part of the \ac{HEG}~\autocite{WangPerdew1991, Gori-GiorgiPerdew2002}. In practical \ac{LDA} functionals, these fits are used to construct the correlation part of the \ac{LDA} functional. The most used fit for the correlation part of the \ac{LDA} is the one by \ac{VWN} and is denoted as \ac{VWN}5, where the 5 stands for the 5\textsuperscript{th} variant in the (same!)\ article \autocite{VoskoWilkNusair1980}. The physics captured by the \ac{LDA} does not really change when using different approximations, so we will not go into the details of different versions for the fit of the correlation part.

\begin{figure}[t]
  \includegraphics[width=\textwidth]{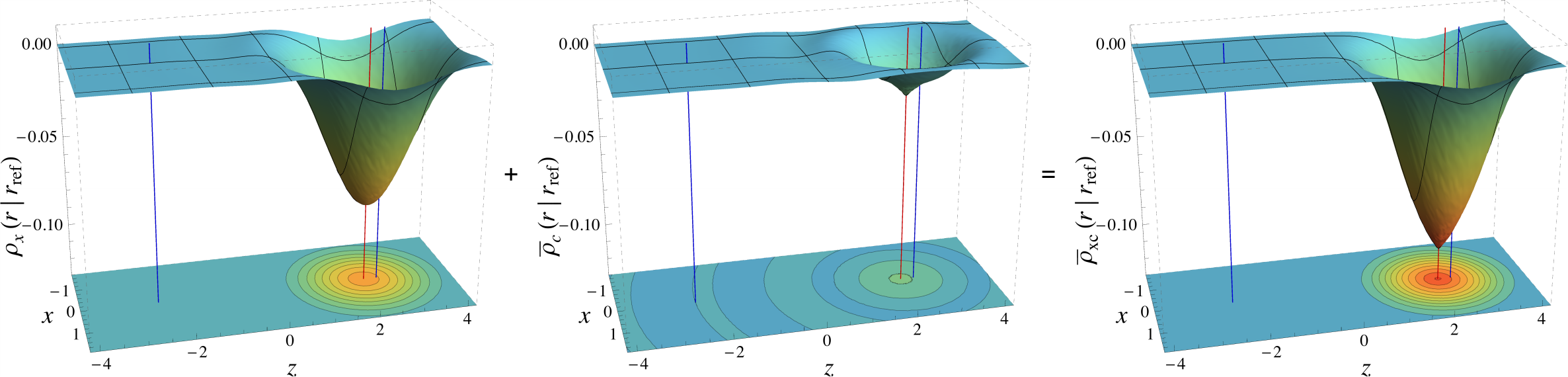}
  \caption{The \ac{LDA} holes for the H$_2$ molecule at $R_{\text{H--H}} = 5.0$ Bohr. The reference electron is again at 0.3 Bohr to the left of the right nucleus along the bond axis ($\rref = (0,0,2.2)$ Bohr).}
  \label{fig:ldaHoles}
\end{figure}

Now let us consider the \ac{LDA} holes for a stretched H$_2$ molecule at $R_{\text{H--H}} = 5.0$ Bohr depicted in Fig.~\ref{fig:ldaHoles}. There are couple of important things to notice:
\begin{itemize}
  \item The \ac{x-hole} is oscillatory due to the sines and cosines in the \ac{1RDM} from the \ac{HEG}~\eqref{eq:HEG1RDM} (see also Fig.~\ref{fig:GradientHoles} later on). Correlation removes these oscillations.
  \item The \ac{LDA} holes are spherical (only depend on $\abs{\vecr - \rref}$) and are centered around the reference electron.
  \item The \ac{LDA} \ac{x-hole} and \ac{c-hole} are very bad approximations to the exact holes and a direct comparison does not make much sense. The total \ac{LDA} hole, $\bar{\rho}_{\text{xc}}(\vecr|\rref)$, however, is not a too bad approximation to the exact \ac{xc}-hole (see Fig.~\ref{fig:holes5}).
  \item The total \ac{LDA} \ac{xc}-hole is a much better approximation to the exact \ac{xc}-hole than the \ac{x-hole} alone which is used in the \ac{HF} approximation. Although the shape of the \ac{LDA} \ac{xc}-hole is not particularly good, at least \ac{LDA} is able to describe the localization of the hole near the reference electron. This is why \ac{LDA} often outperforms \ac{HF}.
  \item The correction from the \ac{LDA} \ac{c-hole} is small. Its main contribution is to deepen the hole near the reference electron, which causes a decrease of the exchange energy, due to the $1/\abs{\vecr-\rref}$ factor in the exchange energy functional (check~\eqref{eq:WxcWorkedOut} with \ac{xc} $\to$ x). This explains why Slater got even better results by increasing the $\alpha$ factor to slightly larger values, since effectively he incorporated more correlation effects.
\end{itemize}
Due to a better localization of the hole, \ac{LDA} outperforms \ac{HF} for molecular dissociation (see Fig.~\ref{fig:energyComp}). The most important features of \ac{LDA} in practical calculations are:
\begin{itemize}
  \item The \ac{LDA} favours homogeneous systems too much, so is too eager to form bonds between atoms and molecules. The bond lengths are therefore too short in \ac{LDA} calculations, so \ac{LDA} overbinds in this sense. Nevertheless, the bond lengths are still within 1--2\% accuracy.
  \item Though better than \ac{HF}, the energy still increases to much upon dissociation (Fig.~\ref{fig:energyComp}), so binding energies and transition states tend to be too high in energy.
\end{itemize}
The \ac{LDA} only brought partial success for \ac{DFT}. The \ac{LDA} energies where not accurate enough for chemists to make useful predictions on molecules. Only when the accuracy of the energies increased sufficiently with the introduction of functionals which also depend on the gradient (\acsp{GGA}), \ac{DFT} started to be useful for chemistry. The solid state physicists were very happy with the \ac{LDA} already. In an infinite solid, it makes no sense to talk about the total energy, so the lack of accuracy of \ac{LDA} on this part was irrelevant. The most important feature of \ac{LDA} compared to \ac{HF} for the solid state physicists was that \ac{LDA} can describe metals whereas \ac{HF} can not. One can even proof rigorously that unrestricted \ac{HF} predicts all materials to be insulators \autocite{BachLiebLoss1994}. The ability of \ac{LDA} to describe both insulators and metals was therefore a major breakthrough in the solid state community.

\begin{figure}[t]
  \begin{center}
    \includegraphics[width=0.75\textwidth]{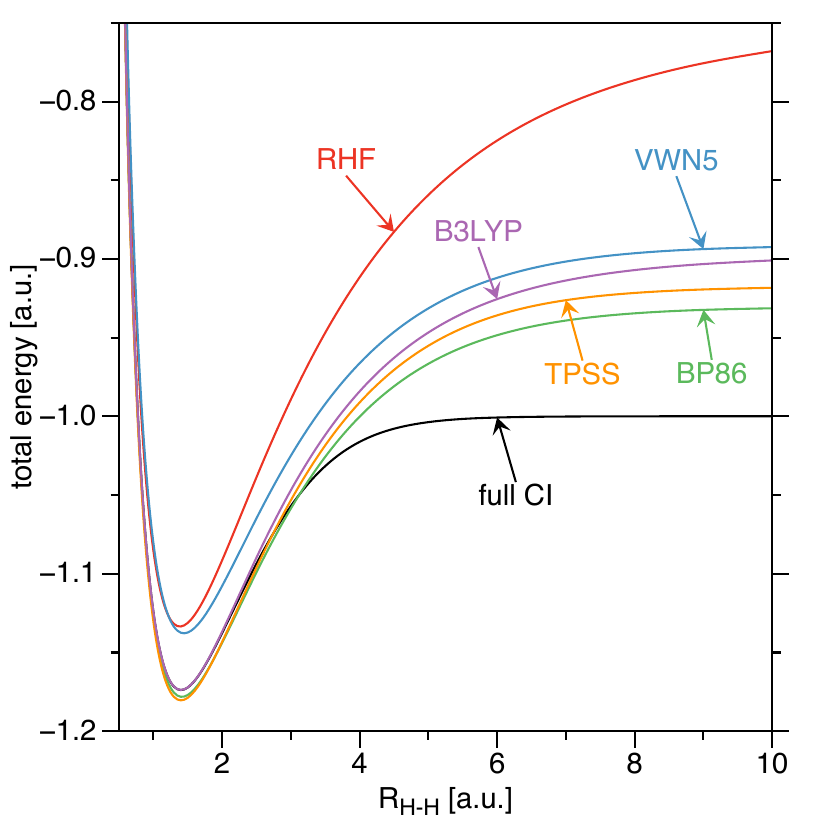} 
    \caption{Comparison of several approximations to the \ac{xc} energy to full \ac{CI} and \ac{HF} in the aug-cc-pVQZ basis for the hydrogen molecule. The hydrogen molecule is special, since the total energies from DFT seem to be very good. For other molecules the total energies are not so good, but the relative energies are.}
    \label{fig:energyComp}
  \end{center}
\end{figure}

\begin{exercise}
Show that the \ac{1RDM} of the non-interacting \ac{HEG} is indeed given by~\eqref{eq:HEG1RDM}. Due to the isotropy of the \ac{HEG}, the \ac{KS} orbitals will remain plane waves, so you can reuse the orbitals from Exercise~\ref{ex:Theg}.
\end{exercise}

\begin{exercise}
Give the expression for the \ac{x-hole} of the \ac{LDA}.
\end{exercise}

\begin{exercise}
Calculate the \ac{LDA} exchange energy, $W^{\text{\acs{LDA}}}_{\text{x}}$. The integral is not so easy to solve. Use $q(x) \isDefinedAs -\sin(x)/x$ to rewrite the integral in terms of $q(x)$, $q'(x)$ and $q''(x)$ only. You can rewrite the integrant now as a total derivative, which allows you to do the integration easily. You should get the same answer as Dirac in 1930~\eqref{eq:LDAexchange}.
\end{exercise}

\begin{exercise}
Show explicitly that the \ac{LDA} \ac{x-hole} integrates to exactly $-1$ electron. There are two ways to solve this exercise
\begin{description}
  \item[brute-force] Use the same trick as in the previous exercise to show that
  \begin{align}
    \integ{\vecr}\rho^{\text{\acs{LDA}}}_{\text{x}}(\vecr|\rref) = -\frac{2}{\pi}\binteg{x}{0}{\infty}\bigl(q(x)\bigr)^2.
  \end{align}
   If you did a course in complex analysis, you can solve the remaining integral over $q^2$ by contour integration, which gives $\pi/2$.
   
   \item[detour] Write the \ac{1RDM} back in its integral form over the wave vectors
   \begin{align}
     \gamma^{\text{\acs{HEG}}}_s(k_F,r_{12}) = \frac{1}{4\pi^3}\binteg{\veck}{\Omega_F}{}\e^{\im\veck\cdot\vecr_{12}},
   \end{align}
   where $\Omega_F$ is the Fermi sphere: the part of $\veck$-space which corresponds to occupied orbitals. Insert this expression in the definition for the \ac{x-hole}~\eqref{eq:xHoleDef} and evaluate the integration condition for the \ac{x-hole} by \emph{first} performing the integration over $\coord{u} = \vecr - \rref$. You also need to use that
   \begin{align}
     \binteg{x}{-\infty}{\infty}\e^{\im k x} = 2\pi\,\delta(k),
   \end{align}
   where $\delta(x)$ is the Dirac delta-function.
\end{description}
\end{exercise}

\subsection{The \acfp{GGA}}
The most logical step to improve the accuracy of the \ac{LDA} is to include also the gradient of the density. An approximate hole can be built by using a slightly perturbed \ac{HEG}
\begin{align}
\bar{\rho}^{\text{\acs{GEA}}}_{\text{\ac{xc}}}(\vecr|\rref)
\isDefinedAs \bar{\rho}^{\text{\acs{HEG}}}_{\text{\acs{xc}}}\bigl(\rho(\rref),\abs{\nabla\rho(\rref)},\abs{\vecr-\rref}\bigr),
\end{align}
the \acf{GEA} hole. Unfortunately this approach did not work, since the \ac{GEA} functionals always gave results worse than the \ac{LDA}. It took a long time before Perdew \autocite{Perdew1985, Perdew1986} realized that the long-range oscillations --- already present in the \ac{LDA} \ac{x-hole} (see Fig.~\ref{fig:ldaHoles} and discussion) --- were hugely enhanced and that the \ac{GEA} \ac{x-hole} is even not negative definite anymore. The solution by Perdew was simple: just remove the positive part of the \ac{x-hole}. The \ac{x-hole} integrates now to less than $-1$ electron, so he limited the extend of the \ac{x-hole} by only taking the part which integrates to $-1$ electron within a sphere centered at the reference position. In this way he could both maintain the integration condition and remove the long-range oscillations from the \ac{GEA} \ac{x-hole}. This procedure gives the \ac{GGA}.

\begin{figure}[t]
  \begin{center}
    \includegraphics[width=0.5\textwidth]{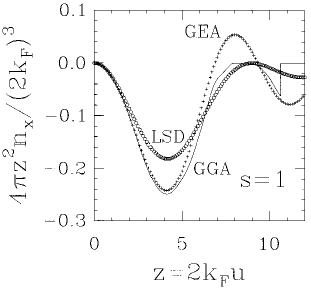}
    \caption{Spherically averaged \ac{x-hole} for the \ac{LDA} (=LSD, circles), the \ac{GEA} (crosses) and the \ac{GGA} (solid line). The plot is taken from \autocite{PerdewBurkeWang1996}, where $n_x =\rho_{\text{x}}$ denotes the \ac{x-hole} and $u = \abs{\vecr-\rref}$ is the inter-electronic distance and $s \isDefinedAs \abs{\nabla\rho}/(2k_F\rho)$ is a dimensionless version of the gradient.}
    \label{fig:GradientHoles}
  \end{center}
\end{figure}

As an example to illustrate all these features, we show spherically averaged \acp{x-hole} for the \ac{LDA} (=LSD), \ac{GEA} and \ac{GGA} in Fig.~\ref{fig:GradientHoles} which is taken from Ref. \autocite{PerdewBurkeWang1996}. You clearly see that the oscillations present in the \ac{LDA} \ac{x-hole} are enhanced in the \ac{GEA} \ac{x-hole} to such extend that the \ac{GEA} \ac{x-hole} has positive parts. In the \ac{GGA} \ac{x-hole} these are removed and only the inner part is retained which integrates to $-1$ electron. This gives a rather ridiculous shape for the \ac{GGA} \ac{x-hole}, but that is mainly in the outer region. Since the exchange energy mainly probes the inner region due to the $1/\abs{\vecr-\rref}$ factor, these irregularities in the outer part do not affect the exchange energy too much.
The inner region is of the \ac{x-hole} is actually improved by including gradient dependent terms and hence, also the prediction for the energy compared to the \ac{LDA}.
\mkbibfootnote{This is exactly what you expect for a perturbative expansion. By including higher order derivatives you improve the description close to your reference and the description far away can be better or worse (almost exclusively worse in practice, e.g.\ the \ac{MP} perturbation series or the failing perturbative approach to deal with the strong force between quarks (chromodynamics).}
The main features of the \acp{GGA} are:
\begin{itemize}
  \item The gradient part of the \acp{GGA} favours inhomogeneous systems more, so corrects for the overbinding of the \ac{LDA}. The \acp{GGA} tend to overcorrect the bond lengths, so the accuracy remains the same 1--2\% of \ac{LDA}.
  \item Since homogeneous systems are not favoured so much anymore, binding energies and transition state barriers are improved, which made the \ac{GGA} useful for chemistry.
  \item Core electrons are treated better, so \acp{GGA} give better total energies than the \ac{LDA}.
\end{itemize}
The total energy for the oldest successful \ac{GGA} functional, the \acs{BP86}, %
\makeatletter%
\AC@placelabel{B88}%
\AC@placelabel{BP86}%
\makeatother%
are shown also in Fig.~\ref{fig:energyComp} for our test system, H$_2$. The correlation part of the \acs{BP86} is the one originally proposed by Perdew \autocite{Perdew1986} and the exchange part was replaced by the \acs{B88} exchange functional proposed by Becke \autocite{Becke1988}, since it gave better numbers. During the years there have been efforts to simplify the parametrization of the \acs{BP86} and has lead to the \acs{PW91} functional \autocite{WangPerdew1991, PerdewChevaryVosko1992} and was simplified even more in the \acs{PBE} functional \autocite{PerdewBurkeErnzerhof1996}. 
\makeatletter%
\AC@placelabel{PW91}%
\AC@placelabel{PBE}%
\makeatother%
The physical idea remains the same. There is only a difference in parametrization strategy, which results in different numbers.

\subsection{The \acsp{meta-GGA}}

The next logical step is to include the second order derivative of the density, the Laplacian $\nabla^2\dens(\vecr)$. Approximate functionals which also include higher order derivatives of the density are called \acused{meta-GGA}\acp{meta-GGA}. Direct calculation of $\nabla^2\dens(\vecr)$ leads to numerical problems for code which are based on Gaussian basis sets. To avoid these numerical problems, one use often the \ac{KS} kinetic energy-density instead, which is defined as
\begin{align}
  \tau_s(\vecx) \isDefinedAs \half \sum_{i=1}^N\abs{\nabla\phi_i(\vecx)}^2.
\end{align}
The use of the \ac{KS} kinetic energy density leads to the difficulty that the functional now becomes an orbital dependent functional. In practice, this additional complication is simply neglected.

\makeatletter%
\AC@placelabel{meta-GGA}%
\AC@placelabel{TPSS}%
\AC@placelabel{SCAN}%
\makeatother%

The \acp{meta-GGA} tend not to make a significant improvement for covalent bonds, but do improve the description of weak bonds \autocite{SunXiaoFang2013}. In practice, their performance for weak bonds still requires the additional use of empirical dispersion correction schemes, limiting the actual use of \acp{meta-GGA} in practice. The most well-known older \acp{meta-GGA} is the \acsu{TPSS} functional \autocite{TaoPerdewStaroverov2003} and a more recent popular \ac{meta-GGA} is the \acs{SCAN} functional. A key ingredient of the \acs{SCAN} functional is recognition that with the help of $\tau_s(\vecr)$ regions of single-orbital character, slowly varying density and overlap of closed shells can be recognized by the following parameter 
\begin{equation}
\alpha(\vecr) = \frac{\tau_s(\vecr) - \tau_{\text{W}}(\vecr)}{\tau^{\text{\acs{HEG}}}(\vecr)} ,
\end{equation}
where $\tau_{\text{W}}$ is the single orbital (Von Weizsäcker) kinetic energy density (cf.\ exercise~\ref{ex:singleOrbital})
\begin{equation}
\tau_{\text{W}}(\vecr) = \frac{1}{8}\frac{\abs{\mat{\nabla}\dens(\vecr)}^2}{\dens(\vecr)}
\end{equation}
and $\tau^{\text{\acs{HEG}}}(\vecr)$ is the kinetic energy density of the non-interacting \ac{HEG}
\begin{align}
\tau^{\text{\acs{HEG}}}(\vecr) &= C_{\text{TF}}\dens(\vecr)^{\nfrac{5}{3}} ; &
C_{\text{TF}}&= \frac{3}{10}\bigl(3\pi^2\bigr)^{\nfrac{2}{3}} .
\end{align}
That a \ac{meta-GGA} would be able to distinguish between regions with single orbital character ($\alpha \approx 0$) and slowly varying density ($\alpha \approx 1$) was already recognized very early by Becke \autocite{Becke1998} and incorporated in the early \acp{meta-GGA} such as the \acs{TPSS} functional. About a decade later, it was realized that the description of weak bonds (overlap of closed shells) could be regonized by the parameter $\alpha$, since in that case $\alpha \gg 1$ \autocite{SunXiaoFang2013}. Still, the description of dispersion interactions significantly improves by including ad-hoc classical dispersion interaction corrections.

\begin{exercise}
Derive a relation between $\tau_s(\vecr)$ and $\nabla^2\rho(\vecr)$. First show that $\tau_s(\vecx)$ and $\nabla^2\dens(\vecx)$ are related as
\begin{align}
\tau_s(\vecx) = \sum_{i=1}^N\varepsilon_i\abs{\phi_i(\vecx)}^2 - v_s(\vecr)\dens(\vecx) + \frac{1}{4}\nabla^2\dens(\vecx).
\end{align}
Now you can integrate out the spin coordinate to find the desired relation.
\end{exercise}

\subsection{The hybrid functionals}
Typically we are interested in molecules at their equilibrium geometries. For H$_2$ we saw that the full \ac{xc}-hole does not completely localize on the the nucleus where the reference electron is located (Fig.~\ref{fig:holesEqui}). A small peak of the \ac{x-hole} remains behind on the other nucleus. To incorporate this effect, one can include a small percentage of $\rho_{\text{x}}$ (also called `exact' exchange) in the functional. The amount is typically fixed. It is clear from the holes of the H$_2$ molecule that the inclusion of exact exchange gives an improved description of the \ac{xc}-hole at short distances compare Figs~\ref{fig:holesEqui},~\ref{fig:ldaHoles} and~\ref{fig:xLDAhole2}). For stretched bonds, however, we inherit the delocalization error of \ac{HF} and the description becomes worse (compare Figs~\ref{fig:holes5},~\ref{fig:ldaHoles} and~\ref{fig:xLDAhole5}).

\begin{figure}[t]
  \begin{tabular}{@{}l@{\hspace{0.04\textwidth}}r@{}}
    \begin{minipage}[t]{0.48\textwidth}
      \includegraphics[width=\textwidth]{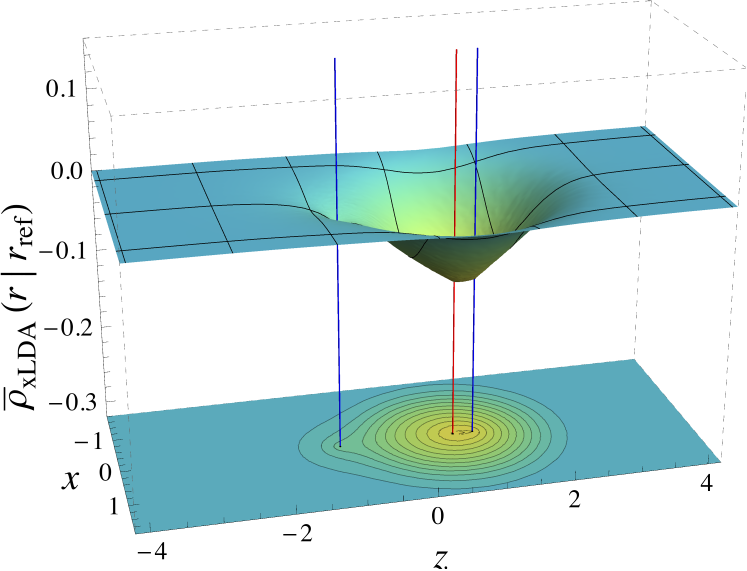}
    \end{minipage} &
    \begin{minipage}[t]{0.48\textwidth}
      \includegraphics[width=\textwidth]{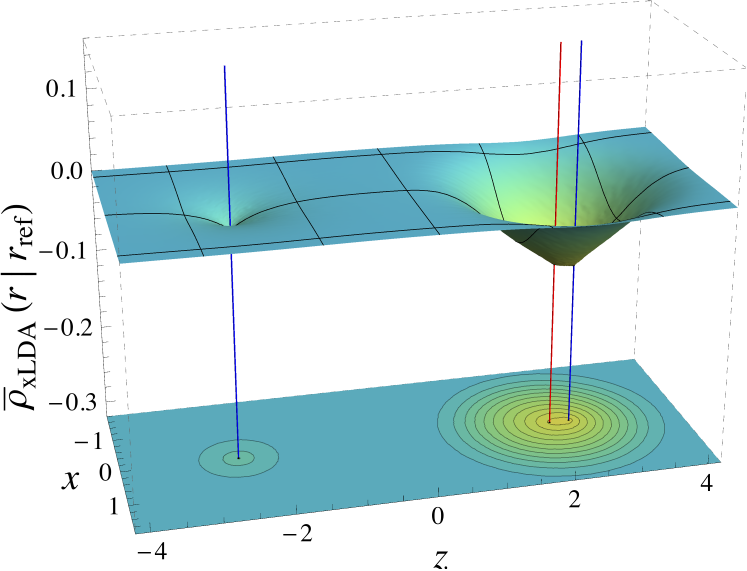}
    \end{minipage} \\
    \begin{minipage}[t]{0.48\textwidth}
      \caption{80\% of the LDA \ac{xc}-hole mixed with 20\% of the exact \ac{x-hole} at $R_{\text{H--H}} = 2.0$ Bohr. The reference electron is again located at at 0.3 Bohr to the left of the right nucleus ($\rref = (0,0,0.7)$ Bohr).}
      \label{fig:xLDAhole2}
    \end{minipage} &
    \begin{minipage}[t]{0.48\textwidth}
      \caption{The same xc hole model (80\% LDA and 20\% exact exchange) for $R_{\text{H--H}} = 5.0$ Bohr.}
      \label{fig:xLDAhole5}
    \end{minipage}
  \end{tabular}
\end{figure}

The improved description of the \ac{xc}-hole at short bond distances also improves the energy near equilibrium. The worse description of the \ac{xc}-hole at long bond distances deteriorates the energy, as is clear from the comparison of total energy from the \acs{B3LYP} functional with the other functionals and \ac{HF} in Fig.~\ref{fig:energyComp}. The \acs{B3LYP} functional performs extremely good near the equilibrium distance of H$_2$, but upon dissociation the \acs{B3LYP} hole does not fully localize and is outperformed by one of the oldest \acp{GGA}: the \ac{BP86}.

\subsection{The \acs{B3LYP} functional}
\makeatletter%
\AC@placelabel{LYP}%
\AC@placelabel{B3LYP}%
\AC@placelabel{B3PW91}%
\makeatother%
The \acs{B3LYP} functional is the most used and well known functional in chemistry, though one might wonder if it deserves this honour. The tale of the \ac{B3LYP} functional is a strange one. It starts with a hybrid functional originally proposed by Becke in 1993~\autocite{Becke1993} with 3 empirical parameters to mix several \acp{GGA} and \ac{LDA} with exact exchange
\begin{align}
E^{\text{\acs{B3PW91}}}_{\text{\acs{xc}}} &= E^{\text{\acs{VWN}5}}_{\text{\acs{xc}}} +
a_0\bigl(E^{\text{exact}}_{\text{x}} - E^{\text{\acs{VWN}5}}_{\text{x}}\bigr) + {} \notag \\
&\eqspace
a_x\bigl(E^{\text{\acs{B88}}}_{\text{x}} - E^{\text{\acs{VWN}5}}_{\text{x}}\bigr) +
a_c\bigl(E^{\text{\acs{PW91}}}_{\text{c}} - E^{\text{\acs{VWN}5}}_{\text{c}}\bigr),
\end{align}
Becke fitted the parameters $a_0$, $a_x$ and $a_c$ to a set of thermodynamic data. This functional from Becke has been implemented in \textsc{Gaussian} --- the most popular quantum chemistry package --- with the following modifications
\begin{itemize}
  \item The \ac{GGA} correlation part from the \ac{PW91} functional was replaced by the correlation from the \ac{LYP} functional.
  \item The \ac{VWN}5 \ac{LDA} parametrization was replaced by the \ac{VWN}3. This was probably a mistake, since \ac{VWN}3 has the wrong asymptotic behaviour and \ac{VWN} recommended \emph{not} to use this parametrization.
\end{itemize}
Such a large modification of the \ac{B3PW91} would normally require a refitting of the empirical parameters, but this has not been done. Although the proper motivations for the \ac{B3LYP} functional are virtually absent, the \ac{B3LYP} functional is the most used approximation for the \ac{xc} energy in chemistry.

%
%
%

\section{Is there any meaning to the \ac{KS} system?}

The \ac{KS} system was introduced to calculate $T_s[\rho]$ as an approximation to the real kinetic energy functional $T[\rho]$. Is there more physical meaning to the \ac{KS} system?

\subsection{The occupied \acs{KS} orbital energies}
The energy of the \ac{HOMO} in finite systems is exactly equal to minus the first ionisation energy
\begin{align}
\varepsilon_{\text{\acs{HOMO}}} = -I_0 = E_0(N) - E_0(N-1) \, .
\end{align}
The first step to prove this relation is to show that the density decays as
\begin{align}
\dens(r \to \infty) \sim \e^{-2\sqrt{2I_0}r}
\end{align}
far away from the system. The derivation is beyond the scope of this lecture, but can be found in~\autocite{AlmbladhBarth1985} and see also~\autocite{KatrielDavidson1980, Hoffmann-Ostenhof1977}. The next step is to realize that the asymptotic behaviour of the density in the \ac{KS} system is governed by the asymptotic decay of the \ac{HOMO}, since the \ac{KS} orbitals asymptotically decay as
\begin{align}\label{eq:asympKSorb}
\psi_k(r \to \infty) \sim \e^{-\sqrt{-2\epsilon_k}r} \, .
\end{align}
All the other occupied orbitals have a more negative orbital energy, so they decay faster than the \ac{HOMO} when $r \to \infty$. The asymptotic behaviour of the \ac{KS} density is therefore dictated by the \ac{HOMO}.

In practice, however, the \ac{LDA} \& \ac{GGA} \ac{HOMO} energies are $\sim 4-5$ eV too high. This too high value for the \ac{HOMO} energy is caused by the too fast decay of their corresponding \ac{xc} potentials. The exact $v_{\text{\acs{xc}}}(\vecr)$ decays as $-1/r$ for neutral systems, since if the electron is pulled away, it leaves a system behind with one electron less, so an effective positive charge.

The other occupied orbital energies provide good approximations to the ionisation to excited ion states, $I_k = E_k(N-1) - E_0(N)$, if these excited ion states are well described by a single Slater determinant~\autocite{ChongGritsenkoBaerends2002, GritsenkoBaerends2002}. This statement only holds for the exact \ac{xc} potential, which is in general badly approximated by \ac{LDA} \& \ac{GGA} functionals.

For this reason, special approximations for the \ac{xc} potential have been constructed which have explicitly built in the correct asymptotic decay. The first potential with the correct asymptotic decay is the \acs{LB94} potential~\autocite{LeeuwenBaerends1994}. Later the \acs{SAOP} potential~\autocite{GritsenkoSchipperBaerends1999, SchipperGritsenkoGisbergen2000} was developed to also improve the description of the inner part of the \ac{xc} potential, which is important of the other occupied orbital energies.
\makeatletter%
\AC@placelabel{LB94}%
\AC@placelabel{SAOP}%
\makeatother%
Unfortunately, only an expression for the \ac{xc} potential is provided and a corresponding $E_{\text{\acs{xc}}}$ is not available. Nevertheless, if only the orbital energies are of interest, especially the \acs{SAOP} potential typically provides very good \ac{KS} orbital energies.

\begin{exercise}
Show that the \ac{LDA} x-potential decays as $\e^{-\nfrac{2}{3}\sqrt{2I_0}r}$ for the exact density.
\end{exercise}

\begin{exercise}
Derive that $\varepsilon_{\text{\acs{HOMO}}} = -I_0$ by considering the \ac{KS} equation in the asymptotic limit. You first need to show that the \ac{KS} orbitals indeed decay as in~\eqref{eq:asympKSorb} and subsequently you can use the the \ac{KS} density should be equal to the exact density (by construction), so also its asymptotic behaviour.
\end{exercise}

\subsection{The unoccupied \acs{KS} orbital energies}

Since the \ac{KS} orbitals are generated by a local potential, also the unoccupied \ac{KS} orbitals feel an $N-1$ system for large $r$. Note that this situation is different from the unoccupied \ac{HF} orbitals which feel an $N$ particle system. This causes the \ac{HF} unoccupied orbital energies to provide approximations to affinities via Koopmans' theorem~\autocite{Koopmans1934}.
Since the unoccupied \ac{KS} orbitals feel an $N-1$ system, their energies do \emph{not} provide approximations to affinities, but their energy differences with the occupied \ac{KS} energies provide good approximations to local excitations, if the excited state can be well described by a single Slater determinant. Local valence excitations are already quite good on the \ac{LDA} \& \ac{GGA} level, but the Rydberg excitations are problematic. Because the \ac{LDA} \& \ac{GGA} potentials decay too fast the unoccupied \ac{KS} are too high in energy which cause the Rydberg excitations to be unbound states on the \ac{LDA} \& \ac{GGA} level. As you might expect, the model potentials \acs{LB94} and \acs{SAOP} are very effective to cure this deficiency of the \ac{LDA} \& \ac{GGA} potentials. Especially \ac{SAOP}, which is a more sophisticated version of \ac{LB94}.

\section{Epilogue}
These lecture notes provided a short introduction into \ac{DFT}, too short to highlight all its subtleties and difficulties. Though great care has been taken to maximize validity while retaining simplicity for an introductory course, there are some incorrect assumptions which should at least be mentioned.
\begin{itemize}
\item There are two important classes of functionals that would deserve attention in a more extended course on \ac{DFT}. The first class of functionals are the orbital depend functionals. Since the \ac{KS} orbitals are also dependent on the density, orbital functionals can also be used as density functionals, though great care is needed when taking functional derivatives and one typically needs to use the \ac{OEP} method which is a numerically unstable procedure which needs additional care to stabilize. The advantage of orbital depend functionals is that the can be constructed in a more systematic manner.

The second class of functionals are the so-called \acp{WDA}. Instead of pretending the density to be constant around the reference position, these density functional approximation cut the hole out of the true density~\autocite{GunnarssonJonsonLundqvist1977,GunnarssonJonsonLundqvist1979}. The \acp{WDA} performs about equally good as the \acp{GGA}, but at much more computational cost, so they have been not been very popular. However, recently there has been a resurgence of interest~\autocite{BahmannErnzerhof2008, Cuevas-SaavedraChakrabortyAyers2012, Cuevas-SaavedraChakrabortyRabi2012, GiesbertzLeeuwenBarth2013, AntayaZhouErnzerhof2014, PecechtlovaBahmannKaupp2014, PrecechtelovaBahmannKaupp2015}.

\item The Levy--Lieb functional $F_{\text{LL}}[\rho]$ is not convex, so a global minimiser is not guaranteed. A suitable convex generalisation is~\autocite{Lieb1983}
\begin{align}
F_{\text{L}}[\rho]
= \min_{\hat{\Gamma} \to \rho}\Trace\bigl\{\hat{\Gamma}(\hat{T} + \hat{W})\bigr\} \, ,
\end{align}
where $\hat{\Gamma} = \sum_Kw_K\ket{\Psi_K}\bra{\Psi_K}$ is a density operator with $w_K \geq 0$ and $\sum_K w_K = 1$.

\item The functional $F_{\text{L}}[\rho]$ is only differentiable at $v$-representable densities, so the $v$-representability question is still an important issue~\autocite{Leeuwen2003, Lammert2006a}.

\item There are several regularisation techniques to deal with $v$ represent\-ability: course graining by Lammert~\autocite{Lammert2006b, Lammert2010} and Moreau--Yosida regularisation~\autocite{KvaalEkstromTeale2014}.

\item In the \ac{KS} construction it is assumed that all interacting $N$-rep\-resent\-able densities are non-interacting $v$-representable. This is not true~\autocite{Lieb1983}.

\item The \ac{KS} kinetic energy should be generalised by extending the search to density matrices instead of only pure states. The \ac{KS} system will not be solved by a single Slater determinant anymore~\autocite{Lieb1983}. This can even be an issue for real systems~\autocite{SchipperGritsenko1998}.

\item In finite basis sets the \ac{KS} typically degenerates to \ac{1RDM} functional theory. As a simple example, consider the ground state of H$_2$ in a minimal basis. The \ac{KS} orbital is readily constructed from the \ac{CI} density as
\begin{align}
\phi(\vecr) = \sqrt{\cos^2(\theta)\abs{\sigma_g(\vecr)}^2 + \sin^2(\theta)\abs{\sigma_u(\vecr)}^2} \, .
\end{align}
Typically, only a complete basis will manage to reproduce this orbital at all distances, i.e.\ for all $\theta$. Otherwise, the \ac{KS} system will yield a superposition of the $(\sigma_g)^2$ and $(\sigma_u)^2$ states and reproduce even the exact \ac{1RDM}.

\item In the recent years much progress has been made in the investigation of the properties of the functional~\autocite{SeidlPerdewLevy1999, Seidl1999, SeidlGori-GiorgiSavin2007, Gori-GiorgiVignaleSeidl2009, SeidlGori-GiorgiSavin2007}.
\begin{align}
W_{\text{\acs{SCE}}}[\rho]
= \min_{\hat{\Gamma} \to \rho}\Trace\bigl\{\hat{\Gamma}\,\hat{W}\bigr\} \, ,
\end{align}
which can be considered as the counterpart of the \ac{KS} kinetic energy functional. Indeed we have the obvious inequality
\begin{align}
F[\rho] \geq T_s[\rho] + W_{\text{\acs{SCE}}}[\rho] .
\end{align}
As only the interaction remains, the electrons become effectively a classical system constrained to yield a smooth density. To minimise the electrostatic energy, the electrons are forced to move in a strictly correlated manner, sometimes referred to as a `floating' Wigner crystal. This is called the \acf{SCE} limit.

An important feature is that instead of higher order derivatives of the density, the functional use a cumulant function as an important ingredient (at least in 1D)
\begin{align}
N_e(x) = \binteg{y}{-\infty}{x}\rho(y) \, .
\end{align}

\end{itemize}

\appendix

\chapter{Prolate spheroidal coordinate system}
\label{ap:prolSperCoords}

For quantum problems with two nuclei, its is often convenient to work with the prolate spheroidal coordinate system. As there are two nuclei, a spherical coordinate system seems out of place as we have two nuclei which we would like to place at the focus. Therefore, one rather starts from ellipses which have two foci at which we can place both nuclei. By adding hyperbolae we get a 2D coordinate system which is called an elliptic coordinate system. Simply revolving it around the bonding axis (long axes of the ellipses) we obtain the prolate spheroidal coordinate system.

\begin{figure}[t]
\centering
\begin{tikzpicture}
  \draw[->] (-5,0) -- (5,0) node[anchor=west] {$z$};
  \draw[->] (0,-3.5) -- (0,3.5) node[anchor=south] {$x$};;
  \node[vertex, label=south:$\vecR_A$] (fa) at (-3,0) {};
  \node[vertex, label=south:$\vecR_B$] (fb) at (3,0) {};
  \foreach \xrad in {3.5,4} \draw (0,0) ellipse [x radius=\xrad, y radius={sqrt(\xrad*\xrad - 9)}];
  \node[rotate=8] at (0.6,-1.45) {$\xi = \frac{7}{6}$};
  \node[rotate=8] at (0.7,-2.3) {$\xi = \frac{4}{3}$};
  \foreach \et/\ettext in {{1/2}/{\half},{-1/3}/{\frac{1}{3}}}
    \draw plot[variable=\mu, domain=-1:1]  ({3*\et*cosh(\mu)},{3*sqrt(1-\et*\et)*sinh(\mu)})
      node[above] {$\eta = \ettext$};
  \node[vertex, red, label=north east:{\color{red}$\vecr$}] (r) at (7/4,{sqrt(39)/4}) {};
  \draw[red] (fa) -- node[pos=0.55, above] {$r_a$} (r) -- node[right] {$r_b$} (fb);
  \node[vertex, blue, label=north west:{\color{blue}$\vecr'$}] (rp) at ({-4/3},{sqrt(56)/3}) {};
  \draw[blue] (fa) -- node[pos=0.3, right] {$r'_a$} (rp) -- node[pos=0.2, above] {$r'_b$} (fb);
\end{tikzpicture}
\caption{Elliptic coordinates for the points $\vecr$ and $\vecr'$.}
\label{fig:ellipticSystem}
\end{figure}

Let us first focus on the elliptic system in the plane. In Fig.~\ref{fig:ellipticSystem} we show two points, $\vecr$ and $\vecr'$ and the corresponding intersection ellipses and hyperbolae. The elliptic coordinates of a points $\vecr$ are easy to calculate from the distance to the foci, $r_{A/B} \isDefinedAs \abs{\vecr - \vecR_{A/B}}$,
\begin{align}
\xi &\isDefinedAs \frac{r_A + r_B}{R} &
&\text{and}&
\eta &\isDefinedAs \frac{r_B - r_A}{R} ,
\end{align}
where $R \isDefinedAs \abs{\vecr_A -\vecr_B}$ is the distance between the foci.

You see that each elliptic coordinate $(\xi,\eta)$ results in an intersection in both the upper half plane and the lower half plane, so it corresponds to two points. With some fiddling around, you can also establish the elliptic reverse transformation to be
\begin{align}
x &= \frac{R}{2}\sqrt{(\xi^2 - 1)(1 - \eta^2)} &
&\text{and}&
z &= \frac{R}{2}\xi\eta .
\end{align}
Note that we only recover the points in the upper half plane with this parametrisation. To generate the full 3D space, we turn the elliptic system around the $z$-axis with an angle $\phi \in [0,2\pi)$. This is called the prolate spheroidal coordinate system.\footnote{Turning the elliptical system around the $x$-axis results in the oblate spheroidal coordinate system.} This immediately also generates the coordinates in the negative half plane
\begin{align}
\vecr &= \begin{pmatrix} x \\ y \\ z \end{pmatrix}
= \frac{R}{2}\begin{pmatrix} \sqrt{(\xi^2 - 1)(1 - \eta^2)}\cos(\phi) \\
\sqrt{(\xi^2 - 1)(1 - \eta^2)}\sin(\phi) \\
\xi\eta
\end{pmatrix} .
\end{align}
The prolate spheroidal basis vectors are readily obtained as
\begin{multline}
\begin{pmatrix} \vece_{\xi} &\vece_{\eta} &\vece_{\phi} \end{pmatrix}
= \mat{J}
= \begin{pmatrix}\frac{\du\vecr}{\du\xi} &\frac{\du\vecr}{\du\eta} &\frac{\du\vecr}{\du\phi}
\end{pmatrix}  \\
= \frac{R}{2}\begin{pmatrix}
\xi\sqrt{\frac{1-\eta^2}{\xi^2-1}}\cos(\phi)	&-\eta\sqrt{\frac{\xi^2-1}{1-\eta^2}}\cos(\phi)
	&-{\scriptstyle\sqrt{(\xi^2 - 1)(1 - \eta^2)}}\sin(\phi) \\
\xi\sqrt{\frac{1-\eta^2}{\xi^2-1}}\sin(\phi)	&-\eta\sqrt{\frac{\xi^2-1}{1-\eta^2}}\sin(\phi)
	&{\scriptstyle\sqrt{(\xi^2 - 1)(1 - \eta^2)}}\cos(\phi) \\
\eta	&\xi	&0
\end{pmatrix} .
\end{multline}
It is readily checked that this is an orthonormal coordinate system, so we have a diagonal metric and the scaling factors are the lengths of the basis vectors
\begin{subequations}
\begin{align}
\lambda_{\xi} &= \abs{\vece_{\xi}}
= \frac{R}{2}\sqrt{\frac{\xi^2 - \eta^2}{\xi^2 - 1}}, \\
\lambda_{\eta} &= \abs{\vece_{\eta}}
= \frac{R}{2}\sqrt{\frac{\xi^2 - \eta^2}{1 - \eta^2}}, \\
\lambda_{\phi} &= \abs{\vece_{\phi}}
= \frac{R}{2}\sqrt{(\xi^2 - 1)(1 - \eta^2)} .
\end{align}
\end{subequations}
The volume element is now simply obtained as the product of all scaling factors
\begin{align}
\ud\vecr = \lambda_{\xi}\lambda_{\eta}\lambda_{\phi} \, \ud\xi\ud\eta\ud\phi
= \frac{R^3}{8}(\xi^2 - \eta^2) \, \ud\xi\ud\eta\ud\phi
\end{align}
and the Laplacian becomes
\begin{align}
\nabla^2 &= \frac{1}{\lambda_{\xi}\lambda_{\eta}\lambda_{\phi}}
\left(\frac{\du}{\du\xi}\frac{\lambda_{\eta}\lambda_{\phi}}{\lambda_{\xi}}\frac{\du}{\du\xi} +
\frac{\du}{\du\eta}\frac{\lambda_{\xi}\lambda_{\phi}}{\lambda_{\eta}}\frac{\du}{\du\eta} +
\frac{\du}{\du\phi}\frac{\lambda_{\xi}\lambda_{\eta}}{\lambda_{\phi}}\frac{\du}{\du\phi}\right) \notag \\*
&= \frac{4}{R^2(\xi^2 - \eta^2)}\biggl[
\frac{\du}{\du\xi}(\xi^2 - 1)\frac{\du}{\du\xi} +
\frac{\du}{\du\eta}(1 - \eta^2)\frac{\du}{\du\eta} + {} \\*
&\eqspace\hphantom{\frac{4}{R^2(\xi^2 - \eta^2)}\biggl(}
\biggl(\frac{1}{\xi^2 - 1} + \frac{1}{1 - \eta^2}\biggr)\frac{\du^2}{\du\phi^2}\biggr] . \notag 
\end{align}

\chapter{Functionals and their derivatives}
\label{ap:functionals}

This explains gives a short introduction to functionals and their derivatives without going in details on the mathematical details like existence and all possible variants. The word functional is has a different meaning depending on the context, but we will define it to be a function which maps to scalars: either real or complex
\begin{equation}
F \colon X \to K ,
\end{equation}
where $X$ is some arbitrary set and $K$ is the scalar field either $\Reals$ or $\Complex$. For our purpose $X$ will be some function space, but for the introduction it is easier to use an $n$-dimensional vector space, which can either be real or complex $K^n$. By considering functions as vectors with a continuous index, we can easily deduce what the functional derivatives are for simple functionals of functions.

As a first example of a functional of a vector, we consider the functional which simply reports the $i$th component of a vector $\mat{v} \in K^n$
\begin{equation}
F_i[\mat{v}] = v_i .
\end{equation}
The derivative of the functional with respect to an arbitrary component of the vector is simply the well known partial derivative
\begin{equation}
\frac{\du F_i[\mat{v}]}{\du v_k} = \delta_{ik} .
\end{equation}
Now let us consider a functional which sums the vector components raised to some arbitrary power $\alpha$
\begin{equation}
F_{\alpha}[\mat{v}] = \sum_{i=1}^nv_i^{\alpha} .
\end{equation}
In that case the functional derivative becomes by the usual differentiation rules
\begin{equation}
\frac{\du F_{\alpha}[\mat{v}]}{\du v_k} = \alpha v_k^{\alpha-1} .
\end{equation}
You can easily check that this is consistent with the functional derivative of the previous functional $F_i$
\begin{align}
\frac{\du F_{\alpha}[\mat{v}]}{\du v_k}
&= \frac{\du}{\du v_k}\sum_{i=1}^n\bigl(F_i[\mat{v}]\bigr)^{\alpha}
= \sum_{i=1}^n \frac{\du}{\du v_k}\bigl(F_i[\mat{v}]\bigr)^{\alpha} \notag \\
&= \sum_{i=1}^n \alpha\bigl(F_i[\mat{v}]\bigr)^{\alpha-1}\frac{\du F_i[\mat{v}]}{\du v_k}
= \sum_{i=1}^n \alpha v_i^{\alpha-1} \delta_{ik}
= \alpha v_k^{\alpha-1} ,
\end{align}
where we used the chain-rule to go to the second line.

Now let us consider functionals of functions. The equivalent version of the functional $F_i$ for vectors would now be a functional that evaluates a function at a particular point $\vecx$
\begin{equation}
F_{\vecx}[f] = f(\vecx) .
\end{equation}
The functional derivative is now not so straightforward, but since the continuum version of the Kronecker delta is the Dirac delta function which is actually a distribution, we expect
\begin{equation}
\frac{\delta F_{\vecx}[f]}{\delta f(\vecx')} = \delta(\vecx - \vecx') = \delta(\vecr - \vecr')\delta_{\sigma,\sigma'} .
\end{equation}
We followed here the tradition to replace the $\du$ by $\delta$. Further, since the spin degree of freedom is represented by a finite vector space, this part of the Dirac delta distribution for our composite coordinate is just the usual Kronecker delta.

Let us check this by considering the function variant of the functional $F_{\alpha}$
\begin{equation}
F_{\alpha}[f] = \integ{\vecx}\bigl(f(\vecx)\bigr)^{\alpha} .
\end{equation}
Taking the functional derivative with respect to $f(\vecy)$ basically means that we consider the change of this functional when modifying the functional only at the point $\vecy$. So in analogy to the vector case, we obtain
\begin{equation}
\frac{\delta F_{\alpha}[f]}{\delta f(\vecy)}
= \alpha\bigl(f(\vecy)\bigr)^{\alpha-1} .
\end{equation}
Now we can check consistency between these derivatives in the same manner as before if we assume that the integral and derivative may be swapped
\begin{align}
\frac{\delta F_{\alpha}[f]}{\delta f(\vecy)}
&= \frac{\delta}{\delta f(\vecy)}\integ{\vecx}\bigl(F_{\vecx}[f]\bigr)^{\alpha}
= \integ{\vecx}\frac{\delta}{\delta f(\vecy)}\bigl(F_{\vecx}[f]\bigr)^{\alpha} \notag \\
&= \integ{\vecx}\alpha\bigl(F_{\vecx}[f]\bigr)^{\alpha-1}\frac{\delta F_{\vecx}[f]}{\delta f(\vecy)} \notag \\
&= \integ{\vecx}\alpha\bigl(f(x)\bigr)^{\alpha-1}\delta(\vecx - \vecy)
= \alpha\bigl(f(\vecy)\bigr)^{\alpha-1} .
\end{align}
Introducing the derivative of functionals of functions only via analogy to functionals of vectors may make you feel unsure, so let us also consider the actual definition of the derivative. In this case we wil use the definition of the Gâteaux derivative which is a generalization of the directional derivative. The Gâteaux derivative of the functional $F$ at the function $f(\vecx)$ in the `direction' $h(\vecx)$ is defined as
\begin{equation}
D_h F[f] = \lim_{\epsilon \to 0}\frac{F[f + \epsilon h] - F[f]}{\epsilon}
= \frac{\ud F[f + \epsilon h]}{\ud\epsilon}\biggr\rvert_{\epsilon = 0} .
\end{equation}
We can now apply this definition to the integral of the arbitrary power of $f(\vecx)$ to find
\begin{align}
D_h F_{\alpha}[f]
&= \lim_{\epsilon \to 0}\frac{1}{\epsilon}\integ{\vecx}
\Bigl[\bigl(f(\vecx) + \epsilon h(\vecx)\bigr)^{\alpha} - \bigl(f(\vecx)\bigr)^{\alpha} \Bigr] \notag \\
&= \lim_{\epsilon \to 0}\integ{\vecx} \alpha\bigl(f(\vecx)\bigr)^{\alpha-1}h(\vecx) + \Order(\epsilon) \notag \\
&= \integ{\vecx} \underbrace{\alpha\bigl(f(\vecx)\bigr)^{\alpha-1}}_{= \frac{\delta F_{\alpha}}{\delta f(\vecx)}}h(\vecx) ,
\end{align}
where we denoted the integral kernel $\alpha\bigl(f(\vecx)\bigr)^{\alpha-1}$ with $\delta F_{\alpha}/\delta f(\vecx)$ in analogy with the directional derivative for functionals of vectors when they can be written in terms of the gradient as \(D_{\mat{h}} F[\mat{v}] = \nabla_{\mat{v}} F \cdot \mat{h} \).

This also shows you that the expression \( \delta F / \delta f(\vecx) \) is not always valid as a total derivative by analogy to the function \( \abs{x} \) for example. The directional derivative exists at $x = 0$ so is Gâteaux differentiable, but its total derivative \(\ud\abs{x} / \ud x \) does not exist. So the notation $\delta F / \delta f(\vecx)$ does not always make sense for arbitrary functionals.

\section{Functionals containing derivative of functions}
In the course we typically circumvent problems with functionals depending on derivatives of functions by using partial integration. We can actually use the partial integration trick to deal with derivatives of functions already at the functional derivative level. Consider the following functional
\begin{equation}
L[f] = \integ{\vecx} \mathcal{L}\bigl(f(\vecx),\nabla f(\vecx)\bigr) ,
\end{equation}
where $\mathcal{L}(f,\nabla f)$ is an explicit and local expression the gradient and the function itself. Now consider the definition of the Gâteaux derivative
\begin{align}
D_h L[f]
&= \lim_{\epsilon \to 0}\frac{1}{\epsilon}\integ{\vecx}
\Bigl[\mathcal{L}\bigl(f(\vecx) + \epsilon h(\vecx),\nabla (f(\vecx) + \epsilon h(\vecx))\bigr) - 
\mathcal{L}\bigl(f(\vecx),\nabla f(\vecx)\bigr)\Bigr] \notag \\
&= \lim_{\epsilon \to 0}\integ{\vecx}\biggl[\frac{\du\mathcal{L}}{\du f}(\vecx)\,h(\vecx) + 
\frac{\du\mathcal{L}}{\du \nabla f}(\vecx)\cdot \nabla h(\vecx) + \Order(\epsilon)\biggr] \notag \\
&= \integ{\vecx}\biggl[\frac{\du\mathcal{L}}{\du f}(\vecx) - \nabla \cdot \frac{\du\mathcal{L}}{\du \nabla f}(\vecx) \biggr] h(\vecx) ,
\end{align}
where we assumed that the boundary terms vanish, e.g.\ Dirichlet boundary conditions for a finite volume or $L^2$ functions. Further, \( \frac{\du\mathcal{L}}{\du f}(\vecx) \) and \( \frac{\du\mathcal{L}}{\du \nabla f}(\vecx) \) mean that we take the derivative of $\mathcal{L}$ as if $f$ and $\nabla f$ respectively would be normal variables. Hence, we find
\begin{equation}
\frac{\delta L}{\delta f(\vecx)} = \frac{\du\mathcal{L}}{\du f}(\vecx) - \nabla \cdot \frac{\du\mathcal{L}}{\du \nabla f}(\vecx)
= \frac{\du L}{\du f(\vecx)} - \nabla \cdot \frac{\du L}{\du \nabla f(\vecx)} ,
\end{equation}
i.e.\ the famous Euler--Lagrange expression.

It is quite straightforward to generalize to higher order derivatives by using partial integration multiple times.

\chapter{Lagrange multipliers}
\label{ap:Lagrangian}
This appendix explains how Lagrange multipliers can be used to formulate the first order optimality conditions for an optimization problem under equality constraints, i.e.\ conditions that only involve first order derivatives. It is a slight modification of Appendix A from~\autocite{PhD-Giesbertz2010}, with some modifications to limit the discussion to equality constraints. The Lagrange multiplier technique can also be generalized to inequality constraints, which leads to the \acf{KKT} conditions~\autocite{Karush1939, KuhnTucker1951}.

First the problem has to be formulated more mathematically
\begin{align}
\label{eq:genKKTproblem}
&\min_{\vecr \in \Reals^N}&	&f(\vecr) \notag \\
&\text{subject to}&			&h_j(\vecr) = 0		&&\text{for $j=1,\dotsc,l$}, \notag 
\end{align}
where $f \colon \Reals^N \to \Reals$ and $h_j \colon \Reals^N \to \Reals$ are continuously differentiable. Strictly, it is not necessary for the domains of $f$ and $h_j$ to be $\Reals^N$, but a more general domain would require to specify some restrictions on it, which would only cloud the discussion. We will now illustrate the concept of Lagrange multipliers with a daily problem.

\begin{figure}[tbp]
\begin{tabular}{@{}l@{\hspace{0.04\textwidth}}r@{}}
\begin{minipage}[t]{0.48\textwidth}
  \centering
  \begin{tikzpicture}
    \draw [->] (0,0) -- (4,0);
    \draw (4,0) node [anchor = west] {$x$};
    \draw [->] (0,0) -- (0,3.5);
    \draw (0,3.5) node [anchor = south] {$y$};
    \draw (0,0) (0,-1.5);

    \draw (-0.5,2.85) cos (0,2.5) sin (1.5,1.5) cos (3,2.5) sin (4.5,3.5);
    \draw (0.8,2) node[anchor = west] {bar};
    \draw (2.8,2.2) node [anchor = west] {$h(x,y) = 0$};

    \filldraw (0,0) circle (2pt);
    \draw (0,0) node [anchor = north east] {$\mat{O}$};
    \filldraw (3,0) circle (2pt);
    \draw (3,0) node [anchor = north] {$\mat{F}$};
  \end{tikzpicture}
\end{minipage} &
\begin{minipage}[t]{0.48\textwidth}
  \centering
  \begin{tikzpicture}
    \draw [->] (0,0) -- (4,0);
    \draw (4,0) node [anchor = west] {$x$};
    \draw [->] (0,0) -- (0,3.5);
    \draw (0,3.5) node [anchor = south] {$y$};
    \draw (0,0) (0,-1.5);

    \draw (-0.5,2.85) cos (0,2.5) sin (1.5,1.5) cos (3,2.5) sin (4.5,3.5);
    \draw (0.8,2) node[anchor = west] {bar};
    \draw (2.8,2.2) node [anchor = west] {$h(x,y) = 0$};

    \filldraw (0,0) circle (2pt);
    \draw (0,0) node [anchor = north east] {$\mat{O}$};
    \filldraw (3,0) circle (2pt);
    \draw (3,0) node [anchor = north] {$\mat{F}$};
    
    \draw (1.5,0) ellipse (1.75 and 0.5);
    \draw (1.5,0) ellipse (2.0 and 1);
    \draw (1.5,0) ellipse (2.25 and 1.5);
    
    \filldraw (1.5,1.5) circle (2pt);
    \draw (1.5,1.5) node [anchor = north] {$\mat{B}$};
    \draw [dashed, thick] (0,0) -- (1.5,1.5);
    \draw [dashed, thick] (1.5,1.5) -- (3,0);
  \end{tikzpicture}
\end{minipage} \\
\begin{minipage}[t]{0.48\textwidth}
  \caption{Schematic representation of the problem. The student is located at $\mat{O}$ and wants to find the shortest path to his friend at point $\mat{F}$ via the bar.}
  \label{fig:problemScetch}
\end{minipage} &
\begin{minipage}[t]{0.48\textwidth}
  \caption{Graphical solution of the problem. Without constraining $\mat{B}$ to be on the bar, all points $\mat{B}$ that give the same total distance. The solution is the ellipse that has only one point in common with the bar.}
  \label{fig:graphSolution}
\end{minipage}
\end{tabular}
\end{figure}

A student goes to the pub to meet his%
\footnote{Every occurrence of he\slash{}his may be replaced by she\slash{}her, if the other gender is preferred.}
 friend. Just after he enters the pub he sees his friend sitting at table. However, he can not go there without a beer, so he needs to go to the bar first. He is a bit drunk already, so the bar does not seem to be completely straight. A schematic view of the situation is given in Fig.~\ref{fig:problemScetch}. The student enters at the origin, $O$, his friend is sitting at point $\mat{F}$ and the edge of the bar is described by the function $h(x,y) = 0$. The student has to find a point $\mat{B}$ on the edge of the bar [$h(\mat{B}) = 0$] such that the distance from $\mat{O}$ to $\mat{B}$, $d_{OB}$, plus the distance from $\mat{B}$ to $\mat{F}$, $d_{BF}$ is as short as possible. So we can introduce the following objective function
\begin{align}
f = d_{OB} + d_{BF}.
\end{align}
The problem can be solved graphically. Suppose we take the total distance to be some fixed value. If we plot all possible combinations of $d_{OB}$ and $d_{BF}$ which sum to this fixed value, we obtain an ellipse. By increasing the size of this ellipse till it just hits the bar ($h(x,y) = 0$), we find the optimal path that the student should take. In Fig.~\ref{fig:graphSolution} we show a couple of these ellipses. The outer ellipse is just large enough to touch the edge of the bar, so this is the point $\mat{B}$ that minimises the total distance $f$.

Note that at the point $\mat{B}$ that the ellipse is tangent to the the bar. In fact, this is not specific for this problem. For all optimization problems with equality constraints, the objective function will be tangent to the constraints. A more mathematical way to formulate this is to say that the normal vectors of both curves (surfaces in higher dimensions) are parallel. The normal vector of a curve or surface is given by the gradient, so this condition can be expressed mathematically as
\begin{align}\label{eq:equality_dual}
\nabla f(\mat{P}) + \lambda\nabla h(\mat{P}) = 0.
\end{align}
The unknown constant multiplier $\lambda$ is known as the Lagrange multiplier and it is necessary, because the magnitudes of the two gradients might be different.
\begin{figure}[tbp]
\begin{tabular}{@{}l@{\hspace{0.04\textwidth}}r@{}}
\begin{minipage}[t]{0.48\textwidth}
  \centering
  \begin{tikzpicture}
    \draw [->] (0,0) -- (4,0);
    \draw (4,0) node [anchor = west] {$x$};
    \draw [->] (0,0) -- (0,3.5);
    \draw (0,3.5) node [anchor = south] {$y$};

    \draw (-0.5,2.85) cos (0,2.5) sin (1.5,1.5) cos (3,2.5) sin (4.5,3.5);
    \draw (3.1,2.5) node [anchor = west] {$h = 0$};
    
    \draw (1.5,2) arc (90:140:2.5 and 1.75);
    \draw (1.5,2) arc (90:35:2.5 and 1.75) node [anchor=north] {$f = \text{const.}$};
    
    \filldraw (2.38,1.88) circle (2pt);
    \coordinate (C) at (2.38,1.89);
    \coordinate (F) at ($ (C) + (255:1.7) $);
    \coordinate (N) at (intersection of F--{$ (F) + (130:2) $} and C--{$ (C) + (220:2) $});
    \coordinate (H) at ($ (N) + (255:-1.7) $);
    \draw [-stealth, thick] (C) -- (F) node [anchor=west] {$-\nabla f$};
    \draw [-stealth, thick] (C) -- (H) node [anchor= south west] {$-\lambda\nabla h$};
    \draw [dashed, thick] (F) -- (N) -- (H);
    \draw [-stealth, thick] (C) -- (N) node [anchor = east] {$\mat{F}_{\text{net}}$};
 \end{tikzpicture}
\end{minipage} &
\begin{minipage}[t]{0.48\textwidth}
  \centering
  \begin{tikzpicture}
    \draw [->] (0,0) -- (4,0);
    \draw (4,0) node [anchor = west] {$x$};
    \draw [->] (0,0) -- (0,3.5);
    \draw (0,3.5) node [anchor = south] {$y$};

    \draw (-0.5,2.85) cos (0,2.5) sin (1.5,1.5) cos (3,2.5) sin (4.5,3.5);
    \draw (3.1,2.5) node [anchor = west] {$h = 0$};
    
    \draw (1.5,1.5) arc (90:140:2.25 and 1.5);
    \draw (1.5,1.5) arc (90:35:2.25 and 1.5) node [anchor=north] {$f = \text{const.}$};
    
    \filldraw (1.5,1.5) circle (2pt);
    \draw [-stealth, thick] (1.5,1.5) -- +(0,1.2) node [anchor=east] {$-\lambda\nabla h$};
    \draw [-stealth, thick] (1.5,1.5) -- +(0,-1.2) node [anchor=east] {$-\nabla f$};
  \end{tikzpicture}
\end{minipage} \\
\begin{minipage}[t]{0.48\textwidth}
  \caption{The force of the objective function $-\nabla f$ and the force of the constraint $-\lambda\nabla h$ are imbalanced. A net force $\mat{F}_{\text{net}}$ remains.}
  \label{fig:noBalance}
\end{minipage} &
\begin{minipage}[t]{0.48\textwidth}
  \caption{At the optimal point the forces of the constraint $-\lambda\nabla h$ exactly cancels the force of the objective function $-\nabla f$, so there is no net force.}
  \label{fig:balance}
\end{minipage}
\end{tabular}
\end{figure}
This expression, including the Lagrange multiplier has a nice physical interpretation. Consider a particle at $\vecr$ and $f$ to be its potential energy, so the force at $\vecr$ is given by $-\nabla f(\vecr)$. Since the constraint prevents the particle from going to the unconstraint minimum, it has to generate an opposing force, $-\lambda\nabla h(\vecr)$. If the forces are not parallel to each other, there will remain a net force pushing the particle to a region with a lower potential energy, without violating the constraint (Fig.~\ref{fig:noBalance}). At the point where the forces exactly cancel each other, $\vecr^*$, the particle is at a minimum of the potential energy $f$, satisfying the constraint $h(\vecr^*) = 0$. The Lagrange multiplier can be thought of a measure how hard $h(\vecr^*)$ has to pull in order to balance the force generated by $f$ (Fig.~\ref{fig:balance}).

Usually the use of Lagrange multipliers is formulated by introducing a Lagrangian. It is simply defined to be the objective function $f$, plus all the required equality constraints weighted by Lagrange multipliers
\begin{align}
L(\vecr,\mat{\lambda}) = f(\vecr) + \sum_{j=1}^l\lambda_jh_j(\vecr).
\end{align}
By taking the partial derivatives with respect to all coordinates (including $\mat{\lambda}$), all the optimality conditions are obtained
\begin{subequations}
\begin{align}
\nabla f(\vecr) &+ \sum_{j=1}^l\lambda_j\nabla h_j(\vecr) = \mat{0}, \\
h_j(\vecr) &= 0 \qquad \text{for $j = 1,\dotsc,l$}.
\end{align}
\end{subequations}

\begin{example}
The problem of the student can be solved if he is not too drunk, so the bar can be described by a simple straight line. However, due to the square roots in the objective function, the algebra is quite formidable, so it hardly serves as an example. Therefore, as a first example we consider a more simple mathematical problem
\begin{align}
\begin{split}
\min_{x,y \in \Reals} \qquad	f(x,y) &= x^2y \\
\text{subject to} \quad 		x^2 + y^2 &= 3.
\end{split}
\end{align}
As the constraint function we will take $h(x,y) = x^2+y^2-3$. The Lagrangian can be written as
\begin{align}
L(x,y,\lambda) = f(x,y) + \lambda h(x,y) = x^2y + \lambda(x^2+y^2-3).
\end{align}
Now we take all the partial derivates of the Lagrangian $L$
\begin{subequations}
\begin{align}
\label{eq:LagrExam_dLdx}
\frac{\du L}{\du x} &= 2xy + 2\lambda x = 0, \\*
\label{eq:LagrExam_dLdy}
\frac{\du L}{\du y} &= x^2 + 2\lambda y = 0, \\*
\label{eq:LagrExam_dLdlam}
\frac{\du L}{\du \lambda} &= x^2+y^2 - 3 = 0.
\end{align}
\end{subequations}
Equation~\eqref{eq:LagrExam_dLdx} gives $x=0$ or $y = -\lambda$. In the first case Eq.~\eqref{eq:LagrExam_dLdlam} gives $y = \pm\sqrt{3}$, so by Eq.~\eqref{eq:LagrExam_dLdy} we have $\lambda = 0$.

\begin{table}[bt]
\centering
\begin{tabular}{CCCC}
\toprule %
x		& y		& \lambda & f(x,y) \\
\midrule
-\sqrt{2}	&-1			& 1		&-2 \\
-\sqrt{2}	& 1			&-1		& 2 \\
0			& \sqrt{3}	& 0 		& 0 \\
0			&-\sqrt{3}	& 0 		& 0 \\
\sqrt{2}	&-1			& 1		&-2 \\
\sqrt{2}	& 1			&-1		& 2 \\
\bottomrule%
\end{tabular}
\caption{All six critical points of the function $f(x,y) = x^2y$ with $x$ and $y$ constraint to lay on a circle with radius $\sqrt{3}$.}
\label{tab:LagrExamCritPoints}
\end{table}

In the second case Eq.~\eqref{eq:LagrExam_dLdy} gives
\begin{align}
x^2 - 2y^2 = 0 \quad \Rightarrow \quad x^2 = 2y^2.
\end{align}
Using this result in Eq.~\eqref{eq:LagrExam_dLdlam} gives
\begin{align}
3y^2=3 \quad \Rightarrow \quad y = \pm 1.
\end{align}
So the Lagrange multiplier is $\lambda = \mp 1$ and $x = \pm\sqrt{2}$, where the sign of $x$ is arbitrary. So the Lagrange equations have six critical points which are summarised in table~\ref{tab:LagrExamCritPoints}. As can be seen from the table, the optimisation problem has two global solutions at $(-\sqrt{2},-1)$ and $(\sqrt{2},-1)$. From the Hessian of the Lagrangian $L$ it may be determined that $(0,\sqrt{3})$ is a local minimum.
\end{example}

\begin{example}
\label{ex:LagrangeStudent}
In this example we will work out the problem for the student when he is not too drunk, so the bar is still straight. In this case, the constraint function can in general be defined as
\begin{align}
h(x,y) = y - (\alpha x + h),
\end{align}
where $x$ and $y$ will be the coordinates of point $\mat{B}$. Its gradient is
\begin{align}
\nabla h(x,y) = (-\alpha, 1)^T.
\end{align}
An explicit expression for the objective function, i.e.\ the total length of the path can be written as
\begin{align}
f(x,y) = d_{OB} + d_{BF} = \sqrt{x^2+y^2} + \sqrt{(x-c)^2 + y^2},
\end{align}
where $c$ is the direct distance between the student and his friend ($c \isDefinedAs \abs{\mat{F} - \mat{O}}$). The square roots in this function might seem pretty harmless, but they will make the task of solving this problem quite formidable, even though we only treat a straight bar.

The derivatives of $f$ can be worked out as
\begin{subequations}
\begin{align}
\frac{\du f}{\du x} &= \frac{2x}{2\sqrt{x^2+y^2}} + \frac{2(x-c)}{2\sqrt{(x-c)^2+y^2}}
= \frac{x}{d_{OB}} + \frac{x-c}{d_{BF}}, \\
\frac{\du f}{\du y} &= \frac{2y}{2\sqrt{x^2+y^2}} + \frac{2y}{2\sqrt{(x-c)^2+y^2}}
= \frac{y}{d_{OB}} + \frac{y}{d_{BF}}.
\end{align}
\end{subequations}
Using these derivatives, the stationarity conditions become (at the point $\mat{B} = \mat{P}$)
\begin{subequations}
\begin{align}
\label{eq:dLdx}
\frac{\du L}{\du x} = \frac{\du f}{\du x} + \lambda\frac{\du g}{\du x}
= \frac{x}{d_{OP}} + \frac{x-c}{d_{PC}} - \alpha\lambda &= 0, \\
\label{eq:dLdy}
\frac{\du L}{\du y} = \frac{\du f}{\du y} + \lambda\frac{\du g}{\du y}
= \frac{y}{d_{OP}} + \frac{y}{d_{PC}} + \lambda &= 0, \\
\label{eq:stud_eq_cond}
\alpha x + h &= y.
\end{align}
\end{subequations}
Since the partial derivative of the Lagrangian $L$ with respect to $y$ [Eq.~\eqref{eq:dLdy}] directly gives an expression for the Lagrange multiplier $\lambda$ in terms of $x$ and $y$, the Lagrange multiplier in Eq.~\eqref{eq:dLdx} can be eliminated. In principle, using the equality condition [Eq.~\eqref{eq:stud_eq_cond}] to eliminate $y$, we obtain an equation with only the variable $x$. However, this equation is quite formidable to solve due to all the square root terms. So feel free to skip the algebra and to jump to the answer at the end of this section.

It is convenient only to substitute for $\lambda$ only and not yet for $y$. The equation needs to be reordered, so that all the terms containing a square root, $d_{OP}$ and $d_{PC}$, are isolated on one site of the equation
\begin{align}
&&
\biggl(\frac{1}{d_{OP}} + \frac{1}{d_{PC}}\biggr)x - \frac{c}{d_{PC}} &+ \alpha\biggl(\frac{1}{d_{OP}} + \frac{1}{d_{PC}}\biggr)y = 0 \notag \\
&\Leftrightarrow&
\biggl(\frac{1}{d_{OP}} + \frac{1}{d_{PC}}\biggr)&(x + \alpha y) = \frac{c}{d_{PC}} \notag \\
&\Leftrightarrow&
\biggl(\frac{d_{PC}}{d_{OP}} + 1\biggr)&(x + \alpha y) = c \notag \\
&\Leftrightarrow&
\frac{d_{PC}}{d_{OP}} &= \frac{c}{x + \alpha y} - 1 \notag \\
&\Leftrightarrow&
\sqrt{\frac{(x-c)^2+y^2}{x^2+y^2}} &= \frac{c - x - \alpha y}{x + \alpha y}.
\end{align}
Now the square root can simply be eliminated by taking the square on both sides. Later, we have to check our solution in the original equation, since the number of solutions of the squared equation is twice as large. Before eliminating $y$, the equation can be cleaned up a bit further
\begin{align}
&&
\frac{(x-c)^2+y^2}{x^2+y^2} &= \frac{(x + \alpha y - c)^2}{(x + \alpha y)^2} \notag \\
&\Leftrightarrow&
\bigl((x-c)^2 + y^2\bigr)(x + \alpha y)^2 &= \bigl(x^2 + y^2\bigr)(x + \alpha y - c)^2 \notag \\
&\Leftrightarrow&
\begin{split}
(x^2 + y^2)(x + \alpha y)^2 + (c^2 - 2cx)&(x + \alpha y)^2 \\
= (x^2 + y^2)&(x + \alpha y)^2 + (x^2 + y^2)\bigl(c^2 - 2c(x + \alpha y)\bigr)
\end{split} \notag \\
&\Leftrightarrow&
(c - 2x)(x + \alpha y)^2 &= (x^2 + y^2)\bigl(c - 2(x + \alpha y)\bigr).
\end{align}
To proceed, the equality condition $y=\alpha x + h$ has to be inserted. However, the equation become quite formidable, so we will deal with the left- and righthand side of the equation separately. 
\begin{subequations}
\begin{align}
\text{l.h.} &= (c - 2x)\bigl(x + \alpha(\alpha x + h)\bigr)^2 \notag \\*
&= (c-2x)\bigl((1+\alpha^2)x + \alpha h\bigr)^2 \notag \\*
&= (c-2x)\bigl((1+\alpha^2)^2x^2 + 2\alpha h(1+\alpha^2)x + \alpha^2h^2\bigr) \notag \\*
&= -2(1+\alpha^2)^2x^3 - 4\alpha h(1+\alpha^2)x^2 + c(1+\alpha^2)^2x^2 - {} \notag \\*
&\eqspace
2\alpha^2h^2x + 2\alpha ch(1+\alpha^2)x + \alpha^2ch^2, \\
\text{r.h.} &= \bigl(x^2 + (\alpha x + h)^2\bigr)\bigl(c - 2(x + \alpha(\alpha x + h))\bigr) \notag \\
&= \bigl((1+\alpha^2)x^2 + 2\alpha h x + h^2\bigr)\bigl(c - 2\alpha h - 2(1+\alpha^2)x\bigr) \notag \\
&= -2(1+\alpha^2)^2x^3 - 4\alpha h(1+\alpha^2)x^2 + (c-2\alpha h)(1+\alpha^2)x^2 - {} \notag \\
&\eqspace
2h^2(1+\alpha^2)x + 2\alpha h(c-2\alpha)x + h^2(c-2\alpha h).
\end{align}
\end{subequations}
Comparing the left- and righthand side, we see that the first two terms are equal. Therefore, we will only be left with a polynomial of order two. Subtracting the righthand side from the lefthand side gives
\begin{align}
0 &= \text{l.h.} - \text{r.h.} \notag \\
&= \bigl[(1+\alpha^2)\bigl(c(1+\alpha^2) - (c - 2\alpha h)\bigr)\bigr]x^2 +
\bigl[\alpha^2ch^2 - h^2(c - 2\alpha h)\bigr] + {} \notag \\
&\eqspace
2\bigl[c\alpha h(1+\alpha^2) - \alpha^2h^2 + h^2(1+\alpha^2) - \alpha h(c-2\alpha)\bigr]x \notag \\
&= (1+\alpha^2)(c + \alpha^2c - c + 2\alpha h)x^2 +
\bigl[\alpha^2ch^2 -ch^2 + 2\alpha h^3\bigr]\notag \\
&\eqspace
2(\alpha ch + \alpha^3ch - \alpha^2h^2 + h^2 + \alpha^2h^2 - \alpha ch + 2\alpha^2h)x \notag \\
&= \alpha(1+\alpha^2)(\alpha c + 2h)x^2 + 2h(\alpha^3c + 2\alpha^2h + h)x +
h^2(\alpha^2c + 2\alpha h - c).
\end{align}
In principle the equation is now easy to solve. However, the coefficients in the polynomial are rather cumbersome, so it is actually still a tough task. First consider the discriminant $D$.
\begin{align}
D &= h^2(\alpha^3c + 2\alpha^2h + h)^2 + h^2(c - \alpha^2c - 2\alpha h)\alpha(1+\alpha^2)(\alpha c + 2h) \notag \\
&= h^2\bigl[(\underline{\alpha^3c + 2\alpha^2h} + h)\{\alpha^3c + 2\alpha^2 h + h\} + {} \notag \\
&\eqspace\hphantom{h^2\bigl[}
\bigl(\alpha c - (\underline{\alpha^3c + 2\alpha^2h})\bigr)
\bigl(\{\alpha^3c + 2\alpha^2h + h\} + h + \alpha c\bigr)\bigr] \notag \\
&= h^2\bigl[h(\underline{\alpha^3c + 2\alpha^2 h} + h) +
\alpha c(\underline{\alpha^3c + 2\alpha^2h} + 2h + \alpha c) - {} \notag \\
&\eqspace\hphantom{h^2\bigl[}
(\underline{\alpha^3c + 2\alpha^2h})(h + \alpha c) \notag \\
&= h^2(h^2 + 2\alpha c h + \alpha^2c^2) \notag \\
&= h^2(\alpha c + h)^2.
\end{align}
With this result for the discriminant, the final solution becomes quite simple
\begin{align}
x = \frac{-h(\alpha^3c+2\alpha h+h) + h(\alpha c+h)}{\alpha(1+\alpha^2)(\alpha c + 2h)}
= h\frac{c - 2\alpha h - \alpha^2c}{(1+\alpha^2)(\alpha c+ 2h)}.
\end{align}
The $y$ coordinate of $\mat{B}$ is now simply found using the equality condition~\eqref{eq:stud_eq_cond}
\begin{align}
y &= h\biggl(\alpha\frac{c - 2\alpha h - \alpha^2c}{(1+\alpha^2)(\alpha c+ 2h)} + 1\biggr) \notag \\*
&= h\frac{\alpha c - 2\alpha^2h - \alpha^3c + \alpha c + 2h + \alpha^3c + 2\alpha^2h}
{(1+\alpha^2)(\alpha c+ 2h)} \notag \\*
&= \frac{2h(\alpha c + h)}{(1+\alpha^2)(\alpha c+ 2h)}.
\end{align}
The expression for the total distances $d_{OP}$ and $d_{PC}$ and the Lagrange multiplier are quite horrendous, so we do not show them.
\end{example}

\chapter*{Acronyms}

\begin{acronym}[meta-GGA]
\acro{1RDM}{one-body reduced density matrix}
\acro{2RDM}{two-body reduced density matrix}
\acro{a.u.}{atomic units\acroextra{: $\hbar = e = m_e = 4\pi\epsilon_0$}}
\acro{B3LYP}{Replacement of the \acs{PW91} correlation part by the \acs{LYP} functional in the \acs{B3PW91}~\cite{KimJordan1994, StephensDevlinChabalowski1994} and originally also the replacement of \acs{VWN}5 by \acs{VWN}3}
\acro{B3PW91}{A hybrid functional by Becke~\cite{Becke1993}}
\acro{B88}{\acs{GGA} approximation to exchange by Becke~\cite{Becke1988}}
\acro{BP86}{\acs{GGA} composed of \acs{B88} and the correlation functional by Perdew from 1986~\cite{Perdew1986}}
\acro{c-hole}{correlation hole}
\acro{CI}{configuration(s) interaction}
\acro{DFT}{density functional theory}
\acro{GEA}{gradient expansion approximation}
\acro{GGA}{generalized gradient approximation}
\acro{GTO}{Gaussian type orbital}
\acro{H}{Hartree\acroextra{classical Coulomb}}
\acro{HEG}{homogeneous electron gas}
\acro{HF}{Hartree--Fock}
\acro{HK}{Hohenberg--Kohn}
\acro{KKT}{Karush--Kuhn--Tucker}
\acro{HOMO}{highest occupied molecular orbital}
\acro{KS}{Kohn--Sham}
\acro{LB94}{The model \acs{xc} potential by Van Leeuwen and Baerends~\cite{LeeuwenBaerends1994}}
\acro{LDA}{local density approximation}
\acro{LUMO}{lowest unoccupied molecular orbital}
\acro{LYP}{\acs{GGA} correlation functional by Lee, Yang and Parr~\cite{LeeYangParr1988}}
\acro{meta-GGA}{A functional one step beyond the \acs{GGA}, so it also includes the Laplacian of the density}
\acro{MP}{Møller--Plesset}
\acro{OEP}{optimized effective potential}
\acro{PBE}{The \acs{GGA} by Perdew, Burke and Ernzerhof~\cite{PerdewBurkeErnzerhof1996}}
\acro{PW91}{The \acs{GGA} by Perdew and Wang~\cite{WangPerdew1991, PerdewChevaryVosko1992}}
\acro{RHF}{restricted \acs{HF}}
\acro{ROHF}{restricted open shell \acs{HF}}
\acro{SAOP}{The statistical averaging of (model) orbital potentials~\cite{GritsenkoSchipperBaerends1999, SchipperGritsenkoGisbergen2000}}
\acro{SCAN}{Strongly Constrained and Appropriately Normed~\cite{SunRuzsinszkyPerdew2015}}
\acro{SCE}{strictly correlated electrons~\cite{SeidlPerdewLevy1999, Seidl1999, SeidlGori-GiorgiSavin2007, Gori-GiorgiVignaleSeidl2009, SeidlGori-GiorgiSavin2007}}
\acro{SCF}{self-consistent field}
\acro{STO}{Slater type orbital}
\acro{TPSS}{The \acs{meta-GGA} by Tao, Perdew, Staroverov and Scuseria~\cite{TaoPerdewStaroverov2003}}
\acro{UHF}{unrestricted \acs*{HF}}
\acro{VWN}{Vosko, Wilk and Nusair}
\acro{WDA}{weighted density approximation}
\acro{x-hole}{exchange hole}
\acro{xc}{exchange-correlation}
\end{acronym}

\printbibliography

\printindex

\end{document}